\documentclass[11pt,a4paper,reqno]{amsart}
\usepackage{amssymb,amsmath,amsthm}

\usepackage[mathscr]{eucal} % mathcal style
\usepackage{graphicx}
\usepackage[hidelinks]{hyperref}
\usepackage{xcolor}
\hypersetup{
 colorlinks,
 linkcolor={red!50!black},
 citecolor={blue!50!black},
 urlcolor={blue!80!black}
}
\usepackage{pdfsync}
\usepackage{enumerate} % Fancy enumeration lists 
\usepackage{esint}

\theoremstyle{plain}
\newtheorem{theorem}{Theorem}[section]

\newtheorem{lemma}[theorem]{Lemma}

\newtheorem{proposition}[theorem]{Proposition}
\theoremstyle{remark}
\newtheorem{remark}[theorem]{Remark}

\numberwithin{equation}{section}
\newcommand{\Z}{\mathbb{Z}}

\usepackage{dsfont}

\allowdisplaybreaks

\newcommand{\R}{\mathbb{R}}
\newcommand{\C}{\mathbb{C}}
\newcommand{\T}{\mathbb{T}}
\newcommand{\D}{\mathcal{D}}
\newcommand{\ic}{\mathrm{i}}

\renewcommand{\d}{\, {\rm d }}
\renewcommand{\div}{\operatorname{div}}
\newcommand{\curl}{\operatorname{curl}}

\renewcommand{\leq}{\leqslant}

\renewcommand{\epsilon}{\varepsilon}
\newcommand\Real{\operatorname{Re}} 
\newcommand\Imag{\operatorname{Im}} 

\def\cal#1{\mathcal{#1}}
\def\mb#1{{\bf {#1}}}

\date\today

\title[Inviscid Water-Waves]{Inviscid Water-Waves and interface modeling}

\author[E. Dormy]{Emmanuel Dormy}

\author[C. Lacave]{Christophe Lacave}

\address[E. Dormy]{D\'epartement de Math\'ematiques et Applications, UMR-8553, \'Ecole Normale Sup\'erieure, CNRS, PSL University, 75005 Paris, France}
\email{Emmanuel.Dormy@ens.fr}

\address[C. Lacave]{Univ. Savoie Mont Blanc, CNRS, LAMA, 73000 Chamb\'ery,
  France.\footnote{Previously at Univ. Grenoble Alpes, CNRS, IF, 38000 Grenoble, France}}
\email{Christophe.Lacave@univ-smb.fr}

\begin{document}

\maketitle

\begin{abstract} 
We present a rigorous mathematical analysis of the modeling of inviscid water waves.
The free-surface is described as a parametrized curve.
We introduce a numerically stable algorithm which accounts for its evolution with time.
The method is shown to converge using approximate solutions, such as Stokes waves and Green-Naghdi
solitary waves. It is finally tested on a wave breaking problem, for which
an odd-even coupling suffices to achieve numerical convergence up to the
splash without the need for additional filtering.
\end{abstract}

\tableofcontents 

\section{Introduction}\label{sect-intro}

The study of water waves has a long mathematical history (Airy, Boussinesq,
Cauchy, Kelvin, Laplace, Navier, Rayleigh, Saint-Venant, Stokes, to cite
only a few). It has been studied in a variety of situations, probably the most complex
of these being the wave breaking problem. What happens when a waves overturns
raises significant mathematical difficulties. The water-air interface cannot be
described as a graph any longer. A parametric description of the interface
and the tracking of its Lagrangian evolution are needed.

We want to derive here a stable numerical strategy to solve for
one-dimensional water waves (i.e. in a 2D domain, or a 3D domain assuming
independence in one horizontal coordinate of space). We numerically
approximate the free-surface Euler equations 
without introducing artificial regularizing parameters. This is
particularly important in the case of loss of regularity of the interface,
in order to study the possible formation of singularity
(e.g. \cite{Baker11}).

\subsection{Problem formulation}
We consider a simple periodic domain
$\mathcal{D}= \T_{L} \times\R$ see Fig.~\ref{Geometry}, and introduce two
boundaries, $\Gamma_{S}$ the free surface water-air, and $\Gamma_{B}$ the bottom.
The domain is thus decomposed in three subdomains,
$\mathcal{D}_{F}$ the fluid domain, $\mathcal{D}_{A}$ the air domain,
$\mathcal{D}_{B}$ below the bottom.

\begin{figure}
 \centerline{\includegraphics[height=0.25\textwidth]{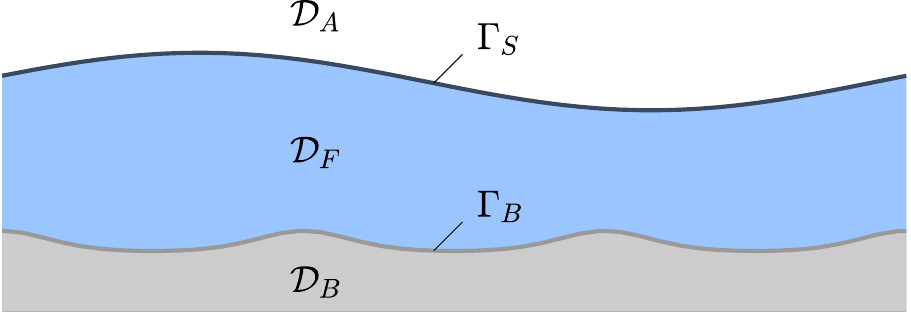}}
 \caption{The domain being considered consists of
 $\mathcal{D}= \T_{L} \times\R \, ,$
 which is decomposed in
 $\mathcal{D}=\mathcal{D}_{F}\cup \mathcal{D}_{A}\cup \mathcal{D}_{B}\cup \Gamma_{S}\cup \Gamma_{B}\, ,$
 where $\Gamma_{S} = \overline{\mathcal{D}_{A}} \cap
 \overline{\mathcal{D}_{F}}\, ,$ and $\Gamma_{B} = \overline{\mathcal{D}_{F}} \cap
 \overline{\mathcal{D}_{B}}\, .$}
\label{Geometry}
\end{figure}

Since we want our mathematical approximation to be able to describe
an interface which is not a graph (i.e. overturning of water in the context of a
breaking wave) we need to be able to track it as a parametrized
curve. Indeed, a description in the form $y=h(x)$ would develop a shock
(discontinuity) as soon as the water starts to overturn. 

Two cases will be considered, the single fluid problem, in which the air
density is neglected, and the bi-fluid problem. In the latter case, the
Euler equation needs to be considered both in the water 
($\mathcal{D}_{F}$) and in the air ($\mathcal{D}_{A}$).
At the water-bottom interface ($\Gamma_{B}$) the normal component of
velocity needs to vanish (impermeability condition).
Whereas at the
water-air interface ($\Gamma_{S}$) two quantities need to be continuous
across the interface: the normal velocity and the pressure. The latter, for
an inviscid fluid, being equivalent to the continuity of the normal component of the stress tensor.

The velocity tangential to the interface is notably not continuous across
$\Gamma_{S}$.
This results in a localized distribution of vorticity along $\Gamma_{S}$ in
the form of a vortex sheet.

The free surface $\Gamma_{S}$ is initially parametrized by arclength from
left to right: $e\in [0,L_{S}] \mapsto z_{S}(e)=z_{S,1}(e)+ \ic z_{S,2}(e) \, .$ 
In the same way, the bottom $\Gamma_{B}$ is parametrized by $e\in [0,L_{B}] \mapsto z_{B}(e)=z_{B,1}(e)+ \ic
z_{B,2}(e)\, .$ 

It is useful to introduce the tangent vector $\tau_{S} = \tau_{S,1} +\ic
\tau_{S,2}= |z_{S,e}(e)|^{-1} z_{S,e}\, , $ which is pointing to the right
and the normal $n_{S}=n_{S,1} +\ic n_{S,2}= -\tau_{S,2} +\ic \tau_{S,1}=\ic \tau_{S}$ is
pointing out of the fluid domain (where $z_{S,e}$ is the derivative in $e$).
The same is done at the bottom with the normal now pointing in the fluid domain.

It should be stressed that the arclength is not preserved as the fluid
surface $\Gamma_{S}$ evolves, it is thus important to consider $|z_{e}|(t,e)\,.$

We introduce on any vector field $\mb u = ( u_{1}, u_{2})$ the following three operations
$\widehat{\mb u}=u_{1}-\ic u_{2} \, ,$
$(u_{1},u_{2})^\perp =(-u_{2},u_{1}) \, ,$
and the curl operator
$\curl \mb u = \partial_{1} u_{2}-\partial_{2} u_{1}\, .$ 
Finally, we also introduce for any vector $\mb x = ( x_{1}, x_{2})$
the complex notation $x=x_{1}+\ic x_{2} \, .$

\subsection{Numerical strategies}
Interface evolution methods aim at capturing the time evolution of the
water-air interface ($\Gamma_{S}$) using solely the knowledge of this
vorticity distribution. This results in having to cope with singular
integrals along the vortex sheet, but simplifies the problem numerically
in that neither the water domain ($\mathcal{D}_{F}$), nor the air domain
($\mathcal{D}_{A}$), need to be meshed, as would be the case for example
using a Finite-Element Method. Such approaches are closer in spirit to a
Boundary-Element Method used in three-dimensional problems
(e.g. \cite{Grilli89,guyenne_grilli_2006,Pomeau_2008,Pomeau}).

A classical formulation of water waves is the celebrated Zakharov-Craig-Sulem description \cite{Zakharov68,CraigSulemSulem}.
While this description has proven extremely useful from a theoretical point
of view (e.g. \cite{Alazard,Lannes05,Germain12}), it raises difficulties from a
numerical point of view and its adaptation to the framework of singular integral formulations 
could be the subject of a future work (see Remark~\ref{rem:D2N} for a discussion).
 
We will consider instead two different approaches to compute the free-surface
evolution. The first one is based on a so-called `vortex method', the
discontinuity of the tangential velocity at the air-water interface is
modelled by a vortex sheet (with vorticity distribution $\gamma$).
In this approach, a finite but large number of localized vortices is used to approximate a continuous vortex
sheet. Such an approach has been shown to be efficient for Euler flows in
the full space \cite{GHL90} and for exterior domains \cite{ADL}.
The interface is
transported by the resulting normal flow (as expressed in \eqref{eq:dtz1}).
In the case of free surface flows, such as water waves, the vorticity
distribution along the $\Gamma_{S}$ surface is however not
preserved following a Lagrangian trajectory, and its evolution in time needs to
be traced in the system describing the evolution of two Euler flows with
two continuity conditions (see Equation \eqref{inv_gammaSt}).
This approach has been pioneered in \cite{Baker82}
and further developed recently in \cite{Wilkening21}. It is
also a useful description for mathematical proofs
(e.g. \cite{Wu97,Castro}).

The second approach is referred to as the `dipole layer', where the velocity is
now described by the jump in potential between the two fluids (measured by
the dipole distribution $\mu$, related to $\gamma$ via $\gamma =
\partial_e \mu $). The jump in tangential velocity then stems from
expressing the velocity as a potential. 
The dipole layer evolution follows from the Bernoulli equation 
(and takes the form given in \eqref{inv_muSt}), whereas the interface is again
transported by the normal flow (as expressed in \eqref{eq:dtz2}).
This approach is known
as the `dipole method' (it is equivalent to double layer potentials in
potential theory). This approach was first investigated in \cite{Baker82}
and then in \cite{BAKER198353,Baker11}.

The former method is lighter to derive and offers the possibility to account for Euler flows with
vorticity, whereas the latter involves more analytical work and assumes an irrotational flow. We will see
however that the latter has better convergence properties for strongly
non-linear configurations.
In both cases, the spatial discretization in terms of singular integrals is
known to converge toward the Euler equation \cite{ADL}, see also
\cite{Beale96} for a study of the vortex method in the deep-water case. 

The deep-water case is a trivial limit of the above description, in which
the bottom vortex sheet (distributed on $\Gamma_{B}$) is sent to infinity. It is formally
equivalent to simply suppressing it, or setting its vorticity to zero.

\subsection{Objectives of the present study}
Our aim is to construct a numerical scheme based on solid
mathematical developments and free of smoothing or regularizations and
which can later be used to guide theoretical understanding and
further mathematical constructions on these problems. 
The goal of this article is thus to derive a formulation of these methods in the
most general case (bi-fluid or single fluid, including possibly a non-flat
bottom, including vorticity and mean currents).

The discrete approximation will be based on a relevant reformulation of the
continuous problem. The first step is to introduce some quantities defined
on the 1D interface, which then allows to recover the full fluid motion.

\begin{theorem}\label{theo:main-ellip}
Given $\Gamma_{S},\Gamma_{B}\in C^{1,1}$, the following holds true.

For any $\mb v\in L^2(\Omega_{F})$ such that $\div \mb v=0\,$, $\curl \mb
v\in L^\infty(\Omega_{F})\,$, $\mb v\cdot n\vert_{\Gamma_{B}}=0$, there
exists $(\gamma_{S},\gamma_{B}) \in C^{0,\alpha}$ and $(\mu_{S},\mu_{B})\in
C^{1,\alpha}$ (for any $\alpha\in (0,1)$) such that $\mb v$ can be
recovered uniquely in terms of $(\gamma_{S},\gamma_{B},\int_{\Gamma_{B}}\mb
v\cdot\tau, \curl v)$ or in terms of $(\mu_{S},\mu_{B},\int_{\Gamma_{B}}\mb
v\cdot\tau, \curl v)$. 
\end{theorem}

This theorem will offer a base to turn a 2D problem into a 1D problem. It
involves singular integral representation. This allows us to describe the
free surface problem knowing only $(\gamma_{S},\gamma_{B})$ in the case of the vortex
formulation, or $(\mu_{S},\mu_{B}$ in the case of the dipole formulation. 

\begin{remark}\label{rem:main-ellip}
 It will be showed in a second step that $\gamma_{B}$ (resp. $\mu_{B}$) can
 be determined uniquely in terms of $\gamma_{S}$ (resp. $\gamma_{B}$) and
 $(\int_{\Gamma_{B}}\mb v\cdot\tau, \curl v)$. 
\end{remark}

The next step will then be to reformulate the water-waves equations into an evolution equation for $\gamma_{S}$ or $\mu_{S}$ and $z_S$.
 
\begin{theorem}\label{theo:main-dyn}
 Let $z_{S}$ and $\mb v$ be a regular solution of the water-waves equations, then the following holds:
\begin{enumerate}
 \item The single fluid equation without vorticity but with mean current and the bi-fluid model without vorticity and without mean current can be reformulated as explicit equations on $(\partial_{t}z_{S},\partial_{t}\mu_{S})$ which involves only $(z_{S}, \mu_{S}, \int_{\Gamma_{B}}\mb v_{0}\cdot\tau)$: see \eqref{eq:dtz2} and \eqref{inv_muSt}.
 \item The single fluid and bi-fluid equations with vorticity and with mean current can be reformulated as explicit equations on $(\partial_{t}z_{S},\partial_{t}\gamma_{S})$ which involves only $(z_{S}, \gamma_{S}, \int_{\Gamma_{B}}\mb v_{0}\cdot\tau, \curl \mb v)$: see \eqref{eq:dtz1} and \eqref{inv_gammaSt}.\end{enumerate}
\end{theorem}

The first item (1) corresponds to the dipole formulation, whereas the
second (2) corresponds to the vortex
formulation. Of course, to be complete, we need to add the transport
equation for $\curl \mb v$ by the velocity given in terms of
$(\gamma_{S},\int_{\Gamma_{B}}\mb v_{0}\cdot\tau, \curl v)$. It could
appear strange to consider the dipole formulation which is less general
than the vortex formulation, however we will see that it gives more stable numerical
schemes. Surface tension is included in the derivation, but can also be omitted. 

Both theorems above do not introduce any regularization and hold true for the continuous problem.
The resulting expressions are thus non-trivial and, in a first reading, we
advise to drop all terms associated to the density of air (i.e. associated
to the bi-fluid formulation), uniform or localized vorticity and circulation (mean
currents).

We then want to ensure
that our numerical scheme converges in a realistic manner (i.e. for
realistic parameters that can be achieved in practice) to the solution of the
continuous problem. We introduce in this work a regularization-free
approach to solve for the water-wave problem
(i.e. without explicit filtering or any other regularization introducing
extra parameters to the problem).
To discretise singular integrals, we follow an approach introduced in
\cite{ADL} for which rigorous convergence proofs were provided in the case
of a smooth boundary ($C^\infty$).

We verify conserved quantities at the
discrete level. We illustrate on simple test cases the
numerical convergence to the
approximate solutions (e.g. Stokes waves, Green-Naghdi solitary waves). We
also demonstrate stability and convergence of our numerical solution for
the wave-breaking problem. 

Finally, we investigate the effects of regularization strategies on the 
solution and illustrate numerically how they can yield irrelevant solutions.

\subsection{Plan of the paper}

In the next section, we introduce the singular integral representation,
thus proving Theorem~\ref{theo:main-ellip} by solving the corresponding
elliptic equations.

In Section~\ref{sec3}, we first prove Remark~\ref{rem:main-ellip} and then reformulate the water-waves equations using this
formalism, hence proving Theorem~\ref{theo:main-dyn}. 

In these two sections, we have chosen to present a general derivation, using a minimal amount of
simplifying assumptions. Simplifying assumptions are thus introduced as the
derivation proceeds, only when they become necessary. We believe this
highlights why and where each hypothesis is needed. This justifies in
particular the additional restrictions associated with the dipole formulation.

Section~\ref{sec4} presents our discretization strategy. Numerical
results as well as convergence tests are presented in Section~\ref{Sect_Num}. The
comparison with previously used regularization strategies (filtering and
offsetting) is performed in Section~\ref{sec6}. Finally in Section~\ref{sec7} we discuss
potential applications and further development.

Finally, the Plemelj formulae, discrete expressions both for the vortex and the
dipole method, as well as a list of notations are presented in appendices.

\section{Singular integrals representation at fixed time} 

In this section we will establish Theorem~\ref{theo:main-ellip} by studying
the elliptic problem related to inviscid flows.

\subsection{Stream and potential functions}\label{sec-elliptic}

The aim of this subsection is to express the velocity in both fluids in
terms of stream-function and velocity potential. This involves the
resolution of a div-curl problem in terms of the vorticity and the circulation.

\subsubsection{Resolution of the Laplace problem in $ \T_{L} \times\R$}

Even if we assume
that the fluid is curl-free in $\cal D_{F}\,$, the non-trivial boundary condition
on $\Gamma_{S}$ will be interpreted as a vortex sheet in $\cal D =\T_L\times \R\,$. For this
reason, we introduce the Green kernel in $\cal D$
\begin{equation} 
G(\mb x)=\frac1{4\pi} \ln\Big(\cosh \frac{2\pi x_{2}}{L}-\cos \frac{2\pi
 x_{1}}{L}\Big), \ G(\mb x,\mb y)=G(\mb x-\mb y) \,,
\label{Greenkernel}
\end{equation} 
and we recall the following result proved in \cite{BeichmanDenisov}:

\begin{proposition}\label{prop-Green}
 For any $f\in L^\infty_{c}(\D)\,$, every solution $\Psi$ of the following elliptic problem
 \[
 \Delta \Psi = f\,,\quad \lim_{x_{2}\to +\infty} \partial_{2}\Psi = -\lim_{x_{2}\to -\infty} \partial_{2}\Psi\,, \quad |\Psi |\leq C_1 (|x_{2}|+1)
 \]
 can be written as
 \begin{equation}
 \Psi(\mb x)= \Psi[f](\mb x)=\int_{\D} G(\mb x,\mb y) f(\mb y)\d \mb y + C_2 \quad
 \text{where\ } C_1 \text{\ and\ } C_2 \text{\ are constants.}
 \label{Green_formulation}
 \end{equation} 
\end{proposition}
The relation between the Green kernel \eqref{Greenkernel} in $\T_{L} \times\R$ and the usual kernel $\frac1{2\pi} \ln |\mb x-\mb y|$ in $\R^2$ is formally derived in Appendix~\ref{app-cot}.

From the explicit formula, it is easy to observe using Taylor expansions that
\begin{equation}\label{behinf}
\begin{aligned}
 \int_{\D} G(\mb x,\mb y)& f(\mb y)\d \mb y = \\
 & \Big( \frac{x_{2}}{2L}-\frac{\ln 2}{4\pi} \Big) \int_{\D}f(\mb y)\d \mb
 y -\frac1{2L} \int_{\D} y_{2}f(\mb y)\d \mb y + \mathcal{O}(e^{-|x_{2}|})
 \text{ as } x_2 \to +\infty\,,\\
 \int_{\D} G(\mb x,\mb y) &f(\mb y)\d \mb y = \\
 &\Big(- \frac{x_{2}}{2L}-\frac{\ln 2}{4\pi} \Big) \int_{\D}f(\mb y)\d \mb
 y +\frac1{2L} \int_{\D} y_{2}f(\mb y)\d \mb y + \mathcal{O}(e^{-|x_{2}|})
 \text{ as } x_2 \to -\infty \,.
\end{aligned}
\end{equation}
Thus $\lim_{x_{2}\to +\infty} \partial_{2}\Psi = -\lim_{x_{2}\to -\infty} \partial_{2}\Psi$ is a necessary and sufficient condition to use the
representation formula \eqref{Green_formulation}.

\subsubsection{Stream function and potential construction in the fluid domain}
We now apply this elliptic formalism to define the
stream function for the vortex method. Let us consider the following elliptic problem on $\D_{F}$ : for any functions
$g$ with zero mean and $\omega\,$, we want to analyze a vector field $\mb u$
such that 
\begin{equation}\label{ellip0}
\div \mb u=0 \text{ in } \D_{F}\,, \quad \curl \mb u =\omega \text{ in } \D_{F}\,,
\quad \mb u\cdot\mb n=0 \text{ on } \Gamma_{B}\,, \quad \mb u\cdot\mb n=g
\text{ on } \Gamma_{S} \,,
\end{equation}
which we want to extend in $\D\,,$ in order to be able to use the above
proposition.

In \eqref{ellip0}, the divergence free assumption stems from the incompressibility property
and the third condition corresponds to the impermeability of the boundary
at the bottom. By the Stokes formula, the fact $g$ has a zero mean is a necessary condition coming from these two
assumptions.

Unfortunately, \eqref{ellip0} has infinitely many solutions because of the harmonic vector field, also called the constant background current in \cite{Moon}, $\mb H$:
\[
\div \mb H=\curl \mb H=0 \text{ in } \D_{F}\,, \quad \mb H\cdot \mb n=0 \text{ on } \Gamma_{B}\cup \Gamma_{S}
\]
for instance when $\Gamma_{B}= \T_{L}\times \{-1\}$ and $
\Gamma_{S}=\T_{L}\times \{0\}\,$, we have $\mb H=\mb e_{1}
\mathds{1}_{\D_{F}} \,. $

In order to uniquely determine $\mb u$ from $\omega\,$, we need to prescribe
the circulation either on the bottom or below the free surface, knowing that we
have the following compatibility condition from the Stokes formula 
\begin{equation}\label{Stokes-comp}
 \int_{\Gamma_{B}}\mb u\cdot \mb \tau \d \sigma - \int_{\Gamma_{S}} \mb
 u\cdot \mb \tau \d \sigma = \int_{\D_{F}} \omega\, \d \mb x
\end{equation}
where the integrals are taken from left to right. 

\begin{lemma}\label{lem-ellip}
 For any $\omega\in L^\infty(\D_{F})$ , $g\in C^0(\Gamma_{S})$ with zero mean value, and $\gamma\in \R$ given, there exists a unique $\mb u\in H^1(\D_{F})$ such that 
\begin{equation}\label{ellip1}
\begin{aligned}
 &\div \mb u=0 \text{ in } \D_{F}\,, \quad \curl \mb u =\omega \text{ in } \D_{F}\,, \quad \mb u\cdot \mb n=0 \text{ on } \Gamma_{B}\,, \\
 &\mb u\cdot \mb n=g \text{ on } \Gamma_{S}\,, \quad \int_{\Gamma_{B}} \mb u\cdot \mb \tau \d \sigma= \gamma\,.
\end{aligned}
\end{equation}
Moreover, there exists a unique $\psi_{F} \in H^2(\D_{F})$ up to an
arbitrary constant, such that
\begin{equation*}
 \mb u=\nabla^\perp \psi_{F}\, .
\end{equation*}
\label{stream_vortex}
\end{lemma} 

\begin{proof}
 The proof of this lemma comes from standard ideas in elliptic theory and we give here only the main lines.
 
 The existence comes from the existence of $\phi\in H^1(\D_{F})$ solving the following Laplace problem with Neumann boundary condition (where the zero mean assumption is needed):
 \begin{equation*}
 \Delta \phi=0 \text{ in } \D_{F}\,, \quad \partial_{n}\phi =-(\nabla^\perp \Psi+\alpha \mb e_{1})\cdot \mb n \text{ on } \Gamma_{B}\,, \quad
 \partial_{n}\phi = g-(\nabla^\perp \Psi+\alpha \mb e_{1})\cdot \mb n \text{ on } \Gamma_{S}\
\end{equation*}
where $\Psi = \Psi[\omega]$ is defined in \eqref{Green_formulation} and $\alpha\in \R$ such that $\int_{\Gamma_{B}} (\nabla^\perp \Psi+\alpha \mb e_{1}) \cdot \mb \tau \d \sigma=\gamma$. Indeed, it is then enough to set $\mb u = \nabla \phi + \nabla^\perp \Psi + \alpha \mb e_{1}$.

Having such a $\mb u$, it is always possible to construct the stream function $\psi_{F}\,,$ because $\div \mb u=0$ and $\int_{\Gamma_{S}}\mb u\cdot \mb n = \int_{\Gamma_{B}}\mb u\cdot \mb n =0$
imply that
$\int_{\Gamma}\mb u^\perp\cdot\mb \tau =0$ for any
closed loop $\Gamma\, $, which allows us to construct $\psi_{F}\,$, uniquely up to a
constant.

The uniqueness can be deduced from the stream function. Indeed, denoting with tilde the difference of two solutions, we note that $\partial_{\tau}\tilde\psi_{F}=\tilde{\mb u} \cdot \mb n=0$ implies that $\tilde\psi_{F}$ is constant on each component of the boundary and we conclude by integrating by parts:
\[
\int_{\D_{F}}|\tilde{\mb u}|^2 = \tilde\psi_{F}\vert_{\Gamma_{B}}\int_{\Gamma_{B}} \tilde{\mb u}\cdot \tau\d s- \tilde\psi_{F}\vert_{\Gamma_{S}}\int_{\Gamma_{S}} \tilde{\mb u}\cdot \tau\d s =0.
\] 
\end{proof}

We should note that the conservation laws for the 2D Euler equations
(including the circulation and the total vorticity) imply that 
$\int_{\Gamma_{B}}\mb u\cdot \mb \tau$ and $\int_{\Gamma_{S}} \mb
u\cdot \mb \tau $ are both conserved quantities.

In the case of the dipole formulation, we need to write $\mb u$ as a gradient, which
is possible only by subtracting the curl and the circulation parts. Of
course, we could take advantage of the fact
$\mb u - \nabla^\perp \Psi[\omega]-\frac{\gamma}L e_{1}$
is curl free with zero circulation, and can thus be written as a
gradient in $\cal D_{F}\,$. Nevertheless this approach introduces additional
difficulties.

For example, if we take into account the density of air, it will be crucial to
properly define the air velocity field. However, for the single-fluid water-waves
equations (in which the density of air is neglected), we are left with
several possible choices, in particular stationary vector-fields could be
used, thus simplifying the computation below.
To underline where the properties of the vector fields are
important, we stay general for now and we will introduce constraints as they
become necessary.

Let us consider any $\mb u_{\omega,\gamma}$ such that
\begin{equation}\label{u-fgamma-1}
\div \mb u_{\omega,\gamma}=0 \text{ in } \D_{F}\,, \ \curl \mb
u_{\omega,\gamma} =\omega \text{ in } \D_{F}\,, \ \int_{\Gamma_{B}} \mb
u_{\omega,\gamma} \cdot \mb n \d s= 0\,, \ \int_{\Gamma_{B}} \mb
u_{\omega,\gamma} \cdot \mb \tau \d s= \gamma \,.
\end{equation}

Finding $\mb u$ the solution of \eqref{ellip1} is equivalent to look for $\mb u_{R}:=\mb u-\mb u_{\omega,\gamma}$ which is div and curl free,
without circulation and flux. The existence and uniqueness of $\mb u_{R}$ comes from Lemma~\ref{lem-ellip}. It can be written as the perpendicular gradient of a stream function, but also as the gradient of a potential function:
\begin{equation}\label{uR-pot}
 \mb u_{R} = \nabla \phi_{F} = \nabla^\perp \tilde\psi_{F}
\end{equation}
where $\phi_{F}$ and $\tilde\psi_{F}$ are uniquely determined, up to a
constant. Here we use that $\mb u_{R}$ is curl free with zero circulation to state that $\int_{\Gamma}\mb u\cdot\mb \tau =0$ for any
closed loop $\Gamma\, $. Even if it may not seem natural to study $\tilde\psi_{F}$ instead of
$\psi_{F}\,$, we will see below that $\tilde\psi_{F}$ is an interesting
quantity to consider for the dipole formulation.

\subsubsection{Extension to the full domain}
\label{firststep}
Now that we have established Lemma~\ref{stream_vortex} and
 \eqref{uR-pot}, and in order to
apply Proposition~\ref{prop-Green} to obtain a representation
formula, we first need to extend continuously the potential $\phi_{F}$ or the
stream functions $\psi_{F}$ or $\tilde\psi_{F}\,$.
Extending the potential is related to 
the fluid charge method developed in \cite{ADL}. This method is
unfortunately not relevant for a free surface problem, see Remark~\ref{rem-FCM}.

We, therefore, prefer to extend the stream functions continuously. This is
equivalent to assuming the continuity of the normal part of the velocity
across the boundary. Such an extended vector field is divergence free in
the whole domain $\D\,$, hence can be written using a stream function,
and the boundaries can be interpreted as vortex sheets, 
corresponding to the jump in the tangential velocity.

At the bottom, we extend $\mb u$ in the simplest possible way, i.e. such that
\[
\div \mb u = \curl \mb u=0 \text{ in } \D_{B}, \quad \mb u\cdot \mb n= 0 \text{ on } \Gamma_{B}\,, \quad \int_{\Gamma_{B}} \mb
u \cdot \mb \tau \d s= 0 \,,
\]
which implies $\mb u\vert_{\D_{B}}=0 \,.$

This is equivalent to extending $\psi$ by the constant
$\psi_{F}\vert_{\partial \Gamma_{B}}$ (and indeed, $\mb u\cdot \mb n=0$
implies that $\psi_{F}$ is constant on $\Gamma_{B}$). In order to use
Proposition~\ref{prop-Green},
we have to extend the stream function in the air such that
$u_{2}\to 0$ as $x_{2}\to +\infty\,$. Hence we extend it in the air $\cal D_A$ with the
unique\footnote{The existence and uniqueness can be proved via the double layer potential together with the Green kernel in $\T_{L}\times \R$ (see Proposition~\ref{prop-Green}).} solution of 
\begin{gather*}
\div \mb u= \curl \mb u =0 \text{ in } \D_{A}, \quad \mb u \cdot \mb n=g \text{ on } \Gamma_{A}, \\
 |\mb u | \to 0 \text{ when } x_{2}\to \infty, \quad \int_{\Gamma_{A}} \mb
 u\cdot \mb \tau \d s= 0 \,.
\end{gather*}
Note that this equation is the physically relevant formulation if we are
interested in the bi-fluid water-wave model, for which the continuity of
normal velocity simply reflects that the two fluids are not mixing.

Note also that it could, in principle, be possible to add some vortices in
the air. We should stress however that the circulation has to vanish at infinity in order to use
Proposition~\ref{prop-Green}, if not, we would have to change the extension below the bottom. 

This extended vector field can be expressed as 
\begin{equation}\label{vortex-stream}
 \mb u=\nabla^\perp \psi \,,
\end{equation}
where $\psi$ is continuous in $\D$ and determined up to an arbitrary
constant.
This extension will be sufficient for the vortex formulation.

Regarding the derivation of the dipole formulation, we now have to extend
$\mb u_{R}\,$. It will be convenient for the bottom condition to
extend $\mb u_{R}$ by zero in $\D_{B}$ (see further down Remark~\ref{rem-phiB}).
In order to achieve this, we must add the following assumption on
$\mb u_{\omega,\gamma}$: 
\begin{equation}\label{u-fgamma-2}
\mb u_{\omega,\gamma} \cdot \mb n =0 \text{ on } \Gamma_{B} \,.
\end{equation}
This assumption allows us to extend $\mb u_{R}$ in $\D_{B}\,$ by zero and
$\tilde \psi$ by the constant $\tilde \psi_{F}\vert_{\partial
 \Gamma_{B}}\,$.

Again, for a compatibility at infinity, and in order to
write $\mb u_{R}$ as a potential, we must extend $\mb u_{R}$ in the air
by a vector field satisfying 
\begin{gather*}
\div \mb u_{R}= \curl \mb u_{R} =0 \text{ in } \D_{A}\,, \quad \mb u_{R}
\cdot \mb n=g -\mb u_{\omega,\gamma} \cdot \mb n \text{ on } \Gamma_{A}\,, \\
 |\mb u_{R} | \to 0 \text{ when } x_{2}\to \infty\,, \quad \int_{\Gamma_{S}}
 \mb u_{R}\cdot \mb \tau \d s= 0 \,.
\end{gather*}
From the above equations, we can write
$\mb u_{R}=\nabla^\perp \tilde\psi=\nabla \phi \,,$ where $\tilde\psi$
is continuous and uniquely defined up to an arbitrary constant.
The potential $\phi$ jumps across $\Gamma_{S}$ and
$\Gamma_{B} \,,$ and we have complete freedom to choose independently the
constants in each of the connected components: $\D_{B}\,$, $\D_{F}$ and
$\D_{A} \,.$ These four constants will be determined below in order to be
able to write $\psi\,$, $\tilde \psi$ and $\phi$ in the form of a singular integral by
applying Proposition~\ref{prop-Green}.

Using the uniqueness of the solution to the elliptic problem \eqref{ellip1}, it follows that 
\begin{equation}\label{decomp}
\mb u=\mb u_{\omega,\gamma} + \nabla^\perp \tilde\psi= \mb u_{\omega,\gamma} + \nabla
\phi \ \text{in}\ \D_{F} \,.
\end{equation}
Even if we have already properly defined the extension of $\tilde \psi$ in
order to be able to use a Biot-Savart representation formula
\eqref{Green_formulation},
we still need to discuss the expression of $\mb u_{\omega,\gamma}$ in $\D_{A}$ to infer the value of $\mb u$ in the air. Such a discussion is postponed to the end of the next section.

\begin{remark}
We can apply the whole analysis of this paper to treat cases involving several submerged solids
$\cal S_{k}\Subset \D_{S}\,$, simply constructing the harmonic vector such
that 
 \[
 \div \mb H=\curl\mb H=0 \text{ in }\D_{F}, \quad \mb H\cdot\mb n=0 \text{ on }\Gamma_{B}\cup_{k} \partial\cal S_{k},\quad \int_{\Gamma_{B}} \mb H\cdot \mb \tau=\gamma_{0}, \quad \int_{\partial \cal S_{k}} \mb H\cdot \mb \tau=\gamma_{k}
 \]
 where $\gamma_{k}$ is initially given. In the same way, if we are only
 interested by the single-fluid water-waves equations, we can simply
 construct $\mb H$ initially in $\D_{F}\cup \Gamma_{S}\cup \D_{A}$ and this
 problem can be solved in the dipole formulation. If
 $\omega=\gamma_{0}=\gamma_{1}=\gamma_{k}=0$ for all $k$, the dipole
 formulation is possible for both the single-fluid and bi-fluid water-waves
 equations. Otherwise, we will need to use the vortex formulation where the
 inclusion of such solids is a minor modification of the numerical code. In the
 vortex formulation, we can even include the case where the solids are
 moving with a prescribed velocity and rotation by setting
 $\mb H\cdot\mb n=(\mb \ell_{i} + r_{i} \mb x^\perp)\cdot \mb n\,$. 

 The case of immerged solids moving under the influence of the flow involves the
 computation of pressure forces at the boundary of the solid (see
 e.g. equation (5.4) in \cite{Wilkening21}).
 The case of a floating (partially immerged) solid would be even more
 challenging (see the recent developments in \cite{Lannes1,Lannes2,Lannes3}).

\end{remark}

\subsection{Potential and dipole formulae}\label{sec-BS}
We now want to use Proposition~\ref{prop-Green} to express
$\widehat{\mb u}=u_{1}-\ic u_{2} $ (in the vortex formulation), $\phi$ and
$\tilde \psi$ (in the dipole formulation) as singular integrals (i.e. equations \eqref{BS-u}, \eqref{phi-int}, \eqref{psi-mu} below).
In the previous subsection, we have constructed continuous $\psi$ or
$\tilde \psi$ on $\D\,$, where the perpendicular gradient is continuous on
$\D_{B}\cup \D_{F}\cup \D_{A}\,$, and its normal part is continuous across
the interfaces $\Gamma_{B}$ and $\Gamma_{S}\,$. Extending $\mb u$ or $\mb
u_{R}$ in this way ensures that $\div \mb u =0$ in $\D\,$, whereas the jump
of the tangential part can be seen as a vortex sheet, namely 
 \[
 \curl \mb u=\Delta \psi =\omega + |z_{S,e}|^{-1} \gamma_{S} \delta_{\Gamma_{S}} +|z_{B,e}|^{-1} \gamma_{B}\delta_{\Gamma_{B}} \text{ in } \D
 \]
where 
\begin{equation}\label{def-gamma}
\begin{aligned}
 \gamma_{S}(e) &: =|z_{S,e}(e)|[\lim_{\mb z\in \D_{F}\to \mb z_{S}(e)} \mb u - \lim_{z\in \D_{A}\to z_{S}(e)} \mb u]\cdot \mb \tau(e) \\
\gamma_{B}(e) & :=-|z_{B,e}(e)|[\lim_{\mb z\in \D_{F}\to \mb z_{B}(e)} \mb u]\cdot \mb \tau(e)
\end{aligned} 
\end{equation}
are such that the mean value is 
\begin{equation}\label{vort-mv}
\int (\omega+ |z_{S,e}|^{-1}\gamma_{S} \delta_{\Gamma_{S}} + |z_{B,e}|^{-1}\gamma_{B}\delta_{\Gamma_{B}})=0\,,
 \end{equation}
 see \eqref{Stokes-comp}.
These formulae and the following ones also hold replacing $\mb u,\,
\omega\,, \gamma_{S}\,, \gamma_{B}$ by $\mb u_{R}\,, 0\,,
\tilde\gamma_{S}\,, \tilde\gamma_{B}\,$. 
 
Proposition~\ref{prop-Green} implies that $\psi$ is determined up to a
constant, which is fixed when we choose to represent\footnote{Even if
$\delta_{\Gamma}$ is not a bounded function, it belongs to $H^{-1}(\D)$
where the well-posedness of elliptic problem is usually proven, and the
formula can be rigorously established for $C^1$ curve, see
\cite{musk,kellogg,fabes}.} it as follows
\begin{align*}
\psi (\mb x)=& \int_{\Gamma_{S}} G(\mb x,\mb y) |z_{S,e}|^{-1}\gamma_{S} \d \sigma(\mb y)+ \int_{\Gamma_{B}} G(\mb x,\mb y) |z_{B,e}|^{-1}\gamma_{B} \d \sigma(\mb y) + \int_{\D_{F}} G(\mb x,\mb y) \omega(\mb y) \d \mb y \\
=& \int_{0}^{L_{S}} G(\mb x,\mb z_{S}(e)) \gamma_{S}(e) \d e+ \int_{0}^{L_{B}} G(\mb x,\mb z_{B}(e)) \gamma_{B}(e) \d e + \int_{\D_{F}} G(\mb x,\mb y) \omega(\mb y) \d \mb y\,. 
\end{align*}
By the explicit formula of the Green kernel, we deduce from the previous
formula the Biot-Savart law which yields the velocity $\mb u=\nabla^\perp
\psi$ for all $x$ in $\D_{F}\cup \D_{A}\cup \D_{S}\,$: 
\begin{align}
\widehat{\mb u} (x)
= &
 \int_{0}^{L_{S}} \gamma_{S}(e) \widehat{ \nabla^{\perp}G}(\mb x-\mb z_{S}(e)) \d e 
 + \int_{0}^{L_{B}} \gamma_{B}(e) \widehat{ \nabla^{\perp}G}(\mb x-\mb z_{B}(e)) \d e \nonumber\\
&+ \int_{\D_{F}} \widehat{ \nabla^{\perp}G}(\mb x,\mb y) \omega(\mb y) \d \mb y \nonumber\\
 =& \int_{0}^{L_{S}} \gamma_{S}(e) \frac1{2L\ic} \cot\Big(\frac{x-z_{S}(e) }{L/\pi}\Big) \d e 
+\int_{0}^{L_{B}} \gamma_{B}(e) \frac1{2L\ic} \cot\Big(\frac{x-z_{B}(e) }{L/\pi}\Big) \d e, \label{BS-u}\\
&+ \int_{\D_{F}} \frac1{2L\ic} \cot\Big(\frac{x-y }{L/\pi}\Big) \omega(\mb y) \d \mb y \nonumber
\end{align}
because
\begin{equation}\label{cot-nabla}
 \widehat{\nabla^\perp_{\mb x}} G(\mb x) = \frac{ - \sinh\frac{x_{2}}{L/(2\pi)} -\ic\sin \frac{x_{1}}{L/(2\pi)} }{2L\Big( \cosh\frac{x_{2}}{L/(2\pi)}-\cos\frac{x_{1}}{L/(2\pi)}\Big)} = \frac1{2L\ic} \cot\Big(\frac{x_{1}+\ic x_{2}}{L/\pi}\Big) \,,
\end{equation}
where we have used that 
$-\sinh b-\ic\sin a = -\ic(\sin a-\sin(\ic b)) = -2\ic \sin \frac{a-\ic b}2 \cos \frac{a+\ic b}2$
and
$\cosh b-\cos a = \cos{\ic b}-\cos a = 2 \sin \frac{a-\ic b}2 \sin \frac{a+\ic b}2\,$.
This formula with cotangent kernel is singular when $x$ goes to the
boundary $\Gamma_{S}\cup \Gamma_{B} \, .$ This is natural because it encodes
the jump of the tangential part of the velocity. The limit formula, the so
called Plemelj formulae, will play a crucial role in the sequel and are
recalled in Appendix~\ref{app-cot}. Another key tool presented in this
Appendix is the following desingularization rule
\begin{multline}\label{eq.desing}
 {\rm pv}\int \cot\Big(\frac{z(e)-z(e')}{L/\pi}\Big) f(e')\d e' \\= \int
 \cot\Big(\frac{z(e)-z(e')}{L/\pi}\Big) \frac{ f(e')z_{e}(e)-f(e)z_{e}(e')
 }{z_{e}(e)} \d e' \, ,
\end{multline}
because it transforms a principal value integral into a classical integral of
a smooth function. This exact relation will be systematically used in order
to handle regular terms, which can be integrated with greater accuracy,
resulting in improved stability. It is worth stressing that this desingularization
does not alter the accuracy of the scheme, as opposed to regularization
technics.
We would like to stress again that this periodic Biot-Savart law is
formally related to the usual Biot-Savart law in $\R^2$: see
Appendix~\ref{app-cot}. 

Theorem~\ref{theo:main-ellip} is stated for bounded vorticities. Because
the resulting equations will be later discretised, we restrict our
attention below to a vorticity $\omega$ composed of a constant part
$\omega_{0}\mathds{1}_{\D_{F}}$ and a part that we approximate by a sum of
Dirac masses $\sum_{j=1}^{N_{v}} \gamma_{v,j} \delta_{z_{v,j}(t)} \,$, see
\cite{GHL90}.

The velocity generated by the Dirac masses is simply 
$
 \frac1{2L\ic} \sum_{j=1}^{N_{v}}\gamma_{v,j} \cot\Big(\frac{x-z_{v,j} }{L/\pi}\Big)
\,.$
The velocity associated to the constant part can be simplified thanks to an integration by parts
\begin{align*}
 \int_{\D_{F}} \nabla^\perp G(\mb x-\mb y) \omega_{0} \d \mb y 
 =& \omega_{0}\begin{pmatrix}
- \int_{\D_{F}}( \nabla G) (\mb x-\mb y) \cdot \mb e_{2}(\mb y) \d \mb y \\
 \int_{\D_{F}} (\nabla G) (\mb x-\mb y) \cdot \mb e_{1} (\mb y) \d \mb y
\end{pmatrix}\\
=& \omega_{0}\begin{pmatrix}
\int_{\D_{F}}\nabla_{y}( G (\mb x-\mb y)) \cdot \mb e_{2}(\mb y) \d \mb y \\
- \int_{\D_{F}} \nabla_{y}( G (\mb x-\mb y)) \cdot \mb e_{1} (\mb y) \d \mb y
\end{pmatrix}\\
=& \omega_{0} \begin{pmatrix}
 \int_{\partial\D_{F}} G (\mb x-\mb y) \mb e_{2} (\mb y)\cdot \tilde{\mb n}_{F} (\mb y)\d \sigma(\mb y) \\
 -\int_{\partial\D_{F}} G (\mb x-\mb y) \mb e_{1} (\mb y)\cdot \tilde{\mb n}_{F}(\mb y) \d \sigma(\mb y)
\end{pmatrix}\\
= & -\frac{ \omega_{0}}{4\pi}\int_{\partial\D_{F}} \ln\Big(\cosh\frac{x_{2}-y_{2}}{L/(2\pi)} -\cos\frac{x_{1}-y_{1}}{L/(2\pi)} \Big) \tilde{\mb n}_{F}^\perp(\mb y) \d \sigma(\mb y)
\end{align*}
where $\tilde{\mb n}_{F}$ is the unit normal vector outward to $\D_{F}$. This implies that
\begin{align}
\int_{\D_{F}} \frac1{2L\ic} \cot&\Big(\frac{x-y }{L/\pi}\Big) \omega_{0} \d \mb y \nonumber\\
=&
\frac{ \omega_{0}}{4\pi}\int_{0}^{L_{S}} \ln\Big(\cosh \Imag
\frac{x-z_{S}(e)}{L/(2\pi)} -\cos \Real \frac{x-z_{S}(e)}{L/(2\pi)} \Big)
\overline{z_{S,e}(e)} \d e \nonumber\\
&-\frac{ \omega_{0}}{4\pi}\int_{0}^{L_{B}} \ln\Big(\cosh \Imag \frac{x-z_{B}(e)}{L/(2\pi)} -\cos \Real \frac{x-z_{B}(e)}{L/(2\pi)} \Big) \overline{z_{B,e}(e)} \d e\,,
\label{IPP}
\end{align}
which is well defined and continuous in $\D\,$. 

Let us note that we can compute $\mb u_{\omega,\gamma}$ 
 for $\omega= \omega_{0}+\sum \gamma_{v,j}\delta_{z_{v,j}}$ in the same way.

Therefore, we have a complete formula \eqref{BS-u} which gives $\mb u=\nabla^\perp \psi$ in terms of $\omega\,$, $\gamma_{S}$ and $\gamma_{B}\,$, which will be used for the vortex formulation.

For the dipole formulation, we have, exactly in the same way,
\begin{equation}\label{BS-uR}
\begin{aligned}
 \widehat{\mb u_{R}} (x)
=& \int_{0}^{L_{S}} \tilde\gamma_{S}(e) \frac1{2L\ic} \cot\Big(\frac{x-z_{S}(e) }{L/\pi}\Big) \d e \\
&+\int_{0}^{L_{B}} \tilde\gamma_{B}(e) \frac1{2L\ic} \cot\Big(\frac{x-z_{B}(e) }{L/\pi}\Big) \d e\,,
\end{aligned}
 \end{equation}
where
\begin{align*}
 \tilde\gamma_{S}(e) &: =|z_{S,e}(e)|[\lim_{\mb z\in \D_{F}\to \mb z_{S}(e)} \mb u_{R} - \lim_{\mb z\in \D_{A}\to \mb z_{S}(e)} \mb u_{R}]\cdot \mb \tau(e) \\
\tilde\gamma_{B}(e) & :=-|z_{B,e}(e)|[\lim_{\mb z\in \D_{F}\to \mb z_{B}(e)} \mb u_{R}]\cdot \mb \tau(e)\,.
\end{align*}

In the previous subsection, we have defined $\mb u_{R}$ and the extension
such that $\mb u_{R}=\nabla \phi$ in $\D_{B}\cup \D_{F}\cap \D_{A}$ where
we have the choice to fix one constant by connected component. As the mean
values of $\tilde\gamma_{S}$ and $\tilde\gamma_{B}$ are zero, we know from
the behavior at infinity \eqref{behinf} that $\nabla \phi =\mb u_R= \nabla^\perp \tilde
\psi$ goes to zero exponentially fast when $x_{2}\to \infty\,$. In order to
control the boundary term in the following computation, we thus set the constant
in $\D_{A}$ such that $\phi$ goes to zero at infinity. In the same way, we
set the constant in $\D_{B}$ so that $\phi_{B}\to 0$ when
$x_{2}\to-\infty\,$.
As $\phi$ is not continuous across the interfaces and we
need the value from both side, we denote the restriction of $\phi$ in
$\D_{F}$ (resp. in $\D_{A}$ and in $\D_{B}$) by $\phi_{F}$ (resp. by
$\phi_{A}$ and $\phi_{B}$). For any $\mb x\in \D_{F}\,$, we compute 
\begin{align*}
 \phi(\mb x)=&\langle \phi_{F}, \Delta G(\cdot -\mb x) \rangle \\
 =& -\int_{\D_{F}} \nabla \phi_{F} (\mb y)\cdot \nabla G(\mb y-\mb x)\d \mb y 
 + \int_{\Gamma_{S}} \phi_{F}(\mb y) \partial_{n} G(\mb y-\mb x) \d \sigma(\mb y) \\
& - \int_{\Gamma_{B}} \phi_{F}(\mb y) \partial_{n} G(\mb y-\mb x) \d \sigma(\mb y)\\
 =& \int_{\Gamma_{S}} \Big(\phi_{F}(\mb y) \partial_{n} G(\mb y-\mb x) - \partial_{n} \phi_{F} (\mb y) G(\mb y-\mb x)\Big) \d \sigma(\mb y)\\
 & - \int_{\Gamma_{B}}\Big(\phi_{F}(\mb y) \partial_{n} G(\mb y-\mb x) - \partial_{n} \phi_{F} (\mb y) G(\mb y-\mb x) \Big) \d \sigma(\mb y)\,,
\end{align*}
where we keep in mind that $\mb n=\mb \tau^\perp$ is pointing outward on
$\Gamma_{S}$ whereas it is pointing inward on $\Gamma_{B}\,$. As
$\mb u_R \cdot \mb n$ is continuous, we have
\begin{align*}
 \phi(\mb x)  =& \int_{\Gamma_{S}} \Big(\phi_{F}(\mb y) \partial_{n} G(\mb y-\mb x) - \partial_{n} \phi_{A} (\mb y) G(\mb y-\mb x)\Big) \d \sigma(\mb y)\\
 & - \int_{\Gamma_{B}}\Big(\phi_{F}(\mb y) \partial_{n} G(\mb y-\mb x) - \partial_{n} \phi_{B} (\mb y) G(\mb y-\mb x) \Big) \d \sigma(\mb y)\,.
\end{align*}
As $\Delta G(\cdot-\mb x)=0$ in $\D_{A}\cup \D_{B}$  (for $\mb x\in \D_{F}$), we can integrate by parts in
the air and in the bottom domains as we did above in $\D_{F}$ to state
\begin{gather*}
 \int_{\Gamma_{S}} \Big(\phi_{A}(\mb y) \partial_{n} G(\mb y-\mb x) - \partial_{n} \phi_{A} (\mb y) G(\mb y-\mb x)\Big) \d \sigma(\mb y)=0\\
 \int_{\Gamma_{B}}\Big(\phi_{B}(\mb y) \partial_{n} G(\mb y-\mb x) - \partial_{n} \phi_{B} (\mb y) G(\mb y-\mb x) \Big) \d \sigma(\mb y)=0
\end{gather*}
where we have used the fact that $G(\mb x)=\mathcal{O}(x_{2})$
and $\nabla G(\mb x)=\mathcal{O}(1)$ at infinity. This implies
\begin{align*}
 \phi(\mb x)
& = \int_{\Gamma_{S}} ( \phi_{F}(\mb y)-\phi_{A}(\mb y) ) \partial_{n} G(\mb y-\mb x) \d \sigma(\mb y) - \int_{\Gamma_{B}} ( \phi_{F}(\mb y)-\phi_{B}(\mb y)) \partial_{n} G(\mb y-\mb x) \d \sigma(\mb y)\,.
\end{align*}
Doing a similar computation for $\mb x\in \D_{A}$: \begin{align*}
 \phi(\mb x)
 =&\langle \phi_{A}, \Delta G(\cdot -\mb x) \rangle 
 = -\int_{\D_{A}} \nabla \phi_{A} (\mb y)\cdot \nabla G(\mb y-\mb x)\d \mb y - \int_{\Gamma_{S}} \phi_{A}(\mb y) \partial_{n} G(\mb y-\mb x) \d \sigma(\mb y)\\
 = &\int_{\Gamma_{S}} \partial_{n} \phi_{A} (\mb y) G(\mb y-\mb x)\d \sigma(\mb y) - \int_{\Gamma_{S}} \phi_{A}(\mb y) \partial_{n} G(\mb y-\mb x) \d \sigma(\mb y)\\
 = & \int_{\Gamma_{S}} \partial_{n} \phi_{F} (\mb y) G(\mb y-\mb x)\d \sigma(\mb y)- \int_{\Gamma_{S}} \phi_{A}(\mb y) \partial_{n} G(\mb y-\mb x) \d \sigma(\mb y)\\
 = & \int_{\Gamma_{S}} \phi_{F}(\mb y) \partial_{n} G(\mb y-\mb x)\d \sigma(\mb y)- \int_{\Gamma_{S}} \phi_{A}(\mb y) \partial_{n} G(\mb y-\mb x) \d \sigma(\mb y)\\ 
 &- \int_{\Gamma_{B}}\Big(\phi_{F}(\mb y) \partial_{n} G(\mb y-\mb x) - \partial_{n} \phi_{F} (\mb y) G(\mb y-\mb x) \Big) \d \sigma(\mb y)\\
 = & \int_{\Gamma_{S}} \Big(\phi_{F}(\mb y) - \phi_{A}(\mb y)\Big) \partial_{n} G(\mb y-\mb x) \d \sigma(\mb y)\\ 
 &- \int_{\Gamma_{B}}\Big(\phi_{F}(\mb y) \partial_{n} G(\mb y-\mb x) - \partial_{n} \phi_{B} (\mb y) G(\mb y-\mb x) \Big) \d \sigma(\mb y)\\
 =& \int_{\Gamma_{S}} (\phi_{F} (\mb y)-\phi_{A} (\mb y)) \partial_{n}G(\mb y-\mb x)\d \sigma(\mb y) - \int_{\Gamma_{B}} (\phi_{F} (\mb y)-\phi_{B} (\mb y)) \partial_{n}G(\mb y-\mb x)\d \sigma(\mb y) \,,
\end{align*}
we notice that this formula holds true in $\D_{B}\cup \D_{F}\cup \D_{A}\,$.
So, we are computing now $\partial_{n}G(\mb y-\mb x)$
\begin{align*}
\nabla G(\mb y-\mb x) \cdot \mb n(\mb y) \d \sigma(\mb y) 
 =&\frac1{2L\Big( \cosh\frac{z_{2}(e)-x_{2}}{L/(2\pi)}-\cos\frac{z_{1}(e)-x_{1}}{L/(2\pi)}\Big)}
\begin{pmatrix}
\sin \frac{z_{1}(e)-x_{1}}{L/(2\pi)} \\ \sinh\frac{z_{2}(e)-x_{2}}{L/(2\pi)}
\end{pmatrix}
\cdot
\begin{pmatrix}
 -z_{2,e}(e) \\ z_{1,e}(e) 
\end{pmatrix}\d e
\\
=&-\Real \Bigg[\frac{-\sinh\frac{z_{2}(e)-x_{2}}{L/(2\pi)} - \ic \sin \frac{z_{1}(e)-x_{1}}{L/(2\pi)} } {2L\Big( \cosh\frac{z_{2}(e)-x_{2}}{L/(2\pi)}-\cos\frac{z_{1}(e)-x_{1}}{L/(2\pi)}\Big)}
z_{e}(e) \Bigg]\d e\\
=&-\Real \Bigg[\frac1{2L\ic} \cot\Big(\frac{z(e) - x}{L/\pi}\Big) z_{e}(e) \Bigg]\d e
\end{align*}
so, setting
\begin{equation}\label{def-mu}
\mu_{S}(e)=(\phi_{F} - \phi_{A})(z_{S}(e))\,, \quad \mu_{B}(e)=(\phi_{B}- \phi_{F}) (z_{B}(e))\,, 
\end{equation}
we finally get for any $\mb x\in \D_{F}\cup \D_{A} \cup \D_{B}$
\begin{align}
 \phi(\mb x)
 =&- \int_{0}^{L_{S}} \mu_{S}(e) \Real \Bigg[\frac1{2L\ic} \cot\Big(\frac{z_{S}(e) - x}{L/\pi}\Big) z_{S,e}(e) \Bigg]\d e \nonumber\\
&-\int_{0}^{L_{B}} \mu_{B}(e)\Real \Bigg[\frac1{2L\ic} \cot\Big(\frac{z_{B}(e) - x}{L/\pi}\Big) z_{B,e}(e) \Bigg]\d e \nonumber \\
 =&\int_{0}^{L_{S}} \mu_{S}(e) \Real \Bigg[\frac1{2L\ic} \cot\Big(\frac{x-z_{S}(e) }{L/\pi}\Big) z_{S,e}(e) \Bigg]\d e \label{phi-int}\\
&+\int_{0}^{L_{B}} \mu_{B}(e)\Real \Bigg[\frac1{2L\ic} \cot\Big(\frac{x-z_{B}(e) }{L/\pi}\Big) z_{B,e}(e) \Bigg]\d e\,.\nonumber
\end{align}
We note here that we did not provide any restriction on the constant for
$\phi_{F}$ so the previous formula holds true if we change $\phi_{F}$ (so
$\mu_{S}$ and $\mu_{B}$) by a constant. It is therefore possible 
to fix initially this constant in such a way that 
\begin{equation}\label{avgmuB}
\int_{0}^{L_{S}}\mu_{S,0}(e)\d e=0\,.
\end{equation}
This condition is not conserved in time.

Let us also note that with our extension and Assumption \eqref{u-fgamma-2},
we have $\phi_{B}=0$ in $\D_{B}\,$. 

It is also possible to derive the stream function $\tilde \psi$ from
$\mu_{S}$ and $\mu_{B}\,$. To do this, we first remark that 
\[
\mb u_R\cdot \mb \tau = \nabla \phi \cdot \mb \tau = |z_{e}|^{-1} \partial_{e} ( \phi(z)) \,,
\] 
hence
\[
\tilde\gamma_{S}(e)=\partial_{e} \Big[ \phi_{F}(z_{S}(e)) - \phi_{A}(z_{S}(e)) \Big] =\partial_{e} \mu_{S}(e)
\text{\quad and\quad}
\tilde\gamma_{B}(e)=-\partial_{e} \phi_{F}(z_{B}(e)) = \partial_{e} \mu_{B}(e)\,,
\]
and then for any constants $C_{S},C_{B}\in \R$
\begin{align*}
 \tilde \psi(x)
& = \int_{0}^{L_{S}} \partial_{e}(\mu_{S}(e)+C_{S}) G(\mb z_{S}(e)-\mb x) \d e + \int_{0}^{L_{B}} \partial_{e}( \mu_{B}(e)+C_{B}) G(\mb z_{B}(e)-\mb x) \d e\\
& = - \int_{0}^{L_{S}}(\mu_{S}(e)+C_{S}) \partial_{e} \Big(G(\mb z_{S}(e)-\mb x)\Big) \d e - \int_{0}^{L_{B}} (\mu_{B}(e)+C_{B}) \partial_{e} \Big(G(\mb z_{B}(e)-\mb x)\Big) \d e\,.
\end{align*}
So we need to compute
\begin{align*}
 \nabla G(\mb z(e)-\mb x )\cdot 
\begin{pmatrix}
 z_{e,1}\\z_{e,2}
\end{pmatrix}
=&
\frac1{2L\Big( \cosh\frac{z_{2}(e)-x_{2}}{L/(2\pi)}-\cos\frac{z_{1}(e)-x_{1}}{L/(2\pi)}\Big)}
\begin{pmatrix}
\sin \frac{z_{1}(e)-x_{1}}{L/(2\pi)} \\ \sinh\frac{z_{2}(e)-x_{2}}{L/(2\pi)}
\end{pmatrix}
\cdot
\begin{pmatrix}
 z_{e,1}(e) \\ z_{e,2}(e)
\end{pmatrix}
\\
=&-\Imag \Bigg[\frac{-\sinh\frac{z_{2}(e)-x_{2}}{L/(2\pi)} - \ic \sin \frac{z_{1}(e)-x_{1}}{L/(2\pi)} } {2L\Big( \cosh\frac{z_{2}(e)-x_{2}}{L/(2\pi)}-\cos\frac{z_{1}(e)-x_{1}}{L/(2\pi)}\Big)}
z_{e}(e) \Bigg] \\
=&-\Imag \Bigg[\frac1{2L\ic} \cot\Big(\frac{z(e)-x}{L/\pi}\Big)
z_{e}(e) \Bigg]
\end{align*}
and we finally get 
\begin{align}
\tilde \psi (x)
 =&\int_{0}^{L_{S}} (\mu_{S}(e)+C_{S}) \Imag \Bigg[\frac1{2L\ic} \cot\Big(\frac{z_{S}(e) - x}{L/\pi}\Big) z_{S,e}(e) \Bigg]\d e \nonumber\\
&+\int_{0}^{L_{B}} (\mu_{B}(e)+C_{B})\Imag \Bigg[\frac1{2L\ic} \cot\Big(\frac{z_{B}(e) - x}{L/\pi}\Big) z_{B,e}(e) \Bigg]\d e\nonumber \\
=&-\int_{0}^{L_{S}} (\mu_{S}(e)+C_{S}) \Imag \Bigg[\frac1{2L\ic} \cot\Big( \frac{x-z_{S}(e) }{L/\pi}\Big) z_{S,e}(e) \Bigg]\d e \label{psi-mu}\\
&-\int_{0}^{L_{B}} (\mu_{B}(e)+C_{B})\Imag \Bigg[\frac1{2L\ic}
 \cot\Big(\frac{ x - z_{B}(e) }{L/\pi}\Big) z_{B,e}(e) \Bigg]\d e
\nonumber\,.
\end{align}

As this formula is valid for any values of $C_{B}$ and $C_{S}\,$, it holds
true for $C_{S}=C_{B}=0$ and besides
\[
\int_{0}^{L_{S}} \Imag \Bigg[\frac1{2L\ic} \cot\Big( \frac{x-z_{S}(e)
 }{L/\pi}\Big) z_{S,e}(e) \Bigg]\,\d e = \int_{0}^{L_{B}}\Imag
\Bigg[\frac1{2L\ic} \cot\Big(\frac{ x - z_{B}(e) }{L/\pi}\Big) z_{B,e}(e)
 \Bigg]\, \d e =0\,.
\]

For Section~\ref{sec-bernoulli}, it will be convenient to introduce the
quantity
\[
\Phi_{S}(e)=(\phi_{F} + \phi_{A})(z_{S}(e)) \, ,
\]
which is complementary to $\mu_{S}\, ,$ 
and which can be expressed thanks to the formula giving $\phi$ and the
limit formula (see Appendix~\ref{app-cot}) 
\begin{align*}
 \Phi_{S}(e)
 =& \int_{0}^{L_{S}} \mu_{S}(e') \Real \Bigg[\frac1{L\ic} \cot\Big(\frac{z_{S}(e)-z_{S}(e') }{L/\pi}\Big) z_{S,e}(e') \Bigg]\d e'\\
&+\int_{0}^{L_{B}} \mu_{B}(e')\Real \Bigg[\frac1{L\ic} \cot\Big(\frac{z_{S}(e)-z_{B}(e') }{L/\pi}\Big) z_{B,e}(e') \Bigg]\d e'\,.
\end{align*}
Computing the limit for $e'\to e\,$, we note that the first
integral is a classical integral of a continuous function, where the
extension for $e'=e$ is 
\[
-\mu_{S}(e) \Real \Big[ \frac{z_{S,ee}(e)}{2\pi \ic z_{S,e}(e)}\Big].
\]
Even if this integral could be well approximated by Riemann sum for smooth
fluid surface, it occurs that the following
formula will be convenient to get non singular integrals, taking advantage
of the desingularization \eqref{eq.desing}
\begin{multline}\label{phi-mu}
 \Phi_{S}(e)
 = \int_{0}^{L_{S}} ( \mu_{S}(e') - \mu_{S}(e)) \Real \Bigg[\frac1{L\ic} \cot\Big(\frac{z_{S}(e)-z_{S}(e') }{L/\pi}\Big) z_{S,e}(e') \Bigg]\d e'\\
+\int_{0}^{L_{B}} \mu_{B}(e')\Real \Bigg[\frac1{L\ic} \cot\Big(\frac{z_{S}(e)-z_{B}(e') }{L/\pi}\Big) z_{B,e}(e') \Bigg]\d e' 
\end{multline}
which is then extended for $e=e'$ by zero.

We conclude this section with one last compatibility condition, which is not used in
this article but will be used in a forthcoming article. 
\begin{remark}
The function $z\mapsto \phi-\ic \tilde \psi$ is harmonic in $\D_{B}\cup
\D_{F}\cup \D_{A}\,$, hence the integrals along two curves going from left
to right are the same if both curves are included in the same connected
component. With the limit at infinity, it is clear that 
 \[
 \int_{0}^{L_{S}}\Big(\phi_{A}(z_{S}(e))-\ic\tilde \psi(z_{S}(e)) \Big)z_{S,e}(e)\d e=-\ic L\lim_{x_{2}\to+\infty} \tilde \psi\,.
 \]
 whereas 
 \[
 \int_{0}^{L_{B}}\Big(\phi_{B}(z_{B}(e))-\ic\tilde \psi(z_{B}(e)) \Big)z_{B,e}(e)\d e=-\ic L \lim_{x_{2}\to-\infty} \tilde \psi = \ic L \lim_{x_{2}\to+\infty} \tilde \psi\,.
 \]
Inside the fluid we have
\[
 \int_{0}^{L_{S}}\Big(\phi_{F}(z_{S}(e))-\ic\tilde \psi(z_{S}(e)) \Big)z_{S,e}(e)\d e = \int_{0}^{L_{B}}\Big(\phi_{F}(z_{B}(e))-\ic\tilde \psi(z_{B}(e)) \Big)z_{B,e}(e)\d e\,.
\]
By continuity of the stream function, we get
\begin{align*}
 \int_{0}^{L_{S}}\mu_{S}(e) z_{S,e}(e) \d e=& \int_{0}^{L_{S}} (\phi_{F} - \phi_{A})(z_{S}(e))z_{S,e}(e) \d e\\
 =&\int_{0}^{L_{S}}\Big(\phi_{F}(z_{S}(e))-\ic\tilde \psi(z_{S}(e)) \Big)z_{S,e}(e)\d e\\
 & - \int_{0}^{L_{S}}\Big(\phi_{A}(z_{S}(e))-\ic\tilde \psi(z_{S}(e)) \Big)z_{S,e}(e)\d e \\
 =&\int_{0}^{L_{B}}\Big(\phi_{F}(z_{B}(e))-\ic\tilde \psi(z_{B}(e)) \Big)z_{B,e}(e)\d e + \ic L\lim_{x_{2}\to+\infty} \tilde \psi\\
 =& \int_{0}^{L_{B}} (\phi_{F}-\phi_{B})(z_{B}(e)) z_{B,e}(e)\d e+ 2\ic L\lim_{x_{2}\to+\infty} \tilde \psi\\
 =&- \int_{0}^{L_{B}}\mu_{B}(e) z_{B,e}(e) \d e+ 2\ic L\lim_{x_{2}\to\infty} \tilde \psi
 \end{align*}
we thus have for all time
\begin{equation} 
 \int_{0}^{L_{S}}\mu_{S}(e) \Real\Big[z_{S,e}(e) \Big]\d e=- \int_{0}^{L_{B}}\mu_{B}(e) \Real\Big[z_{B,e}(e) \Big]\d e\,.
 \end{equation} 
\end{remark}

\begin{remark}\label{rem:D2N}
The celebrated Dirichlet to Neumann operator in the Zakharov-Craig-Sulem
formulation \cite{Zakharov68,CraigSulemSulem} is very close to the dipole derivation. For $\varphi\in
H^{1/2}(\Gamma_{S})$ given, the principle is indeed to find
$u_{R}=\nabla\phi_F=\nabla^\perp \tilde \psi$ such that 
\[
\Delta \phi_{F}=0\text{ in } \D_{F}\ ,\ \partial_{n} \phi_{F}=0\text{ on } \Gamma_{B}\ , \ \phi_{F}=\varphi\text{ on } \Gamma_{S}\, .
\]
Extending as we did $\tilde\psi$ by continuity and defining $\phi$, we can represent $\phi$ through the singular representation formulation \eqref{phi-int}. Therefore, we should first find uniquely $\mu_{S}$ and $\mu_{B}$ such that $\phi_{B}=0$ on $\Gamma_{B}$ and $\phi_{F}=\varphi$ on $\Gamma_{S}$ thanks to the limit formulae of Appendix~\ref{app-cot} (see \eqref{muB-muS} for this kind of application). With $(\mu_{S},\mu_{B})$ found, we differentiate in order to get $(\tilde\gamma_{S},\tilde\gamma_{B})$ which allow us to construct $u_{R}$ \eqref{BS-uR}, hence $\partial_{n}\phi_{F}\vert_{\Gamma_{S}}$ again with the limit formulae. This ends the definition of the Dirichlet to Neumann operator $\varphi\mapsto \partial_{n}\phi_{F}\vert_{\Gamma_{S}}$.
\end{remark}

\subsection{Discussion about $\mb u_{\omega,\gamma}$ for the dipole formulation.}\label{sec-uog}

The simplest case that we will study in details is the case where $\omega=\gamma=0$ where we choose of course $\mb u_{\omega,\gamma}=0$. Then $\mb u = \nabla \phi_{F}=\nabla^\perp \tilde\psi_{F}$ is naturally defined in the full domain $\D$. In this easy case, we can consider both the bi-fluid water-wave
equation, in which the air is assumed to be an incompressible fluid with a non-zero density, or
the single-fluid water-waves equations, where we neglect the density of the
air in $\D_{A}\,$.

When we have some vorticity or background current, we need to discuss the expression of $\mb u_{\omega,\gamma}$ in $\D_{A}$ to infer the value of $\mb u$ in the air.
There are essentially two natural options:
\begin{itemize}
 \item either to have an explicit formula for $\mb u_{\omega,\gamma}\,$, or
 at least assume it is independent of time;
 \item or to extend by zero.
\end{itemize}
The choice depends on whether we need the physical air velocity, i.e. if we consider the single fluid or the bi-fluid water-wave
equation.

In the latter case (single-fluid), we do not need to know the velocity in the air, and we can simply set
\[
\mb u_{\omega,\gamma} = \frac{\gamma}{L}\mb e_{1} \mathds{1}_{\D_{F}}\quad \text{if
 the bottom is flat and }\omega=0 \,.
\]
Then we cannot say that $\mb u_{\omega,\gamma} + \nabla^\perp \tilde\psi$ defines
the velocity in the air, because the normal part of the velocity is not
continuous. A natural idea would then be to set
\[
\mb u_{\omega,\gamma} = \frac{\gamma}{L}\mb e_{1} \chi(x_{2})\quad \text{if the bottom is flat and }\omega=0\,.
\]
If we choose $\chi(x_{2})= 1$ for all $x_{2}\,$, this implies that a
non-physical circulation is present in the air, which is
equal to the circulation in the water.
Alternatively if $\chi(x_{2})$ is chosen to decay smoothly from $1$ near
the interface to $0$ at infinity, this implies a strange, also non-physical, vorticity in the air
$\curl \mb u=-\frac{\gamma}{L} \chi'(x_{2}) \,.$
Both cases do not correspond to the actual air velocity.
Hence, in the limiting case of vanishing air density, we can use this simpler
expression for the velocity $\mb u_{\omega,\gamma}\, .$
The air velocity can, however, not be reconstructed in that case (as it
does not influence the interface evolution). Therefore, we will consider later the case of the single-fluid water-waves equation without vorticity but with background current, constructing the time independent $\mb H$ solving
 \[
 \div \mb H=\curl\mb H=0 \text{ in }\D_{F}\cup \Gamma_{S}\cup \D_{A}\,, \quad \mb H\cdot\mb n=0 \text{ on }\Gamma_{B}\,,\quad \int_{\Gamma_{B}} \mb H\cdot \mb \tau=\gamma
 \]
and setting $\mb u_{\omega,\gamma} =\mb H\,$. Namely, we set
 \[
\widehat{\mb u_{\omega,\gamma}} (x)= \widehat{\mb H} (x)=\int_{0}^{L_{B}} \gamma_{B,H}(e) \frac1{2L\ic} \cot\Big(\frac{x-z_{B}(e) }{L/\pi}\Big) \d e
\]
where $\gamma_{B,H}$ is the unique\footnote{See Section~\ref{sec-gammaB}.} solution of
\[
{\rm pv} \int_{0}^{L_{B}} \gamma_{B,H}(e') \Imag \Bigg[ \frac{z_{B,e}(e)}{2L\ic} \cot\Big(\frac{z_{B}(e)-z_{B}(e') }{L/\pi}\Big)\Bigg] \d e'
 =0 \, ,
 \]
with
\begin{equation*}
 \int_{0}^{L_{B}} \gamma_{B,H}(e') \d e' = - 2\gamma \,.
\end{equation*}
Indeed, $\widehat{\mb u_{\omega,\gamma}}$ constructed in this way has a
circulation $\gamma$ in $\D_{F}\cup \Gamma_{S}\cup \D_{A}$ and $-\gamma$ in
$\D_{B}\,$, which is compatible with the limit behavior of the stream
function associated to a vorticity which has a non vanishing mean value (see
Proposition~\ref{prop-Green}). For the flat bottom, we recover $\mb H = \frac{\gamma}{L}\mb e_{1}$ in $\D_{F}\cup \Gamma_{S}\cup \D_{A}$ and $\mb H = -\frac{\gamma}{L}\mb e_{1}$ in $\D_{S}$.

In the case of the single-fluid water-waves equations with a flat
 bottom $\Gamma_{B}=\T_{L}\times \{ -h_{0}\}\,$, we could think to the simple formula coming from the image method:
\[
\widehat{\mb u_{\omega,\gamma} }(x)
= 
 \frac{\gamma}{L} + \int_{\D} \frac1{2L\ic} \cot\Big(\frac{x-y }{L/\pi}\Big) (\omega \mathds{1}_{\D_{F}}+\tilde \omega \mathds{1}_{\D_{B}})(\mb y) \d \mb y 
\]
where $\tilde \omega(x_{1},x_{2}):=-\omega(x_{1}, -x_{2}-2h_{0})\,$. Hence,
we have after two integrations by parts
\begin{align*}
\widehat{\mb u_{\omega,\gamma}} (x) 
=& \frac{\gamma}{L} + \frac1{2L\ic} \sum_{j=1}^{N_{v}}\gamma_{v,j} \Bigg( \cot\Big(\frac{x-z_{v,j} }{L/\pi} \Big)- \cot\Big(\frac{x-\overline{z_{v,j}}+2\ic h_{0} }{L/\pi} \Big) \Bigg)\\
&+\frac{ \omega_{0}}{4\pi}\int_{0}^{L_{S}} \ln\Big(\cosh \Imag \frac{x-z_{S}(e)}{L/(2\pi)} -\cos \Real \frac{x-z_{S}(e)}{L/(2\pi)} \Big) \overline{z_{S,e}(e)} \d e\\
&-\frac{ \omega_{0}}{2\pi}\int_{0}^{L_{B}} \ln\Big(\cosh \Imag \frac{x-z_{B}(e)}{L/(2\pi)} -\cos \Real \frac{x-z_{B}(e)}{L/(2\pi)} \Big) \overline{z_{B,e}(e)} \d e\\
&+\frac{ \omega_{0}}{4\pi}\int_{0}^{L_{S}} \ln\Big(\cosh \Imag \frac{x-\overline{z_{S}(e)}+2\ic h_{0}}{L/(2\pi)} -\cos \Real \frac{x-\overline{z_{S}(e)}+2\ic h_{0}}{L/(2\pi)} \Big) z_{S,e}(e) \d e\,.
\end{align*}
Unfortunately, we will observe later (see Section~\ref{sec-bernoulli}) that this approach is
unpractical.

In the first case (bi-fluid formulation), if we want to
extend $\mb u$ in such a way that the normal component of the velocity is
continuous and $\div \mb u=\curl \mb u=0$ in $\D_{A}\,$, we then have to
solve at any time an elliptic problem in $\D_{A}$ to extend the flow
$\mb u_{\omega,\gamma}$ in the correct way.
Alternatively, we could prefer to extend $\mb u_{\omega,\gamma}$ by zero.
In this case, we would need to add the following condition
\begin{equation}\label{u-fgamma-3}
\mb u_{\omega,\gamma} \cdot \mb n =0 \text{ on } \Gamma_{S} \,.
\end{equation}
and solve at any time the elliptic problem \eqref{u-fgamma-1} in $\D_{F}$
with \eqref{u-fgamma-2} and \eqref{u-fgamma-3}. Namely, we set
\begin{align*}
\widehat{\mb u_{\omega,\gamma}} (x)=&\int_{0}^{L_{S}} \gamma_{S,\omega,\gamma}(e) \frac1{2L\ic} \cot\Big(\frac{x-z_{S}(e) }{L/\pi}\Big) \d e 
+ \int_{0}^{L_{B}} \gamma_{B,\omega,\gamma}(e) \frac1{2L\ic} \cot\Big(\frac{x-z_{B}(e) }{L/\pi}\Big) \d e \\
&+ \int_{\D_{F}} \frac1{2L\ic} \cot\Big(\frac{x-y }{L/\pi}\Big) \omega (\mb y) \d \mb y 
\end{align*}
where $(\gamma_{S,\omega,\gamma}\,, \gamma_{B,\omega,\gamma})$ is properly constructed at each time. It will also be observed in Section~\ref{sec-bernoulli} that this approach is too complicated. In the case of a bi-fluid formulation with circulation, or with internal vorticity, the vortex method will be preferred.

The dipole formulation can however be considered only in two cases: the
case of a bi-fluid formulation in the absence of both internal vorticity
and circulation or the case of a single fluid formulation in the absence of
internal vorticity.

This concludes the proof of Theorem~\ref{theo:main-ellip} since the
solution of an elliptic problem in a $C^{1,1}$ domain belongs to
$\cap_{p>2} W^{2,p}$ when $\omega \in L^{\infty} \, .$ This justifies that
the velocity and then $\gamma$ is $C^{0,\alpha} $
for all $\alpha \in (0,1)\, .$
We recall that $\gamma_S = \partial \mu_S/\partial e \, ,$ which yields
higher regularity on $\mu_S \, .$

Note that $\gamma_S \in C^{0,\alpha}$ is enough to be able to reformulate
principal value integrals in the form of classical integrals \eqref{eq.desing}.

\section{Evolution of water-waves}\label{sec3}

The bottom $z_{B}$ and the constant part of the vorticity $\omega_{0}$ are
initially given. At any time, for a given $(z_{v,j})_{j=1,\dots,N_{v}}$ and
$z_{S}\,$, we have established in Section~\ref{sec-elliptic} the existence of
$(\gamma_{S},\gamma_{B})$ or $(\mu_{S},\mu_{B})$ from $g\,$. Conversely, from $\gamma_{S}$ or $\mu_{S}\,$,
we will first show that there is a unique $\gamma_{B}$ or $\mu_{B}$ satisfying
the boundary conditions at the bottom, and hence proving Remark~\ref{rem:main-ellip}. We will thus use
Section~\ref{sec-BS} to get the velocity everywhere, and then deduce the 
displacements of the point vortex and the free surface: $\partial_{t}
z_{v,j}$ and $\partial_{t}z_{S}\,$. The last step is to use the Euler or
the Bernoulli equations to determine $\partial_{t} \gamma_{S}$ or
$\partial_{t} \mu_{S}\, ,$ namely proving Theorem~\ref{theo:main-dyn}. 

Therefore, if we know $g\,$ initially, we can construct
$(\gamma_{S,0},\gamma_{B,0})$ or $(\tilde\gamma_{S,0}, \tilde\gamma_{B,0})$ such that the
corresponding velocity \eqref{BS-u} or \eqref{BS-uR} verifies the correct
boundary conditions. From $\tilde\gamma_{S,0}$ we will construct $\mu_{S,0}$
as the primitive of $\tilde\gamma_{S,0}$ with zero mean.
For $t>0\,$, the main numerical strategy can be summarized as
\begin{itemize}
 \item for the vortex method:
 \[
\Big(z_{S}\,,(z_{v,j})_{j}\,, \gamma_{S}\Big) \mapsto \Big(z_{S}\,,(z_{v,j})_{j} \,, \gamma_{S}\,, \gamma_{B}\Big) \mapsto \Big(\partial_{t}z_{S}\,,(\partial_{t}z_{v,j})_{j} \,, \partial_{t}\gamma_{S}\Big) \, ;
\]
 \item for the dipole method:
 \[
\Big(z_{S}\,,(z_{v,j})_{j} \,, \mu_{S}\Big) \mapsto \Big(z_{S}\,,(z_{v,j})_{j} \,, \mu_{S}\,, \mu_{B}\Big) \mapsto \Big(\partial_{t}z_{S}\,,(\partial_{t}z_{v,j})_{j} \,, \partial_{t}\mu_{S}\Big) \,.
\]
\end{itemize}

\subsection{Determination of $\gamma$ or $\mu$ from the boundary condition}\label{sec-gammaBmuB}

The quantities $z_{B}\,,$ $\gamma\,,$ $\omega_{0}\,,$ $(\gamma_{v,j})_{j}\,
, $ $g \, ,$ are
given by the initial conditions, and we want to solve 
\begin{itemize}
 \item $(z_{S,0}\,,(z_{v,j,0})_{j} \,, g_{0})\mapsto \gamma_{S,0}$ for the initial setting in the vortex formulation;
 \item $(z_{S}\,,(z_{v,j})_{j}\,, \gamma_{S})\mapsto \gamma_{B}$ for every time step in the vortex formulation; 
 \item $(z_{S,0}\,, u_{\omega,\gamma,0}\,,g_{0})\mapsto \tilde\gamma_{S,0}\mapsto \mu_{S,0}$ for the initial setting in the dipole formulation;
 \item $(z_{S}\,, u_{\omega,\gamma}\,,\mu_{S})\mapsto \mu_{B}$ for every time step in the dipole formulation.
\end{itemize}

\subsubsection{Initial $\gamma_{S,0}$ for the vortex formulation}\label{sec-gammaS0}

In many situation, such as solitary waves, 
$z_{B}\,,z_{S}\,, g=\mb u_{F}\cdot \mb n\vert_{\Gamma_{S}}\,, \gamma\,,
\omega_{0}$ and $(\gamma_{v,j}\,,z_{v,j})_{j=1,\dots, N_{v}}\,$ are known initially. By
uniqueness of the elliptic problem (see \eqref{ellip1} with our extension
\eqref{vortex-stream}), we know that there exists a unique pair
$(\gamma_{S}\,,\gamma_{B})$ such that the normal velocity $\mb u\cdot \mb
n=-\Imag (\hat{\mb u}\frac{z_{S,e}}{|z_{S,e}|})$ verifies the proper boundary
condition on the free surface. Using \eqref{BS-u},
\begin{multline*}
{\rm pv}\int_{0}^{L_{S}} \gamma_{S}(e') \Imag \Bigg[\frac{z_{S,e}(e)}{2L\ic} \cot\Big(\frac{z_{S}(e)-z_{S}(e') }{L/\pi}\Big)\Bigg] \d e'\\
+
 \int_{0}^{L_{B}} \gamma_{B}(e') \Imag \Bigg[ \frac{z_{S,e}(e)}{2L\ic} \cot\Big(\frac{z_{S}(e)-z_{B}(e') }{L/\pi}\Big)\Bigg] \d e'
 = {\rm RHS}_{V0,S}(e) \, ,
 \end{multline*}
where thanks to \eqref{IPP}
\begin{align*}
 &\hspace{-10pt}{\rm RHS}_{V0,S}(e) \\=& -g |z_{S}(e)| - \int_{\D_{F}} \Imag \Bigg[\frac{z_{S,e}(e)}{2L\ic} \cot\Big(\frac{z_{S}(e)-y }{L/\pi}\Big)\Bigg] \omega_{0} \d \mb y \\
 =& -g |z_{S}(e)|
- \sum_{j=1}^{N_{v}}\gamma_{v,j} \Imag \Bigg[\frac{z_{S,e}(e)}{2L\ic} \cot\Big(\frac{z_{S}(e)-z_{v,j} }{L/\pi}\Big) \Bigg] \\
&-\frac{ \omega_{0}}{4\pi}\int_{0}^{L_{S}} \ln\Big(\cosh \Imag \frac{z_{S}(e)-z_{S}(e')}{L/(2\pi)} -\cos \Real \frac{z_{S}(e)-z_{S}(e')}{L/(2\pi)} \Big) \Imag \Bigg[z_{S,e}(e)\overline{z_{S,e}(e')}\Bigg] \d e'\\
&+\frac{ \omega_{0}}{4\pi}\int_{0}^{L_{B}} \ln\Big(\cosh \Imag \frac{z_{S}(e)-z_{B}(e')}{L/(2\pi)} -\cos \Real \frac{z_{S}(e)-z_{B}(e')}{L/(2\pi)} \Big) \Imag \Bigg[z_{S,e}(e)\overline{z_{B,e}(e')}\Bigg] \d e'\, ;
\end{align*}
on the bottom:
\begin{multline*}
\int_{0}^{L_{S}} \gamma_{S}(e') \Imag \Bigg[\frac{z_{B,e}(e)}{2L\ic} \cot\Big(\frac{z_{B}(e)-z_{S}(e') }{L/\pi}\Big)\Bigg] \d e'\\
+
{\rm pv} \int_{0}^{L_{B}} \gamma_{B}(e') \Imag \Bigg[ \frac{z_{B,e}(e)}{2L\ic} \cot\Big(\frac{z_{B}(e)-z_{B}(e') }{L/\pi}\Big)\Bigg] \d e'
 = {\rm RHS}_{V0,B}(e)
 \end{multline*}
 where 
\begin{align*}
 &\hspace{-10pt}{\rm RHS}_{V0,B}(e)\\
 =& -\int_{\D_{F}} \Imag \Bigg[\frac{z_{B,e}(e)}{2L\ic} \cot\Big(\frac{z_{B}(e)-y }{L/\pi}\Big)\Bigg] \omega_{0} \d \mb y \\
 =& 
- \sum_{j=1}^{N_{v}}\gamma_{v,j} \Imag \Bigg[\frac{z_{B,e}(e)}{2L\ic} \cot\Big(\frac{z_{B}(e)-z_{v,j} }{L/\pi}\Big) \Bigg] \\
&-\frac{ \omega_{0}}{4\pi}\int_{0}^{L_{S}} \ln\Big(\cosh \Imag \frac{z_{B}(e)-z_{S}(e')}{L/(2\pi)} -\cos \Real \frac{z_{B}(e)-z_{S}(e')}{L/(2\pi)} \Big) \Imag \Bigg[z_{B,e}(e)\overline{z_{S,e}(e')}\Bigg] \d e'\\
&+\frac{ \omega_{0}}{4\pi}\int_{0}^{L_{B}} \ln\Big(\cosh \Imag \frac{z_{B}(e)-z_{B}(e')}{L/(2\pi)} -\cos \Real \frac{z_{B}(e)-z_{B}(e')}{L/(2\pi)} \Big) \Imag \Bigg[z_{B,e}(e)\overline{z_{B,e}(e')}\Bigg] \d e'\,,
\end{align*}
together with the circulation assumptions:
\begin{gather*}
 \int_{0}^{L_{S}} \gamma_{S}(e') \d e' =\gamma-\omega_{0} | \D_{F}| - \sum_{j}\gamma_{v,j} \\
 \int_{0}^{L_{B}} \gamma_{B}(e') \d e' = - \gamma\,.
\end{gather*}

The existence and uniqueness of a solution is related to the operator $B$
in \cite{ADL}, and Section~\ref{sec-discret} will detail how
these integrals can be discretized, ensuring that the resulting matrices are invertible.

\subsubsection{Time dependent $\gamma_{B}$ for the vortex formulation}\label{sec-gammaB}

For any given time, knowing $z_{B}\,$,$ \gamma\,$,
$\omega_{0}\,$, $(\gamma_{v,j})_{j=1,\dots, N_{v}}$ from the initial conditions, we
need to construct $\gamma_{B}$ from $z_{S}\,$, $\gamma_{S}\,$,
$(z_{v,j})_{j=1,\dots, N_{v}}$ such that the normal velocity $\mb u\cdot \mb
n=-\Imag (\hat{\mb u}\frac{z_{S,e}}{|z_{S,e}|})$ satisfies the
impermeability boundary condition on the bottom. This problem is then
simpler than the previous one: 
\begin{equation}\label{gammaB-gammaS}
{\rm pv} \int_{0}^{L_{B}} \gamma_{B}(e') \Imag \Bigg[ \frac{z_{B,e}(e)}{2L\ic} \cot\Big(\frac{z_{B}(e)-z_{B}(e') }{L/\pi}\Big)\Bigg] \d e'
 = {\rm RHS}_{VB}(e)
 \end{equation}
 where 
\begin{align*}
 &\hspace{-10pt} {\rm RHS}_{VB}(e) \\
 =& 
 -\int_{0}^{L_{S}} \gamma_{S}(e') \Imag \Bigg[\frac{z_{B,e}(e)}{2L\ic} \cot\Big(\frac{z_{B}(e)-z_{S}(e') }{L/\pi}\Big)\Bigg] \d e'\\
&- \int_{\D_{F}} \Imag \Bigg[\frac{z_{B,e}(e)}{2L\ic} \cot\Big(\frac{z_{B}(e)-y }{L/\pi}\Big)\Bigg] \omega_{0} \d \mb y \\
 =& 
- \int_{0}^{L_{S}} \gamma_{S}(e') \Imag \Bigg[\frac{z_{B,e}(e)}{2L\ic} \cot\Big(\frac{z_{B}(e)-z_{S}(e') }{L/\pi}\Big)\Bigg] \d e'\\
&- \sum_{j=1}^{N_{v}}\gamma_{v,j} \Imag \Bigg[\frac{z_{B,e}(e)}{2L\ic} \cot\Big(\frac{z_{B}(e)-z_{v,j} }{L/\pi}\Big) \Bigg] \\
&-\frac{ \omega_{0}}{4\pi}\int_{0}^{L_{S}} \ln\Big(\cosh \Imag \frac{z_{B}(e)-z_{S}(e')}{L/(2\pi)} -\cos \Real \frac{z_{B}(e)-z_{S}(e')}{L/(2\pi)} \Big) \Imag \Bigg[z_{B,e}(e)\overline{z_{S,e}(e')}\Bigg] \d e'\\
&+\frac{ \omega_{0}}{4\pi}\int_{0}^{L_{B}} \ln\Big(\cosh \Imag \frac{z_{B}(e)-z_{B}(e')}{L/(2\pi)} -\cos \Real \frac{z_{B}(e)-z_{B}(e')}{L/(2\pi)} \Big) \Imag \Bigg[z_{B,e}(e)\overline{z_{B,e}(e')}\Bigg] \d e'\,,
\end{align*}
together with the circulation assumptions:
\begin{equation*}
 \int_{0}^{L_{B}} \gamma_{B}(e') \d e' = - \gamma\,.
\end{equation*}

This problem is related to the vortex method in the case of an impermeable
boundary. The invertibility of this problem was studied in details in
\cite{ADL}, and the discretization will be described in Section~\ref{sec-discret}. 

This justifies Remark~\ref{rem:main-ellip} in the case of the vortex formulation.

\subsubsection{Initial $\mu_{S,0}$ for the dipole formulation}\label{sec-muS0}

Regarding the dipole formulation, we will often have to construct
$\mu_{S,0}$ knowing $g=\mb u_{F}\cdot \mb n\vert_{\Gamma_{S}}\,$.
As usual, $z_{B}\,,z_{S}\,, \gamma\,, \omega_{0}$ and
$(\gamma_{v,j}\,,z_{v,j})_{j=1,\dots, N_{v}}\,$ are initially given. 

The first step is to construct $\mb u_{\omega,\gamma}$ if
$(\omega\,,\gamma)\neq (0\,,0)\,$ (see Section~\ref{sec-uog}).
Given an expression of $\widehat{\mb u_{\omega,\gamma}}\, ,$
we are looking first for $\tilde\gamma_{S},\tilde\gamma_{B}$ such that the
associated $\mb u_{R}$ \eqref{BS-uR} solves the elliptic problem coming
from \eqref{ellip1}, \eqref{u-fgamma-1} and \eqref{u-fgamma-2}, i.e. we
consider the unique solution of 
\begin{multline*}
{\rm pv}\int_{0}^{L_{S}} \tilde\gamma_{S}(e') \Imag \Bigg[\frac{z_{S,e}(e)}{2L\ic} \cot\Big(\frac{z_{S}(e)-z_{S}(e') }{L/\pi}\Big)\Bigg] \d e'\\
+
 \int_{0}^{L_{B}} \tilde\gamma_{B}(e') \Imag \Bigg[ \frac{z_{S,e}(e)}{2L\ic} \cot\Big(\frac{z_{S}(e)-z_{B}(e') }{L/\pi}\Big)\Bigg] \d e'
 = -g |z_{S}(e)| - \Imag \Big[z_{S,e}(e) \widehat{\mb u_{\omega,\gamma}}(z_{S}(e))\Big] 
 \end{multline*}
 and
 \begin{multline*}
\int_{0}^{L_{S}} \tilde\gamma_{S}(e') \Imag \Bigg[\frac{z_{B,e}(e)}{2L\ic} \cot\Big(\frac{z_{B}(e)-z_{S}(e') }{L/\pi}\Big)\Bigg] \d e'\\
+
{\rm pv} \int_{0}^{L_{B}} \tilde\gamma_{B}(e') \Imag \Bigg[ \frac{z_{B,e}(e)}{2L\ic} \cot\Big(\frac{z_{B}(e)-z_{B}(e') }{L/\pi}\Big)\Bigg] \d e'
 = 0
 \end{multline*}
 together with the circulation assumptions
\begin{equation*}
 \int_{0}^{L_{S}} \tilde\gamma_{S}(e') \d e' = \int_{0}^{L_{B}}\tilde \gamma_{B}(e') \d e' = 0\,.
\end{equation*}

Of course, if we have chosen the third construction of $\widehat{\mb u_{\omega,\gamma}}\,$, then we can replace $\Imag \Big[z_{S,e}(e) \widehat{\mb u_{\omega,\gamma}}(z_{S}(e))\Big] $ by zero.

Finally, we set the initial value of $\mu_{S}$ as the anti-derivative of $\tilde \gamma_{S}$ with zero mean value
\[
\mu_{S,0}(e)= \int_{0}^e \tilde \gamma_{S}(e')\d e' -\frac1{L_{S}}\int_{0}^{L_{S}} \int_{0}^e \tilde \gamma_{S}(e')\d e'\d e \,.
\]

\subsubsection{Time dependent $\mu_{B}$ for the dipole formulation}\label{sec-muB}

When the time evolves, we need, for the dipole formulation, to construct
$\mu_{B}$ from $z_{S}\,, \mu_{S}\,$, again from the initial data
$z_{B}\,$. This problem is very simple, as we know that the potential obtained from $\mu_{B}$ and
$\mu_{S}$ (see \eqref{phi-int}) satisfies $\phi_{B}=0$ in $\mathcal{D}_{B}$. In particular, the limit of the potential (see Appendix~\ref{app-cot}) by below $\Gamma_{B}$ vanishes, which reads 
\begin{multline}\label{muB-muS}
\frac12 \mu_{B}(e) + \int_{0}^{L_{B}} \mu_{B}(e')\Real \Bigg[\frac1{2L\ic} \cot\Big(\frac{z_{B}(e)-z_{B}(e') }{L/\pi}\Big) z_{B,e}(e') \Bigg]\d e'\\
=-\int_{0}^{L_{S}} \mu_{S}(e') \Real \Bigg[\frac1{2L\ic} \cot\Big(\frac{z_{B}(e)-z_{S}(e') }{L/\pi}\Big) z_{S,e}(e') \Bigg]\d e'
\end{multline}
where the function in the left hand side integral is extended for $e=e'$ by
\[
-\mu_{B}(e)\Real\Bigg[\frac1{4\pi\ic} \frac{z_{B,ee}(e)}{z_{B,e}(e)} \Bigg]\,. 
\]

By uniqueness of $\mu_{B}$ satisfying such an equation, we state that it is enough to solve it for $\mu_{S}$ given.
This operator is different than the one in \eqref{gammaB-gammaS} and corresponds
to the operator $A^*$ in \cite{ADL}. We will highlight in
Section~\ref{sec-discret}
that it is possible to interpret this problem as a small
perturbation of the identity, and its inverse will be obtained in the form
of a Neumann series.

This justifies Remark~\ref{rem:main-ellip} in the case of the dipole formulation.

\begin{remark}\label{rem-phiB}
This problem is easily stated as we are looking for $\mu_{B}$ such
that $\phi_{B}=0$ below the bottom, which is possible only if we have
constructed $\mb u_{\omega,\gamma}$ tangent to the bottom
\eqref{u-fgamma-2}. 
\end{remark}

\subsection{Displacement of the free surface}

For the vortex formulation, we already have $\gamma_{S}$ and
$\gamma_{B}\,$, whereas for the dipole formulation we simply construct them
from $\mu_{S}$ and $\mu_{B}$: $\tilde\gamma_{B}=\partial_{e} \mu_{B}(e)$
and $\tilde\gamma_{S}=\partial_{e} \mu_{S}(e)\,$. 

With $\gamma_{S}$ and $\gamma_{B}\,$, it is now possible from \eqref{BS-u}
to compute the velocity anywhere. We recall that the tangential velocity is
discontinuous at the water-air interface.
We can thus assume that the free surface is moving as 
\[
\partial_{t} \mb z_{S}(t,e) = \alpha \mb u_{F}(t,z_{F}(e)) + (1-\alpha) \mb u_{A}(t,z_{F}(e))
\]
with a parameter $\alpha\in [0,1]\,$. By continuity of the
normal component, the evolution of the free surface does not depend on the
choice of $\alpha\,.$ It is however an interesting parameter from a
numerical point of view. It controls the distribution of points on the
interface. A choice of $\alpha$ can, for example, allow us to vary the
resolution in space. In practice, we found that
$\alpha=1$ efficiently concentrates computational points near the tip of a breaking wave.

The limit formulae of Appendix~\ref{app-cot} give
\begin{equation}
\begin{aligned}
 &\hspace{-10pt}\partial_{t} \overline{z_{S}(t,e)} \\
 =& \int_{0}^{L_{S}} \frac{\gamma_{S}(e')z_{S,e}(e) -\gamma_{S}(e)z_{S,e}(e') }{z_{S,e}(e)}\frac1{2L\ic} \cot\Big(\frac{z_{S}(e)-z_{S}(e') }{L/\pi}\Big) \d e' \\
&+\int_{0}^{L_{B}} \gamma_{B}(e') \frac1{2L\ic} \cot\Big(\frac{z_{S}(e)-z_{B}(e') }{L/\pi}\Big) \d e'\\
&+\frac{2\alpha-1}2 \frac{\gamma_{S}(e)}{z_{S,e}(e)} 
+ \sum_{j=1}^{N_{v}} \frac{\gamma_{v,j}}{2L\ic} \cot\Big(\frac{z_{S}(e)-z_{v,j} }{L/\pi}\Big) \\
&+\frac{ \omega_{0}}{4\pi}\int_{0}^{L_{S}} \ln\Big(\cosh \Imag \frac{z_{S}(e)-z_{S}(e')}{L/(2\pi)} -\cos \Real \frac{z_{S}(e)-z_{S}(e')}{L/(2\pi)} \Big) \overline{z_{S,e}(e')} \d e'\\
&-\frac{ \omega_{0}}{4\pi}\int_{0}^{L_{B}} \ln\Big(\cosh \Imag \frac{z_{S}(e)-z_{B}(e')}{L/(2\pi)} -\cos \Real \frac{z_{S}(e)-z_{B}(e')}{L/(2\pi)} \Big) \overline{z_{B,e}(e')} \d e'\,.
\end{aligned}
\label{eq:dtz1}
\end{equation} 
This formula highlights that the dependence on $\alpha$ appears only through
a tangent vector field
$\frac{\gamma_{S}(e)}{|z_{S,e}(e)|} \overline{z_{S,e}(e)}\,$. 

The displacement of the point vortices is an obvious application of this
formula.
For all $i=1,\dots,N_{v}\, ,$
\begin{align*}
\partial_{t} \overline{z_{v,i}(t)} 
 =& \int_{0}^{L_{S}} \gamma_{S}(e') \frac1{2L\ic} \cot\Big(\frac{z_{v,i}-z_{S}(e') }{L/\pi}\Big) \d e' 
+\int_{0}^{L_{B}} \gamma_{B}(e') \frac1{2L\ic} \cot\Big(\frac{z_{v,i}-z_{B}(e') }{L/\pi}\Big) \d e'\\
&+ \sum_{j=1,\ j\neq i}^{N_{v}} \frac{\gamma_{v,j}}{2L\ic} \cot\Big(\frac{z_{v,i}-z_{v,j} }{L/\pi}\Big) \\
&+\frac{ \omega_{0}}{4\pi}\int_{0}^{L_{S}} \ln\Big(\cosh \Imag \frac{z_{v,i}-z_{S}(e')}{L/(2\pi)} -\cos \Real \frac{z_{v,i}-z_{S}(e')}{L/(2\pi)} \Big) \overline{z_{S,e}(e')} \d e'\\
&-\frac{ \omega_{0}}{4\pi}\int_{0}^{L_{B}} \ln\Big(\cosh \Imag \frac{z_{v,i}-z_{B}(e')}{L/(2\pi)} -\cos \Real \frac{z_{v,i}-z_{B}(e')}{L/(2\pi)} \Big) \overline{z_{B,e}(e')} \d e'\,.
\end{align*}

In the case of the dipole formulation, we simply use \eqref{BS-uR} to get
$\mb u_{R}$ everywhere. We thus need the expression of
$\mb u_{\omega,\gamma}\,$, see Section~\ref{sec-muS0} for these formulae
depending on the physical setting. In particular, we notice that often, we
do not have the velocity in the air, hence we should only consider the case 
 \[\partial_{t} \mb z_{S}(t,e) = \mb u_{F}(t,z_{F}(e)) 
\]
which makes sense for the single-fluid water-waves equations.
We thus write
\begin{align*}
\partial_{t} \overline{z_{S}(t,e)} 
 =& \int_{0}^{L_{S}} \frac{\tilde\gamma_{S}(e')z_{S,e}(e) -\tilde\gamma_{S}(e)z_{S,e}(e') }{z_{S,e}(e)}\frac1{2L\ic} \cot\Big(\frac{z_{S}(e)-z_{S}(e') }{L/\pi}\Big) \d e' \\
&+\int_{0}^{L_{B}} \tilde\gamma_{B}(e') \frac1{2L\ic} \cot\Big(\frac{z_{S}(e)-z_{B}(e') }{L/\pi}\Big) \d e'
+\frac{1}2 \frac{\tilde \gamma_{S}(e)}{z_{S,e}(e)} 
+ \widehat{\mb u_{\omega,\gamma}}(z_{S}(e))\,,
\end{align*}
where the last formula of $\widehat{\mb u_{\omega,\gamma}}$ in Section
\ref{sec-muS0} has to be desingularized when $x=z_{S}(e)$ as we did in the
previous formula \eqref{eq:dtz1}. 

It is worth stressing that some freedom is left on the choice of $\alpha\,,$
which will not affect the shape of the solution, but only the tangential
distribution of points on the interface. 

In the case of the bi-fluid problem, we have constructed $\mb u_{\omega,\gamma}$
tangent to free surface, it is thus more natural to include the parameter $\alpha$
\begin{equation} 
\begin{aligned}
 &\hspace{-10pt}\partial_{t} \overline{z_{S}(t,e)} \\
 =& \int_{0}^{L_{S}} \frac{\tilde\gamma_{S}(e')z_{S,e}(e) -\tilde\gamma_{S}(e)z_{S,e}(e') }{z_{S,e}(e)}\frac1{2L\ic} \cot\Big(\frac{z_{S}(e)-z_{S}(e') }{L/\pi}\Big) \d e' \\
&+\int_{0}^{L_{B}} \tilde\gamma_{B}(e') \frac1{2L\ic} \cot\Big(\frac{z_{S}(e)-z_{B}(e') }{L/\pi}\Big) \d e'
+\frac{2\alpha-1}2 \frac{\tilde \gamma_{S}(e)}{z_{S,e}(e)} 
+\alpha \widehat{\mb u_{\omega,\gamma}}(z_{S}(e))\,.
\end{aligned}
\label{eq:dtz2}
\end{equation}

\subsection{Bernoulli equation and dipole formulation}\label{sec-bernoulli}

Writing $\mb u = \mb u_{\omega,\gamma}+\nabla \phi_{F}\,$, the Euler
equations take the form
\[
\nabla\Big[ \partial_{t}\phi_{F} +\frac12|\mb u|^2\Big] + \partial_{t}\mb u_{\omega,\gamma} +(\curl \mb u) \mb u^\perp=-\nabla \Big[\frac{p_{F}}{\rho_{F}}+gx_{2}\Big]
\]
where $\rho_{F}$ is the density of the fluid and $g$ is the gravity
acceleration. Hence in the neighborhood of the free surface where the
vorticity is contant, the Euler equations can be reduced to the modified
Bernoulli equation 
\[
\partial_t \phi_{F}+\frac{1}{2}\vert \nabla \phi_{F} + \mb u_{\omega,\gamma} \vert^2 + \phi_{t,\omega,\gamma} -\omega_{0} (\psi_{\omega,\gamma} + \tilde\psi_{F}) =-\frac{p_{F}}{\rho_{F}}-g x_{2} \,,
\]
with $\psi_{\omega,\gamma}$ the stream function of $\mb u_{\omega,\gamma}$
and $\phi_{t,\omega,\gamma}$ the potential of $\partial_{t} \mb
u_{\omega,\gamma}$\footnote{Such a potential exists in the neighborhood of
the free surface, because $\curl\partial_{t} \mb
u_{\omega,\gamma}=\partial_{t}\omega_{0}=0$ and by the conservation of 
circulation.}. 
\begin{remark}
For all the examples given above, the difficult component to
express is the stream function associated to the constant part. This is because it
is not obvious how to express $\int_{\D_{F}} G(\mb x-\mb y) \d \mb y$ as
an integral over the boundaries. Hence, it is easier to assume
$\omega_{0}=0$ which removes the presence of $\psi_{\omega,\gamma}\,$. 
 
We should also observe that the potential $\phi_{t,\omega,\gamma}$ is even
more complicated to obtain. Therefore, in the sequel, we only
consider stationary $\mb u_{\omega,\gamma}$ where we can forget
$\partial_{t}\mb u_{\omega,\gamma}$ and $\phi_{t,\omega,\gamma}\,$. Of
course, this implies that we also assume $\gamma_{v,j}=0\,$, hence
$\omega=0$ and we can replace $\mb u_{\omega,\gamma}$ with $\mb u_{\gamma}\,$. 

From now on, we will restrict our attention to the following two situations for the dipole formulation:
\begin{itemize}
 \item the bi-fluid and the single fluid water-waves equations without circulation and vorticity, (i.e. $\mb u_{\gamma}=0$);
 \item the single fluid water-waves equations with circulation and without vorticity, (i.e. $\mb u_{\gamma}=\gamma \mb e_{1}/L$ for the flat bottom and $\mb u_{\gamma}=\mb H\,$, initially constructed for other bottom).
\end{itemize}
\end{remark}
 
For the bi-fluid water-waves equations, we also consider the Bernoulli equation in the air
\[
\partial_t \phi_{A}+\frac{1}{2}\vert \nabla \phi_{A} \vert^2=-\frac{p_{A}}{\rho_{A}}-g x_{2} \,.
\]
For the single-fluid water-waves equations, the density of the air is neglected and the pressure in the air is constant.

It is useful to count here boundary conditions for the fluid domain
$\mathcal{D}_{F} \, .$ At the bottom boundary $\Gamma_{B}$ the normal
component of the flow vanishes, which provides the needed boundary
condition.
At the top boundary $\Gamma_{S}$, in the bi-fluid problem both the normal
component of velocity and the pressure are continuous with that is the air domain
$\mathcal{D}_{A}$ above, across the boundary $\Gamma_{S}\, .$
These two continuity relations are thus enough to close the fluid system in
$\mathcal{D}_{F}$ (morally the number of jump conditions needed at the
boundary between two domains is $n_1+n_2$, where $n_1$ and $n_2$ are the
numbers of outer conditions required for the PDE solution in domains 1 and
2 respectively).
The single-fluid water-waves problem, corresponds to the limit of a vanishing density
for the air. The two jump conditions on $\Gamma_{S}$ are then replaced by a single
boundary condition on pressure $p=0$, which again is enough to close the
Euler system in $\mathcal{D}_{F}$.

In any case, we need to relate $p_{A}$ and $p_{F}$.
In the presence
of surface tension, the pressure is not continuous and a pressure jump is
achieved, which is directly related to the surface curvature $\kappa$
\begin{equation}\label{SurfaceTension}
[ p{\bf n} ]=(p_{A}-p_{F})\mb n=-\sigma \kappa{\bf n} \, .
\end{equation}
The surface tension coefficient $\sigma$ is here related to capillary effects (e.g., \cite[$\S $9.1.2]{Lannes-Livre}).

Setting the atmospheric pressure to $0$ for the single-fluid equations, we
consider the limit of the Bernoulli equation at the interface
\begin{multline*}
 \partial_t \Big( \phi_{F}(z_{S}(e)) \Big) - \partial_{t}\mb z_{S}(e) \cdot (\nabla\phi_{F})(z_{S}(e)) +\frac{1}{2}\vert \nabla \phi_{F}+ \mb u_{\gamma} \vert^2(z_{S}(e)) \\
 =-\frac{\sigma }{\rho_{F}}\kappa(z_{S}(e)) -g z_{S,2}(e)\,, 
\end{multline*}
hence we get 
\begin{align}
\frac12\partial_{t}\mu_{S}(e) =&\frac12 \partial_t \Big( \phi_{F}(z_{S}(e)) \Big)- \frac12\partial_t \Big( \phi_{A}(z_{S}(e)) \Big)= \partial_t \Big( \phi_{F}(z_{S}(e)) \Big)-\frac12 \partial_t \Phi_{S}(e) \nonumber\\
=&- \frac12\partial_t \Phi_{S}(e) + \partial_{t}\mb z_{S}(e) \cdot( \nabla\phi_{F})(z_{S}(e)) -\frac12 \vert \nabla \phi_{F}+ \mb u_{\gamma} \vert^2(z_{S}(e))\label{dtmupart0}\\
& -\frac{\sigma }{\rho_{F}}\kappa(z_{S}(e)) -g z_{S,2}(e) \,.\nonumber
\end{align}

For the bi-fluid water-waves equation without circulation, we need to consider
the Bernoulli equation in the air, hence performing the difference we get 
\begin{multline*}
\partial_t \mu_{S}(e) + \partial_{t}\mb z_{S}(e) \cdot (\nabla\phi_{A} - \nabla\phi_{F})(z_{S}(e)) +\frac{1}{2} \Big(\vert \nabla \phi_{F}\vert^2- \vert \nabla \phi_{A} \vert^2 \Big)(z_{S}(e))\\
 =\frac{\rho_{F}-\rho_{A}}{\rho_{F}\rho_{A}}p_{F}(z_{S}(e))- \frac{\sigma }{\rho_{A}}\kappa(z_{S}(e))\,, 
\end{multline*}
whereas the sum gives
\begin{multline*}
\partial_t \Phi_{S}(e) - \partial_{t}\mb z_{S}(e) \cdot (\nabla\phi_{A} + \nabla\phi_{F})(z_{S}(e)) +\frac{1}{2} \Big(\vert \nabla \phi_{F} \vert^2+ \vert \nabla \phi_{A} \vert^2 \Big)(z_{S}(e)) \\
=-\frac{\rho_{F}+\rho_{A}}{\rho_{F}\rho_{A}}p_{F}(z_{S}(e))+ \frac{\sigma }{\rho_{A}}\kappa(z_{S}(e)) -2g z_{S,2}(e)\,.
\end{multline*}
We remove the pressure by multiplying the second equation by the Atwood number
\begin{equation}\label{Atwood}
A_{tw} = \frac{\rho_{F}-\rho_{A}}{\rho_{F}+\rho_{A}} 
\end{equation}
and finally obtain 
\begin{align}
\frac12\partial_t \mu_{S}(e)
=& -\frac{A_{tw}}2 \partial_t \Phi_{S}(e) +\frac12 \partial_{t}\mb z_{S}(e) \cdot ( (A_{tw}+1)\nabla\phi_{F} +(A_{tw}-1) \nabla\phi_{A})(z_{S}(e))\nonumber\\
&-\frac{1}{4} \Big((A_{tw}+1) \vert \nabla \phi_{F}
\vert^2+ (A_{tw}-1) \vert \nabla \phi_{A} \vert^2 \Big)(z_{S}(e))\nonumber \\
& -\frac{(A_{tw}+1)\sigma }{2\rho_{F}}\kappa(z_{S}(e)) -gA_{tw} z_{S,2}(e)
\label{dtmupart1}
\end{align}
because
$\frac{A_{tw}-1}{\rho_{A}}=\frac{-2}{\rho_{F}+\rho_{A}}=-\frac{A_{tw}+1}{\rho_{F}}\, .$
Even if the derivation differs, it is worth stressing that the
single-fluid water-waves equations without circulation are recovered when setting $A_{tw}=1\,$
in the above equation.

We have already derived the expression for $\partial_{t} z_{S}(e)$ and in
the similar way for $\widehat{\nabla \phi_{F}} (z_{S}(e))$ and
$\widehat{\nabla \phi_{A}} (z_{S}(e))\,$, our next step is to compute
$\partial_t \Phi_{S}(e)\,$.
From \eqref{phi-mu} 
\begin{align}
 &\hspace{-10pt}\frac{\partial_{t} \Phi_{S}(e)}2\nonumber\\
 = &\int_{0}^{L_{S}} ( \partial_{t}\mu_{S}(e') - \partial_{t}\mu_{S}(e))\Real \Bigg[\frac1{2L\ic} \cot\Big(\frac{z_{S}(e)-z_{S}(e') }{L/\pi}\Big) z_{S,e}(e') \Bigg]\d e'\nonumber\\
&+\int_{0}^{L_{B}} \partial_{t}\mu_{B}(e')\Real \Bigg[\frac1{2L\ic} \cot\Big(\frac{z_{S}(e)-z_{B}(e') }{L/\pi}\Big) z_{B,e}(e') \Bigg]\d e'\nonumber\\
&+\int_{0}^{L_{S}} ( \mu_{S}(e) - \mu_{S}(e'))\Real \Bigg[\frac{\pi}{2L^2\ic} \sin^{-2}\Big(\frac{z_{S}(e)-z_{S}(e') }{L/\pi}\Big) (\partial_{t}z_{S}(e)-\partial_{t}z_{S}(e') )z_{S,e}(e') \Bigg]\d e'\nonumber\\
&+\int_{0}^{L_{S}} ( \mu_{S}(e') - \mu_{S}(e))\Real \Bigg[\frac1{2L\ic} \cot\Big(\frac{z_{S}(e)-z_{S}(e') }{L/\pi}\Big) \partial_{t}z_{S,e}(e') \Bigg]\d e'\nonumber\\
&-\int_{0}^{L_{B}} \mu_{B}(e')\Real \Bigg[\frac{\pi}{2L^2\ic} \sin^{-2}\Big(\frac{z_{S}(e)-z_{B}(e') }{L/\pi}\Big)\partial_{t}z_{S}(e) z_{B,e}(e') \Bigg]\d e'
 \label{dtmupart2}
\end{align}
where the third right hand side integral concerns a continuous function whose value for $e'=e$ is 
\[
\mu_{S,e}(e) \Real \Bigg[\frac{1}{2\pi\ic} \frac{\partial_{t}z_{S,e}(e)}{z_{S,e}(e)} \Bigg] 
\]
whereas in the fourth integral the extension is
\[
-\mu_{S,e}(e) \Real \Bigg[\frac{1}{2\pi\ic} \frac{\partial_{t}z_{S,e}(e)}{z_{S,e}(e)} \Bigg]
\]
which is exactly opposite to the first one, and then can be omitted.

Finally, substituting \eqref{dtmupart2} into \eqref{dtmupart0} or \eqref{dtmupart1}, yields an
equation of the form
\[
A_{S}^* [\partial_{t} \mu_{S}](e)+ C_{D} [\partial_{t} \mu_{B}](e) =
G_{D,1}(e) \, ,
\]
where the operator $A_{S}$ is the same kind of operator as in
\eqref{muB-muS} (see Section~\ref{sec-discret} for an explanation on how to
discretize such an operator, the discrete expressions of $A_{S}\,$, $C_{D}$
and $G_{D,1}$ are given in Appendix~\ref{app-dipole}). 

As the above equation involves $\partial_{t} \mu_{B}\,$, we derive another
equation differentiating \eqref{muB-muS} with respect to time: 
\begin{align*}
& \int_{0}^{L_{S}} \partial_{t} \mu_{S}(e') \Real \Bigg[\frac{z_{S,e}(e')}{2L\ic} \cot\Big(\frac{z_{B}(e)-z_{S}(e') }{L/\pi}\Big)\Bigg] \d e' \\
&+\frac12 \partial_{t} \mu_{B}(e)+ \int_{0}^{L_{B}} \partial_{t} \mu_{B}(e') \Real \Bigg[ \frac{z_{B,e}(e')}{2L\ic} \cot\Big(\frac{z_{B}(e)-z_{B}(e') }{L/\pi}\Big)\Bigg] \d e'\\
&\hspace{3cm} = -\int_{0}^{L_{S}} \mu_{S}(e') \Real \Bigg[\frac{\partial_{t}z_{S,e}(e')}{2L\ic} \cot\Big(\frac{z_{B}(e)-z_{S}(e') }{L/\pi}\Big)\Bigg] \d e' \\
&\hspace{3.3cm} -\int_{0}^{L_{S}} \mu_{S}(e') \Real \Bigg[\frac{\pi z_{S,e}(e') \partial_{t} z_{S}(e')}{2L^2\ic} \sin^{-2}\Big(\frac{z_{B}(e)-z_{S}(e') }{L/\pi}\Big)\Bigg] \d e' 
\end{align*}
which gives an equation of the form
\[
D_{D} [\partial_{t} \mu_{S}](e)+ A_{B}^* [\partial_{t} \mu_{B}](e) = G_{D,2}(e)
\]
where $A_{B}^*$ is precisely the same operator as on the left hand side of
\eqref{muB-muS}, which was already inverted. From this equation, we have
\[
\partial_{t} \mu_{B} =A_{B}^{*-1} \Big[ G_{D,2}-D_{D} [\partial_{t} \mu_{S}]\Big]
\]
which means that $\partial_{t} \mu_{S}$ will be obtain by solving
\begin{equation}\label{inv_muSt}
\Big(A_{S}^*- C_{D}A_{B}^{*-1} D_{D} \Big)[\partial_{t} \mu_{S}]= G_{D,1} - C_{D}A_{B}^{*-1} G_{D,2}\,. 
\end{equation}
We will discuss in Section~\ref{sec-discret} that $A_{S}^*-
C_{D}A_{B}^{*-1} D_{D}$ can be seen as a perturbation of a simple matrix,
though not the identity. The discrete version is also given in
Appendix~\ref{app-dipole}. 

This equation, together with \eqref{eq:dtz2} corresponds to the vortex formulation
and the first part of Theorem~\ref{theo:main-dyn}.

\subsection{Euler equation and vortex formulation}\label{sec-Euler}

In \cite{Baker82}, the authors differentiate the equation for 
$\partial_{t}\mu$ \eqref{dtmupart1} to get the equation for
$\partial_{t} \gamma\, .$ This is difficult to justify because the kernels in the
integrals are singular.
Let us also stress that such a derivation yields a vortex formation which
should only be used without circulation, i.e. with zero mean value for
$\gamma_{S,0}\,$.

Alternatively, we write the Euler equations in $\D_{F}$ 
\[
 \partial_{t} \mb u_{F} + (\mb u_{F}\cdot \nabla) \mb u_{F} = \frac{-\nabla
 p_{F}}{\rho_{F}} -g \mb e_{2} \, .
 \]
For the single-fluid water-waves equation, we have $\nabla p_{A}=0\, ,$
whereas for the bi-fluid water-waves equations we also write the Euler equations in $\D_{A}$
 \[
 \partial_{t} \mb u_{A} + (\mb u_{A}\cdot \nabla) \mb u_{A} = \frac{-\nabla p_{A}}{\rho_{A}} -g \mb e_{2}\,.
\]

As for the dipole formulation, we need to relate the pressures on both sides of
the interface using the continuity of the normal component of the stress
tensor at the interface \eqref{SurfaceTension}.
This implies by differentiating with respect to $e$
\[
\mb z_{S,e}(e) \cdot (\nabla p_{A}-\nabla p_{F})(z_{S}(e)) =-\sigma \frac{\d}{\d e} \Big(\kappa(z_{S}(e))\Big)\,.
\]
So we need to consider the limit of the tangential part of the Euler equations at the interface. 

For the single-fluid water-waves equations, one can simply replace $\mb z_{S,e}(e)\cdot \nabla p_{F}(z_{S}(e))$ by $\sigma \frac{\d}{\d e} \Big(\kappa(z_{S}(e))\Big)$ and obtain
\begin{multline*}
 \partial_{t} \Big( \mb u_{F} (z_{S}(e)) \cdot \mb z_{S,e}(e) \Big) + \Big[\Big((\mb u_{F}(z_{S}(e))- \partial_{t} \mb z_{S}(e))\cdot \nabla\Big) \mb u_{F} (z_{S}(e))\Big]\cdot \mb z_{S,e}(e) \\- \mb u_{F} (z_{S}(e)) \cdot \partial_{t}\mb z_{S,e}(e)
 =- \frac{\sigma}{\rho_{F}} \frac{\d}{\d e} \Big(\kappa(z_{S}(e))\Big) -g
 z_{S,e,2} \, .
\end{multline*}
In order to introduce $\partial_{t}\gamma_{S}\, ,$ we use
\begin{align*}
 &\hspace{-10pt} \Psi_{S}(e) := ( \mb u_{F} + \mb u_{A}) (z_{S}(e)) \cdot \mb z_{S,e}(e)\\
 =& {\rm pv}\int_{0}^{L_{S}}\gamma_{S}(e') \Real\Bigg[\frac{z_{S,e}(e)}{L\ic } \cot\Big(\frac{z_{S}(e)-z_{S}(e')}{L/\pi}\Big) \Bigg] \d e'\\
&+ \int_{0}^{L_{B}}\gamma_{B}(e') \Real\Bigg[\frac{z_{S,e}(e)}{L\ic } \cot\Big(\frac{z_{S}(e)-z_{B}(e')}{L/\pi}\Big) \Bigg] \d e'\\
&+ \sum_{j=1}^{N_{v}} \gamma_{v,j} \Real\Bigg[\frac{z_{S,e}(e)}{L\ic } \cot\Big(\frac{z_{S}(e)-z_{v,j} }{L/\pi}\Big)\Bigg] \\
&+\frac{ \omega_{0}}{2\pi}\int_{0}^{L_{S}} \ln\Big(\cosh \Imag \frac{z_{S}(e)-z_{S}(e')}{L/(2\pi)} -\cos \Real \frac{z_{S}(e)-z_{S}(e')}{L/(2\pi)} \Big) \Real\Big[z_{S,e}(e)\overline{z_{S,e}(e')}\Big] \d e'\\
&-\frac{ \omega_{0}}{2\pi}\int_{0}^{L_{B}} \ln\Big(\cosh \Imag \frac{z_{S}(e)-z_{B}(e')}{L/(2\pi)} -\cos \Real \frac{z_{S}(e)-z_{B}(e')}{L/(2\pi)} \Big) \Real\Big[z_{S,e}(e)\overline{z_{B,e}(e')}\Big] \d e'
\end{align*}
so by the definition of $\gamma_{S}$ \eqref{def-gamma}, we write
\begin{align}
\frac12\partial_{t}\gamma_{S}(e) =& \partial_{t} \Big( \mb u_{F} (z_{S}(e)) \cdot \mb z_{S,e}(e) \Big) -\frac12 \partial_t \Psi_{S}(e) \nonumber\\
=&-\frac12 \partial_t \Psi_{S}(e) 
- \Big[\Big((\mb u_{F}(z_{S}(e))- \partial_{t} \mb z_{S}(e))\cdot \nabla\Big) \mb u_{F} (z_{S}(e))\Big]\cdot \mb z_{S,e}(e)\label{dtgammapart0} \\
&+ \mb u_{F} (z_{S}(e)) \cdot \partial_{t}\mb z_{S,e}(e)
 - \frac{\sigma}{\rho_{F}} \frac{\d}{\d e} \Big(\kappa(z_{S}(e))\Big) -g z_{S,e,2} \,.\nonumber
\end{align}

For the bi-fluid formulation, we proceed as for the dipole formulation, i.e.
\begin{itemize}
 \item we compute from both Euler equations 
 \[\partial_{t}\gamma_{S}(e)= \partial_{t} \Big( \mb u_{F} (z_{S}(e)) \cdot \mb z_{S,e}(e) \Big) -\partial_{t} \Big( \mb u_{A} (z_{S}(e)) \cdot \mb z_{S,e}(e) \Big)\, ;\]
 \item we express in the same way 
 \[\partial_{t}\Psi_{S}(e)= \partial_{t} \Big( \mb u_{F} (z_{S}(e)) \cdot \mb z_{S,e}(e) \Big) +\partial_{t} \Big( \mb u_{A} (z_{S}(e)) \cdot \mb z_{S,e}(e) \Big)\, ;\]
 \item we replace $\mb z_{S,e}(e) \cdot \nabla p_{A}(z_{S}(e))$ by $\mb z_{S,e}(e) \cdot \nabla p_{F}(z_{S}(e)) -\sigma \frac{\d}{\d e} \Big(\kappa(z_{S}(e))\Big)$\, ;
 \item we remove the pressure term by multiplying the second equation by
 the Atwood number \eqref{Atwood}, and by adding the two equations, and we use $\frac{A_{tw}-1}{\rho_{A}}=\frac{-2}{\rho_{F}+\rho_{A}}=-\frac{A_{tw}+1}{\rho_{F}}\,$.
\end{itemize}
In the end, we get a slightly modified equation for $\partial_{t}\gamma_{S}(e)$:
\begin{align}
\frac12\partial_{t}\gamma_{S}(e)
=&-\frac{A_{tw}}2 \partial_t \Psi_{S}(e) 
-\frac{1+A_{tw}}2 \Big[\Big((\mb u_{F}(z_{S}(e))- \partial_{t} \mb z_{S}(e))\cdot \nabla\Big) \mb u_{F} (z_{S}(e))\Big]\cdot \mb z_{S,e}(e) \nonumber \\
&+\frac{1-A_{tw}}2 \Big[\Big((\mb u_{A}(z_{S}(e))- \partial_{t} \mb z_{S}(e))\cdot \nabla\Big) \mb u_{A} (z_{S}(e))\Big]\cdot \mb z_{S,e}(e)
-g A_{tw} z_{S,e,2} \label{dtgammapart1} \\
&+\Big(\frac{1+A_{tw}}2 \mb u_{F}-\frac{1-A_{tw}}2 \mb u_{A} \Big)(z_{S}(e)) \cdot \partial_{t}\mb z_{S,e}(e)
 - \frac{(1+A_{tw})\sigma}{2\rho_{F}} \frac{\d}{\d e} \Big(\kappa(z_{S}(e))\Big) \nonumber
\end{align}
which coincides with the first one when we set $A_{tw}=1\,$.

We may simplify the second and third right hand side terms by 
\begin{align*}
 \Big((\mb u_{F}(z_{S})- \partial_{t} \mb z_{S})\cdot \nabla\Big) \mb u_{F} (z_{S}) 
 &= (1-\alpha)\frac{\gamma_{S}}{|z_{S,e}|^2} \Big(\mb z_{S,e}\cdot \nabla\Big) \mb u_{F} (z_{S})\\
 & = (1-\alpha)\frac{\gamma_{S}}{|z_{S,e}|^2} \partial_{e} \Big(\mb u_{F} (z_{S})\Big)\, ,\\
 \Big((\mb u_{A}(z_{S})- \partial_{t} \mb z_{S})\cdot \nabla\Big) \mb u_{A} (z_{S}) 
 &= -\alpha\frac{\gamma_{S}}{|z_{S,e}|^2} \Big(\mb z_{S,e}\cdot \nabla\Big) \mb u_{A} (z_{S}) = -\alpha\frac{\gamma_{S}}{|z_{S,e}|^2} \partial_{e} \Big(\mb u_{A} (z_{S})\Big)\,.
\end{align*}

We then substitute the computation of $\partial_{t}\Psi_{S}$ (see
Appendix~\ref{app-vortex}) into \eqref{dtgammapart0} or
\eqref{dtgammapart1} to get an equation of the form
\[
A_{S} [\partial_{t} \gamma_{S}](e)+ C_{V} [\partial_{t} \gamma_{B}](e) = G_{V,1}(e)
\]
where $A_{S}$ is the adjoint operator of $A_{S}^*$ in the dipole formulation, see Section~\ref{sec-discret} for properties of these operators. 

Finally, we differentiate with respect to time \eqref{gammaB-gammaS} to get an equation of the form
\[
D_{V} [\partial_{t} \gamma_{S}](e)+ B_{B} [\partial_{t} \gamma_{B}](e) = G_{V,2}(e)
\]
where the conservation of circulation $\int \partial_{t} \gamma_{B}=0$ is
included in the last line of the discretization. This allows to obtain
$\partial_{t} \gamma_{S}$ by solving 
\begin{equation} 
\Big(A_{S}- C_{V}B_{B}^{-1} D_{V} \Big)[\partial_{t} \gamma_{S}]= G_{V,1} - C_{V}B_{B}^{-1} G_{V,2}\,.
\label{inv_gammaSt}
\end{equation} 
Again, the discrete forms are given in Appendix~\ref{app-vortex}.

This equation, together with \eqref{eq:dtz1} corresponds to the vortex formulation
and the last part of Theorem~\ref{theo:main-dyn}.

\begin{remark}\label{rem-FCM} 
 In \cite{ADL}, we have developed a method referred to as the ``fluid charge
 method''.
 In this method, after retrieving $\mb u_{\omega,\gamma}$
 we have written $\mb u_{R}=\nabla \phi_{F}$ and solved the Laplace problem
 with homogeneous Neumann boundary condition on $\Gamma_{B}$.
 This method relies on an extension of $\phi$ in ${\cal D}_{B}$ by
 continuity.
 We can establish that 
\begin{equation*}
 \widehat{\nabla \phi} (x)
= \int_{0}^{L_{S}} \sigma_{S}(e) \frac1{2L} \cot\Big(\frac{x-z_{S}(e) }{L/\pi}\Big) \d e 
+\int_{0}^{L_{B}} \sigma_{B}(e) \frac1{2L} \cot\Big(\frac{x-z_{B}(e) }{L/\pi}\Big) \d e\,,
\end{equation*}
where
\begin{align*}
 \sigma_{S}(e) &: =|z_{S,e}(e)|[\lim_{z\in \D_{A}\to z_{S}(e)} \partial_{n}\phi - \lim_{z\in \D_{F}\to z_{S}(e)} \partial_{n}\phi] \\
\sigma_{B}(e) & :=|z_{B,e}(e)| [\lim_{z\in \D_{F}\to z_{B}(e)} \partial_{n}\phi - \lim_{z\in \D_{B}\to z_{B}(e)} \partial_{n}\phi]=-|z_{B,e}(e)| [ \lim_{z\in \D_{B}\to z_{B}(e)} \partial_{n}\phi] \,.
\end{align*}
Indeed, in this formulation, the tangential part is continuous whereas the
normal part has a jump. Therefore, we can adapt Section~\ref{sec-gammaBmuB}
to find $\sigma_{B}$ such that $\mb u_{R}$ satisfies the impermeability
condition on $\Gamma_{B}$. Note that this problem is related to the
inversion of the operator $A$ in \cite{ADL}. Next, we can use the previous
formula to get the displacement of the free surface in terms of
$\sigma_{S}$. 

Unfortunately, in the case of water-waves, we do not have an equation for $\partial_{t}\sigma_{S}$. In
Sections~\ref{sec-bernoulli} and \ref{sec-Euler}, we have used the
continuity of the normal component of the stress tensor
\eqref{SurfaceTension}.
For the dipole formulation, this equation on
$p_{F}-p_{A}$ provides the connexion between the two Bernoulli equations and
allows to get an equation on $\partial_{t}\mu_{S}=\frac12 \partial_t \Big(
\phi_{F}(z_{S}(e)) \Big)- \frac12\partial_t \Big( \phi_{A}(z_{S}(e))
\Big)$.
For the vortex formulation, a differentiation along the free surface
of this relation on $p_{F}-p_{A}$ yields an equation on
$\mb z_{S,e}(e) \cdot (\nabla p_{A}-\nabla p_{F})(z_{S}(e))$.
This establishes a connexion between the tangential part of the two Euler
equations.
Here, it is important that $\gamma_{S}$ corresponds to the jump of the
tangential velocities. In the fluid charge method, $\sigma_{S}$ corresponds
to the jump of the normal velocities. To get an equation for
$\partial_{t}\sigma_{S} \, ,$ we would thus need a relation for
$\mb z_{S,e}(e)^\perp \cdot (\nabla p_{A}-\nabla p_{F})(z_{S}(e))$,
i.e. a sort of continuity of the normal derivative of the normal component
of the stress tensor, which has no physical meaning.
Note that, defining the velocity above the free surface such that the
normal component has a jump implies that it does not corresponds to the air velocity. 

Let us note that Baker in \cite[Equation (4.4)]{Baker82} 
considered either the vortex or dipole formulation for the free surface combined
with the fluid charge method at the bottom.
Because of the extension of $\psi$ above the fluid and $\phi$ below the
fluid domain, Proposition~\ref{prop-Green} cannot be applied to such a formulation.
For this reason, we have chosen in the previous sections to consider the same
formulation for the free surface and for the bottom.
\end{remark}

\subsection{The deep-water case}

It is easy to derive a deep-water formulation removing contributions
from the bottom. Following line by line the previous sections without the
presence of $\D_{B}$ and $\Gamma_{B}$, i.e. assuming that the fluid domain
is infinite in the vertical direction, we can get the following model
\begin{itemize}
 \item the dipole formulation for the bi-fluid water-waves equations
 without vorticity and circulation, where the velocity is given by
 $\gamma_{S}=\partial_{e}\mu_{S}$. The equation verifyed by
 $\partial_{t}\mu_{S}$ stays exactly \eqref{dtmupart1}, but dropping all
 terms involving $\mu_{B}$ in the expression \eqref{dtmupart2} for
 $\partial_{t}\Phi_{S}$. This yields 
 \[
A_{S}^* [\partial_{t} \mu_{S}](e) = G_{D,1}(e) \, ,
\]
where $G_{D,1}$ is giving in \eqref{GD1} where we remove $\mb u_{\gamma}$
and $\mu_{B}$; 
 \item the dipole formulation for the single-fluid water-waves equations
 with circulation and without vorticity, where $\mb
 u_{\gamma}=\frac\gamma L\mb e_{1}$ and
 $\tilde\gamma_{S}=\partial_{e}\mu_{S}$. The equation verified by
 $\partial_{t}\mu_{S}$ stays exactly \eqref{dtmupart0}, but again dropping
 all terms involving $\mu_{B}$ in the expression
 \eqref{dtmupart2} for $\partial_{t}\Phi_{S}$. This yields 
 \[
A_{S}^* [\partial_{t} \mu_{S}](e) = G_{D,1}(e) \, ,
\]
where $G_{D,1}$ is giving in \eqref{GD1} where we remove $\mu_{B}$ and
replace $A_{tw}$ by $1$ and $\mb u_{\gamma}$ by $\frac\gamma L\mb e_{1}\,$; 
\item the vortex formulation with circulation and where the vorticity is
 composed of point vortices (no constant part). The velocity is given by
 $\gamma_{S}$ and the equation verified by $\partial_{t}\gamma_{S}$ stays
 exactly \eqref{dtgammapart1}, but dropping all terms involving
 $\gamma_{B}$ in the expression for $\partial_{t}\Psi_{S}$, see
 Appendix~\ref{app-vortex}. This yields 
 \[
A_{S} [\partial_{t} \gamma_{S}](e) = G_{V,1}(e) \, .
\]
\end{itemize}

Let us note that the equations obtained in this case are very close to
\cite[Equations (2.14)-(2.17)]{Baker82}, but where we have used the
desingularization \eqref{phi-mu}, which allows us to justify the
derivatives and to handle only classical integrals. In this earlier
article, principal value integrals with $x/\sin^2 x$ singularities may be
a cause of numerical instabilities.

\section{Numerical discretization}\label{sec4}

\subsection{Time integration}

Whereas most earlier numerical studies on water waves breaking used
high-order Runge-Kutta integrators \cite{Wilkening21,Baker11,Scolan}, we preferred to
restrict our study to a second order in time, but symplectic integrator for harmonic oscillators.
We use the so-called Verlet integrator, which amounts to using a staggered
grid in time, and preserves the Hamiltonian structure of harmonic oscillators.

The governing equations take the form
\begin{equation*}
 \partial _t X = G(Z,X)\, , \qquad
 \partial _t Z = F(Z,X)\, , 
\end{equation*} 
where $X\equiv \gamma$ in the vortex formulation and $X\equiv \mu$ in the
dipole formulation. $F$ and $G$ denote here non-linear differential
operators.
These are discretized in the form
\begin{equation*}
 X^{n+1/2}-X^{n-1/2} = \Delta t \, G(Z^n,X^n)\, , \quad
 Z^{n+1}-Z^{n} = \Delta t \, F(Z^{n+1/2},X^{n+1/2})\, .
\end{equation*} 
The right-hand-side of the first equation involves $X^n$
which is not known and that of the second equation similarly involves the
unknown $Z^{n+1/2}\, .$
These are respectively constructed as $2\,X^n \simeq X^{n+1/2}+X^{n-1/2}$
and $2\, Z^{n+1/2} \simeq Z^{n+1}+Z^{n}$ and calculated using a fixed point
relaxation.

We observed numerically that this symplectic integrator offers better
stability properties than standard Runge-Kutta integrator and yields
remarkable conservation properties on test cases (see Section~\ref{Sect_Num}).

\subsection{Shifted grids in space}\label{sec-shift}
Besides the use of a staggered mesh in time, a staggered mesh in space can
also be used. This approach was for example used and fully justified (via
a mathematical demonstration) in \cite{ADL} to enforce impermeability
boundary conditions.

Shifted grids are also used here for the free surface. Our aim is to avoid
regularization techniques and yet desingularize the integrals involved in
the computation. We reformulated all singular integrals as regular ones
using relation \eqref{eq.desing} (see also Appendix~\ref{app-cot}).
The resulting integral is now non-singular, but can only be defined at the
former singularity as a continuous prolongation. This extension necessarily
involves higher order derivatives, which can induce some numerical errors
when the curvature becomes large (i.e. in a situation relevant to wave
breaking). Evaluating the integrals on a shifted dual grid resolves this problem
as the function is at worst evaluated half a grid point away from the
former singularity. The integral eventually needs to be interpolated on
the original grid for time stepping. This introduces some numerical
smoothing, which is however entirely controlled by the grid size and thus
vanishes in the limit of a large number of points. Finally, all derivatives
in space, with respect to $e$, in the discrete expressions are evaluated
thanks to second order finite difference formula.

\subsection{Discretization and inversion of singular operator by Neumann series}\label{sec-discret}

One of the main numerical difficulties lies in the resolution of linear
systems with matrices related to singular kernel operators. If it is well
known that the continuous operators are indeed invertible \cite{DautrayLions}, the relevant discretization, the invertibility
and the convergence of the discrete operators have recently been studied in
\cite{ADL}. The notation $A\, ,\,B$ and their adjoints $A^*\, ,\,B^*$ come
from Equation (3.1) in this article, where the relation and the inversion
is based on the Poincar\'e-Bertrand formula concerning the inversion of Cauchy
integrals, see \cite[Section 3]{ADL} for full details. 

Given an arc-length parametrization : $[0,L_{B}] \to \Gamma_{B}\,$, if 
\[
0=e_{B,1} <e_{B,2} < \dots <e_{B,N} < L_{B} 
\]
 are close to the uniform distribution $\theta_{i}= (i-1)L_{B}/N_{B}$, then
 the matrix $A^*_{B,N}$ appearing to compute $\mu_{B}$, see Section~\ref{sec-muB},
 defined as
\begin{align*}
 A_{B,N}^*(i,j)=& \frac{L_{B}}{N_{B}} \Real \Bigg[\frac1{2L_{B}\ic} \cot\Big(\frac{z_{B}(e_{B,i})-z_{B}(e_{B,j}) }{L/\pi}\Big) z_{B,e}(e_{B,j}) \Bigg]
 \ \forall i\neq j \in [1,N_{B}]\times [1,N_{B}]\,,\\
 A_{B,N}^*(i,i)=&\frac12- \frac{L_{B}}{N_{B}} \Real\Bigg[\frac1{2\pi\ic} \frac{z_{B,ee}(e_{B,i})}{z_{B,e}(e_{B,i})} \Bigg] \ \forall i\in [1,N_{B}]\,,
\end{align*}
is invertible and can be seen as a perturbation of $\frac12{\rm I}_{2}$.
It is thus a well conditioned matrix, see for instance \cite[Section
 8.3]{ADL}. It can be inverted very efficiently by a Neumann series. We
refer to \cite[Section 8]{ADL}, in particular Theorem 8.1 therein where the
convergence rate is given. 

Namely, we write 
\[
 A_{B,N}^*=\frac12({\rm I}_{N}-R_{B,N}),\quad \|R_{B,N}\|<1 \, ,
\]
which implies that
\[
 A_{B,N}^{*-1}=2({\rm I}_{N} - R_{B,N} )^{-1}= 2 \sum_{k=0}^{+\infty}
 R^k_{B,N} \, .
\]
In view of a fixed point procedure, we denote $R_{B,N}:=-2A_{B,N}^*+{\rm I}_{N}$, and 
\[
U_{n+1}=2 \sum_{k=0}^{n+1} R^k_{B,N} = R_{B,N} \Big(2 \sum_{k=0}^{n}
R_{B,N}^k\Big) + 2{\rm I}_{N} = R_{B,N} U_{n} +2{\rm I}_{N} , \quad
U_{0}=2{\rm I}_{N} \, .
\]
Indeed, the distance between $U_{n+1}$ and $U_{n}$ controls the error
 in the operator norm
\[
\|A_{B,N}^{*-1}-U_{n} \| = 2 \| \sum_{k=n+1}^{+\infty} R_{B,N}^k\| \leq 2
\|R_{B,N}^{n+1} \| \frac1{1- \|R_{B,N}\|} = \frac{\| U_{n+1}- U_{n}\|}{1-
 \|R_{B,N}\|} \, .
\]

Concerning the computation of $\partial_{t}\mu_{S}$, we have noticed in
Section~\ref{sec-bernoulli}, see \eqref{inv_muSt}, that we should invert 
\[
\cal A_{D,N}:=A^*_{S,N}-C_{D,N} A^{*-1}_{B,N}D_{D,N} \, ,
\]
 where $C_{D,N}$ and $D_{D,N}$ account for the interactions between the
 bottom and the free surface. If it has been established that $A^*_{B,N}$ is a
 perturbation of $\frac12{\rm I}_{N_{S}}$, it is not the case of $\cal A_{D,N}$
 because of the asymptotic behavior of $C_{D,N}$ and $D_{D,N}$
 when the free surface is far from the bottom. Indeed, studying the
 behavior when $z_{S}-z_{B} = \ic X + o(X)$ for large $X\in \R_+$, we get
 from the definition of $C_{D,N}$ and $D_{D,N}$ that 
\begin{align*}
 (C_{D,N})_{i,j}&=A_{tw}\frac{L_{B}}{N_{B}} \Real \Bigg[\frac1{2L\ic} \cot\Big(\frac{z_{S}(i)-z_{B}(j) }{L/\pi}\Big) z_{B,e}(j) \Bigg] \\
 &= A_{tw}\frac{L_{B}}{N_{B}} \Real \Bigg[\frac1{2L\ic} \left(-\ic z_{B,e}(j) \right)\Bigg] +o(1) \\
 &= -A_{tw} de_{B} \frac{\Real z_{B,e}(j)}{2L}+o(1).
\end{align*}
and
\[
 (D_{D,N})_{i,j} = \frac{L_{S}}{N_{S}} \Real \Bigg[\frac1{2L\ic} \cot\Big(\frac{ z_{B}(i)-z_{S}(j) }{L/\pi}\Big) z_{S,e}(j) \Bigg] = de_{S} \frac{\Real z_{S,e}(j)}{2L}+o(1).
\]
Using the decomposition of $A_{S,N}^*=\frac12({\rm I}_{N_{S}}-R_{S,N})$ and $A^{*-1}_{B,N}=2({\rm I}_{N_{B}}+\tilde R_{B,N})$ we conclude that
\begin{align*}
\cal A_{D,N} &= \frac12{\rm I}_{N_{S}} + \frac{A_{tw}de_{B} de_{S}}{2L^2}\Big(\Real z_{B,e}(j) \Big)_{i,j}\Big( \Real z_{S,e}(j)\Big)_{i,j}+\hat R_{B,N} +o(1)\\
& = \frac12{\rm I}_{N_{S}} + \frac{A_{tw}de_{B} de_{S}}{2L^2}\Big( \sum_{k=0}^{N_{B}} \Real z_{B,e}(k) \Real z_{S,e}(j) \Big)_{i,j}+\hat R_{B,N}+o(1)\\
& = \frac12 {\rm I}_{N_{S}} + \frac{A_{tw}de_{S}}{2L^2}\Big( \Real z_{S,e}(j) \Real\int_{0}^{L_{B}} z_{B,e}(e)\, de \Big)_{i,j}+\hat R_{B,N}+o(1)\\
& = \frac12 {\rm I}_{N_{S}} + \frac{A_{tw}de_{S}}{2L}\Big( \Real z_{S,e}(j)
\Big)_{i,j}+\hat R_{B,N}+o(1) = \frac12\tilde{\cal A}_{N}+\hat R_{B,N}+o(1)
\, ,
\end{align*}
where $de_{B}=L_B/N_B$ and $de_{S}=L_S/N_S \, .$

We are then first interested in inverting such matrices
\[
\tilde A= {\rm I}_{N_{S}} + (a_{j})_{i,j}.
\]
\begin{lemma}
 If $1+\sum_{k} a_{k}\neq 0$, then $\tilde A$ is invertible and we have
 \[
 \tilde A^{-1}= {\rm I}_{N_{S}} -\frac1{1+\sum a_{k}} (a_{j})_{i,j}.
 \]
\end{lemma}
\begin{proof}
 It is enough to check that the right hand side matrix is the right inverse to $\tilde A$, namely
 \begin{align*}
\tilde A \Big( {\rm I}_{N_{S}} -\frac1{1+\sum a_{k}} (a_{j})_{i,j} \Big) 
&= {\rm I}_{N_{S}} + (a_{j})_{i,j} -\frac1{1+\sum a_{k}} (a_{j})_{i,j} -\frac1{1+\sum a_{k}} (a_{j})_{i,j}^2\\
&= {\rm I}_{N_{S}} + \Bigg( \frac{a_{j}(1+\sum a_{k})}{1+\sum a_{k}} -\frac{a_{j}}{1+\sum a_{k}} -\frac{a_{j} \sum a_{k}}{1+\sum a_{k}} \Bigg)_{i,j} = {\rm I}_{N_{S}} .
 \end{align*}
\end{proof}

In our case, we have
\begin{equation}\label{sum-ak}
 \begin{aligned}
1+ \sum a_{j} &= 1+ \frac{A_{tw}de_{S}}{L} \sum_{j} \Real z_{S,e}(j) \\
&= 1+\frac{A_{tw}}{L}\Real\int_{0}^{L_{S}} z_{S,e}(e) \d e +o(1) =1+ A_{tw}+o(1) 
\end{aligned}
\end{equation}
which is non zero for the single-fluid formulation where $A_{tw}=1 \, ,$
but also for for the bi-fluid water-waves equations where $A_{tw}> -1$ if
$\rho_{F}>0$. 

By the previous lemma and \eqref{sum-ak}, we naturally set 
\[ 
\tilde{\cal A} := {\rm I}_{N_{S}} + \frac{A_{tw}de_{S}}{L}\Big( \Real z_{S,e}(j) \Big)_{i,j} ,\quad \tilde{\cal A}_{-1} := {\rm I}_{N_{S}} - \frac{A_{tw}de_{S}}{L(1+A_{tw})}\Big( \Real z_{S,e}(j) \Big)_{i,j}
\]
and
\[
\cal R := {\rm I}_{N_{S}} -2 \tilde{\cal A}_{-1} \cal A_{D,N} \, ,
\]
so that
\[
\cal A_{D,N} = \frac12\tilde{\cal A} ({\rm I}_{N_{S}} - \cal R ) \,,
\]
where we have neglected the error made in \eqref{sum-ak}.

We then have for $\|\cal R\|<1$ (because ${\cal A}_{D,N}=\frac12\tilde{\cal A}+o(1)$)
\[
\partial_{t} \mu_{S}= \cal A_{D,N}^{-1} {\rm RHS} = 2 ({\rm I}_{N_{B}} - \cal R )^{-1}\tilde{\cal A}_{-1} {\rm RHS}= 2 \sum_{k=0}^{+\infty} \cal R^k \tilde{\cal A}_{-1} {\rm RHS}.
\]
which can be written as a fixed point procedure
\begin{align*}
u_{n+1}&= 2\sum_{k=0}^{n+1} \cal R^k \tilde{\cal A}_{-1} {\rm RHS} = \cal R \Big(2 \sum_{k=0}^{n} \cal R^k\tilde{\cal A}_{-1} {\rm RHS} \Big) + 2\tilde{\cal A}_{-1} {\rm RHS} \\
&= \cal R u_{n} +u_{0} , \qquad u_{0}=2\tilde{\cal A}_{-1} {\rm RHS} \, . 
\end{align*}

Concerning the vortex method, we need to inverte $B_{B,N}$ which appears in the determination of $\gamma_{B}$ in the vortex formulation, see Sections~\ref{sec-gammaS0}, \ref{sec-gammaB}, \ref{sec-muS0}. Such an operator is related to the classical vortex method in domains with boundaries, see the operator $B$ in \cite[Section 3]{ADL}. The discrete version
\begin{align*}
 &B_{B,N}(i,j)= \frac{L_{B}}{N_{B}} \Imag \Bigg[ \frac{z_{B,e}(\tilde e_{B,i})}{2L\ic} \cot\Big(\frac{z_{B}(\tilde e_{B,i})-z_{B}(e_{B,j}) }{L/\pi}\Big)\Bigg]
 \ \forall (i,j)\in [1,N_{B}-1]\times [1,N_{B}]\,,\\
 &B_{B,N}(N_{B},j)=\frac{L_{B}}{N_{B}} \ \forall j\in [1,N_{B}]\,,
\end{align*}
is invertible if $\tilde e_{B,i}\in (e_{B,i},e_{B,i+1})$ are close to a
uniform distribution $\tilde \theta_{i}= (\theta_{i}+
\theta_{i+1})/2\,$.
Moreover \cite[Theorem 2.1]{ADL} states that
$\gamma_{B, N} = B_{B,N}^{-1} RHS_{VB,N}$ is a good approximation of
$\gamma_{B}\,$. Even though it is invertible, the matrix is not
well-conditioned and cannot be seen as a Neumann series, except relating
$B^{-1}$ to $A^{-1}$ through \cite[Equations (3.19) $\&$
 (3.22)]{ADL}.
However, this step only needs to be performed once at the beginning of the numerical
integration (since this matrix does not
evolve with time). It is thus worth inverting it accurately. 

The last matrix to inverte is 
\[
\cal A_{V,N}:=A_{S,N}-C_{V,N} B^{-1}_{B,N}D_{V,N}
\]
to compute $\partial_{t}\gamma_{S}$, see \eqref{inv_gammaSt}, which can be
also inverted by Neumann series as we did for $\cal A_{D,N}$.

\section{Numerical results and convergence}\label{Sect_Num}

In order to compare and validate the various numerical approaches, we have
considered three different initial conditions and test cases when the
bottom is flat $\Gamma_{B}=\T_{L}\times \{-h_{0}\}$. In all test cases, we
used $N_S=N_B=N$.

While simulations were performed using different values of the Atwood
number, we report here simulations with $A_{tw}=1 \, ,$ i.e. the
single-fluid formulation, which can be compared with existing solutions.
We also neglect surface tension in all the test cases below. This makes the
tests below more demanding, since the surface tension generally has 
regularising effects.

The parameter $\alpha$ was set to $\alpha=1$,
implying that the points are advected tangentially at the velocity of the
lower fluid. Also all test cases presented here include a flat bottom
(though the cases of infinite depth or variable bottom can also be handle
using the same code). The simulations are made non-dimensional, with a
length-scale based on the water depth, such that $h_0=1$, and a time-scale
based on gravity, such that $g=1 \, .$

We will numerically investigate the stability of our scheme in a few test cases.
We will then validate the numerically obtained solutions against analytical solutions, when
available.
When not available, we will simply check convergence to an $N$ independent solution
in the limit of large $N\, .$ 
Another useful validation concerns conserved quantities. 
The first thing to be checked is the conservation of mass
\begin{align*}
M(t):=&\rho_{F}{\rm Vol\,}\D_{F}(t)=\rho_{F}\iint_{\D_{F}(t)} \div \begin{pmatrix} 0\\ x_{2} \end{pmatrix} \d \mb x = \rho_{F}\int_{\partial\D_{F}(t)} \begin{pmatrix} 0\\ x_{2} \end{pmatrix} \cdot \tilde{\mb n} \d \sigma(\mb x)\\
=&\rho_{F} \int_{0}^{L_{S}} z_{S,2}(e) z_{S,e,1}(e)\d e - \rho_{F} \int_{0}^{L_{B}} z_{B,2}(e) z_{B,e,1}(e)\d e
\end{align*}
and 
the total energy
\begin{align*}
E(t):=&\frac{1}2 \iint_{\D_{F}(t)}\rho_{F} |\mb u|^2 + \iint_{\D_{F}(t)}\rho_{F} g x_{2}\\
=& \frac{\rho_{F}}2 \iint_{\D_{F}(t)} \nabla \phi_{F}\cdot \nabla^\perp \psi + \frac{\rho_{F} g}2 \iint_{\D_{F}(t)} \div \begin{pmatrix} 0\\ x_{2}^2 \end{pmatrix} \d \mb x \\
=&-\frac{\rho_{F}}2 \int_{0}^{L_{S}} \Real\Big[\widehat{\mb u}_{F}(z_{S}(e)) z_{S,e}(e)\Big] \psi(z_{S}(e))\d e\\&+ \frac{\rho_{F}}2 \int_{0}^{L_{B}} \Real\Big[\widehat{\mb u}_{F}(z_{B}(e)) z_{B,e}(e)\Big] \psi(z_{B}(e))\d e\\
&+\frac{\rho_{F}g}2 \int_{0}^{L_{S}} z_{S,2}^2(e) z_{S,e,1}(e)\d e - \frac{\rho_{F}g}2 \int_{0}^{L_{B}} z_{B,2}^2(e) z_{B,e,1}(e)\d e
\end{align*}
where the expression of $\widehat{\mb u}_{F}$ and $\psi$, given in
\eqref{BS-u} and \eqref{psi-mu}, has to be considered with the limit
formulae of Appendix~\ref{app-cot}. 

We also successfully checked conserved integrals for domains with horizontal symmetries
\cite{AlazardConserv} (not detailed here).

\subsection{Case 1: linear water waves}

The first test case we consider is that of simple waves of small
amplitude, it takes the form provided by Stokes at first order
\begin{equation}
\eta(t,x)= A \,\cos( k x - \omega t)\, , \qquad
\Phi(t,x, y)= A \, \frac{\omega}{k} \, \frac{\cosh k(y+h_0) }{\sinh k h_0}
\,\sin( k x - \omega t)\, ,
\label{eq z Phi lin}
\end{equation}
\vskip -2mm
\begin{equation*}
\text{with\ } \omega= \sqrt{ g k \tanh k h_0 } \, ,
\text{\ and\ } u_x=\partial_x \Phi \, , u_y=\partial_y \Phi \, .
\end{equation*}
Our initial condition is thus
\begin{equation}
 \eta_{0}(x)= A \,\cos( k x )\, , \qquad
 {\mb u} \cdot {\mb n} = A \, \sin kx \, \sqrt{g k \, \tanh kh_0} \, .
\label{eq CI simple wave}
\end{equation}

We also consider Stokes waves at second order, where the initial data is slightly modified (see \cite{Johnson}):
\begin{align*}
 \eta_{0}(x)&= A \,\cos( k x ) + kA^2 \frac{3- \tanh^2 kh_{0}}{4\tanh^3 kh_{0}}\, \cos( 2k x) \, , \\
 {\mb u} \cdot {\mb n} &= \frac{A \, \sin kx \, \sqrt{g k \, \tanh kh_0}}{\sqrt{1+k^2A^2\, \sin^2(kx)}}\Big(1+ \frac{kA}{\tanh kh_0}\, \cos (kx) \Big) \, .
\end{align*}

In our simulations, we used ${\mb u} \cdot {\mb n}$ as boundary condition to
numerically find $\gamma_{0}$ and $\mu_{0}$ (see \S~\ref{sec-gammaS0} and \S~\ref{sec-muS0}).
This contrasts with \cite{Baker82} which provides the analytical expression
for both $\gamma$ and $\mu$ in the limit of a vanishing amplitude. This distinction is small
at this stage, but turns out to be important at later stage (see our third test case below).

We consider numerically a wave number $k=1$, and a domain of horizontal extent $L=2\, \pi\, .$
We integrate our simulations for a time $t_f=10 \, k / \omega \, .$
We introduce a CFL-number
\begin{equation*}
CFL= \min (| \dot z_S | \, \Delta t / (|z_{S,e}| \, \d e)) \, ,
\end{equation*}
based on the Lagrangian velocity of the points sampling the interface.
We have checked that our numerical code is stable for a CFL-number
$CFL \leq CFL^*$, we observed numerically that $CFL^* \in [0.25, 0.5)\, .$
We can then test the convergence of our numerical solution by varying the spatial
resolution $N \, .$ In order to keep temporal truncation error, we then
used $CFL \leq 1/10 \, .$
Convergence is then assessed by comparing the
final state to the initial condition $\cal E=\| z(t_f) - z(0) \| _\infty$. The
resulting errors for waves of amplitudes $A=10^{-2}$ and $A=10^{-3}$
respectively are represented in Figure~\ref{fig simple wave}(a,b).
The error, defined in $L^\infty$-norm, decreases with increasing resolution to a value which decreases
with the amplitude of the wave. This can be interpreted as the signature of
weak non-linear corrections. 
Interestingly, the vortex and the dipole methods are indistinguishable in these graphs.

\begin{figure}
 \centerline{(a){\includegraphics[height=0.35\textwidth]{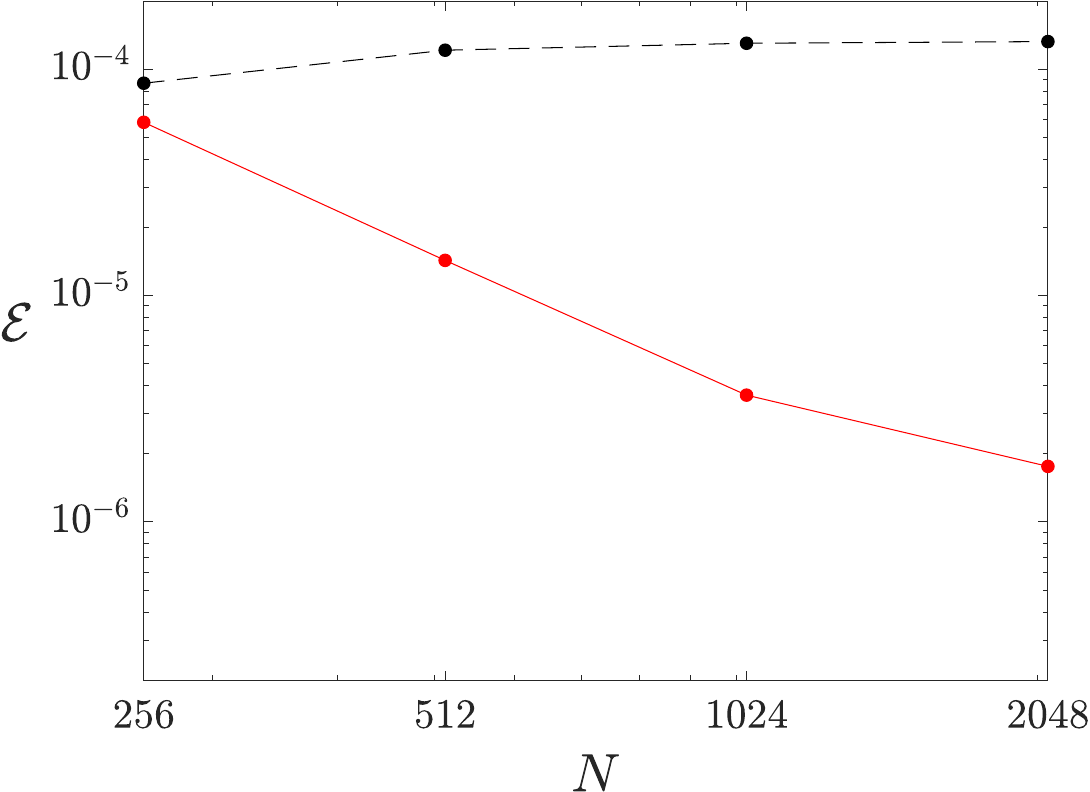}}
 \ \ (b){\includegraphics[height=0.35\textwidth]{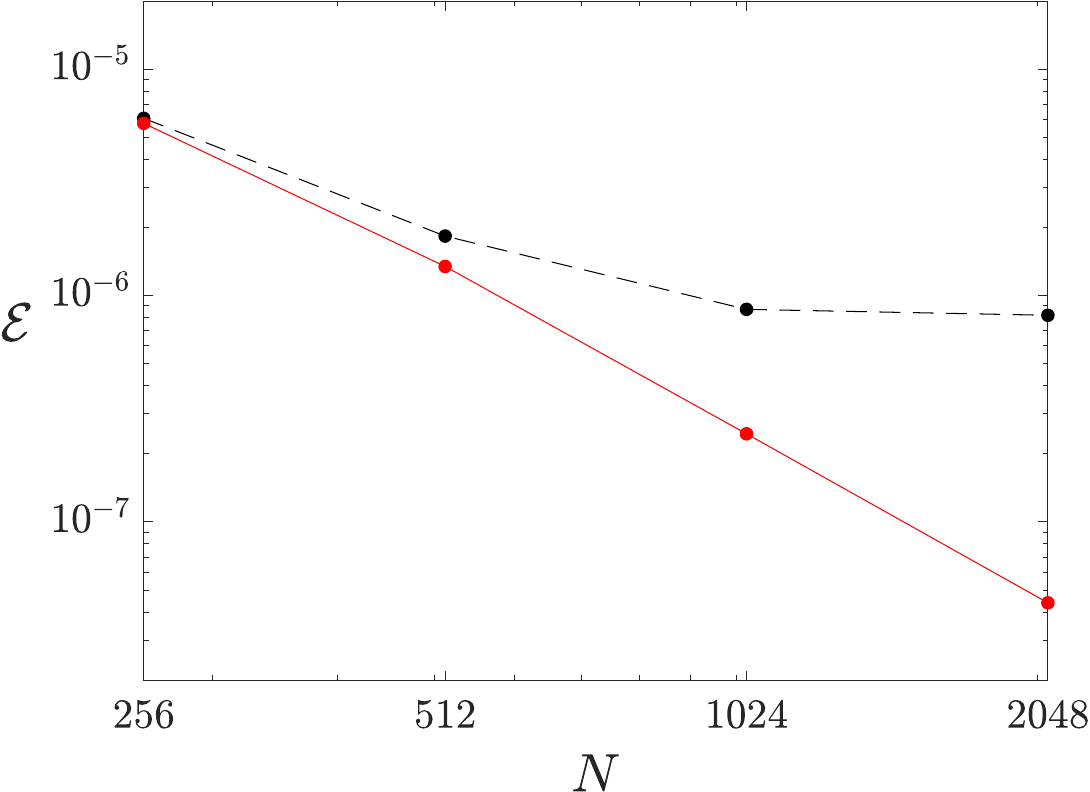}}}
\caption{Convergence of the numerical integration after 10 revolution periods, for
 simple waves with amplitude (a) 0.01 and (b) 0.001. The error (compared to
 the initial condition) defined as $\cal E=\| z(t_f) - z(0) \| _\infty$ is
 represented as a function of the number $N$ of points used in the
 numerics.
 The black dashed
 curve corresponds to the first order simple wave, whereas the solid red
 curve corresponds to the second order Stokes solution. The vortex and the
 dipole methods are indistinguishable in these graphs. 
}
\label{fig simple wave}
\end{figure}

\subsection{Case 2: solitary waves}
Our second test case concerns solitary waves. We consider here an
extension of solitary waves to a periodic domain.
This involves the cnoidal wave solution for the Green-Naghdi equation.

To present this solution, first we define the Jacobi elliptic functions
${\rm cn}(u,m):=\cos \varphi(u,m)$ as the inverse of the elliptic integral
\[
u(\varphi,m):=\int_{0}^\varphi \frac{\d \theta}{\sqrt{1-m\sin^2\theta}}
\] 
where $m\in (0,1)$. We also need the {\it complete elliptic integral of the first and second kind}:
\begin{equation*}
K(m):=\int_{0}^{\frac{\pi}{2}}\frac{\d\theta}{\sqrt{1-m\sin^2\theta}}\quad\text{and}\quad E(m):=\int_{0}^{\frac{\pi}{2}}\sqrt{1-m\sin^2\theta}\d\theta\,.
\end{equation*}

For a given {\it period} $L$, {\it amplitude} $A$ and {\it depth} $h_{0} \, ,$
we determine the nonlinearity parameter $m\in (0,1)$ verifying the following dispersion relation
\begin{equation*}
AL^2=\frac{16}{3}mK^2(m)\frac{h_{0}^2}{g}c^2(m),
\end{equation*}
where the velocity $c$ is given by
\[
c(m):=\sqrt{g h_{0}\Big(1+\frac{\eta_1(m)}{h_{0}}\Big)\Big(1+\frac{\eta_2(m)}{h_{0}}\Big)\Big(1+\frac{\eta_3(m)}{h_{0}}\Big)}
\]
with
\[
\eta_1(m):=-\frac{A}{m}\frac{E(m)}{K(m)},\quad \eta_2(m):=\frac{A}{m}\Big(1-m-\frac{E(m)}{K(m)}\Big),\quad\eta_3(m):=\frac{A}{m}\Big(1-\frac{E(m)}{K(m)}\Big).
\]

Once $m$ is obtained, the Green-Naghdi soliton is then defined by the surface elevation
\begin{equation*}
\eta(t,x)=\eta_2+A\, \hbox{cn}^2\Big(\frac{2K(m)}{L}(x-ct),m\Big).
\end{equation*}
As it translates to the right with velocity $c$, we write 
\[
\mb u\cdot \mb n\vert_{(0,x,\eta(0,x))}= -c\frac{\partial_{x}\eta(0,x)}{\sqrt{1+|\partial_{x}\eta(0,x)|^2}},
\]
which uniquely defines $\gamma_{0}$ and $\mu_{0}$, see \S~\ref{sec-gammaS0}
and \S~\ref{sec-muS0}.

For more details on cnoidal waves, we refer to \cite{cnoide1, cnoide2}, and references therein. 

We consider numerically a domain of large extent to allow for a localized solitary
wave, and choose $L= 40 \, \pi \, .$
For an amplitude $A=0.1\, ,$ we computed a circulation $\gamma = 3 \cdot 10^{-3} \, .$
We checked numerically that setting $\gamma=0$ did not alter the solution
in this case (i.e. the circulation is weak enough not to affect the solution).
The results are presented in Figure~\ref{fig solitary}.
Figure~\ref{fig solitary}(a) highlights the modification induced on the above initial
condition after one full period, i.e. $t^*=L/c\simeq 124.23$. The difference (represented in dashed
blue) appears dominated by a correction on the velocity $c$, presumably
due to higher order correction terms. We performed simulations up to
$t=10\, t^*$ which confirmed that no change of shape is observed, but the
phase lag with the analytical solution increases with time.

Both methods appear stable in time and no significant change on the
numerical solution was observed when the time-step was reduced by a factor
of $2 \, .$

The numerical convergence of both methods is demonstrated in
Figure~\ref{fig solitary}(b).
Here the solitons are localized and a small difference of phase
velocity yields a large $L^\infty$-norm, not reflecting the actual distance
between the two curves.
We thus use a discrete Hausdorff distance between two solutions to measure
the error (see Fig.~\ref{fig solitary}.b):
\begin{equation}
{\cal E}=d_H(z_1,z_2)\equiv
\max\left(\max_i(\min_j(|z_1(i)-z_2(j)|)),\max_j(\min_i(|z_1(i)-z_2(j)|))\right)\,
,
\label{hausdorff}
\end{equation} 
where both curves were interpolated using $2^{17}\simeq 130.000$ points in
order to identify accurately enough the closer points on both curves.

Both methods exhibit an approximately
second order convergence in space (as highlighted by the black dotted
line).
Both the vortex and the dipole methods yield very good results on this test
case as well.

For this rather long integration (up to $t=t^*\simeq 124.23$), we observe that the
volume is conserved up to fluctuations of the order of $ 10^{-8}$ and the
total energy (kinetic and potential) is very weakly dissipated (or the order
of $10^{-4}$ over the full period).

\begin{figure}
\centerline{(a){\includegraphics[height=0.35\textwidth]{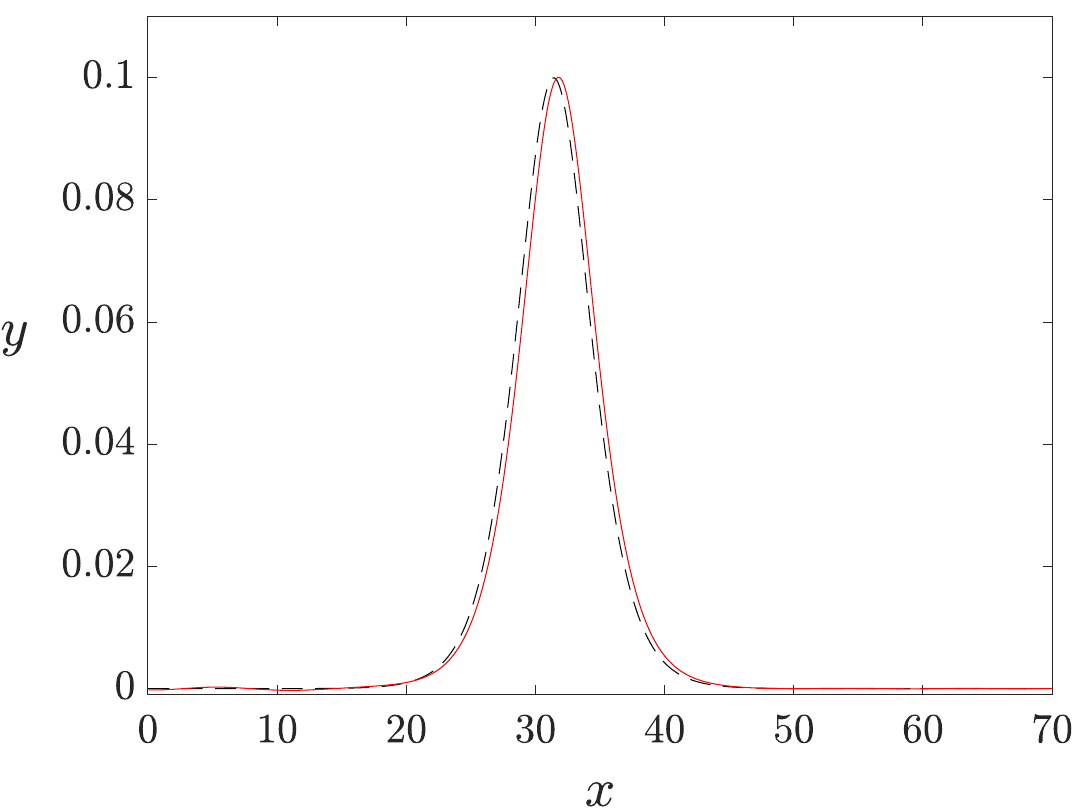}}
 \ \ (b)\raisebox{-0.2ex}{\includegraphics[height=0.35\textwidth]{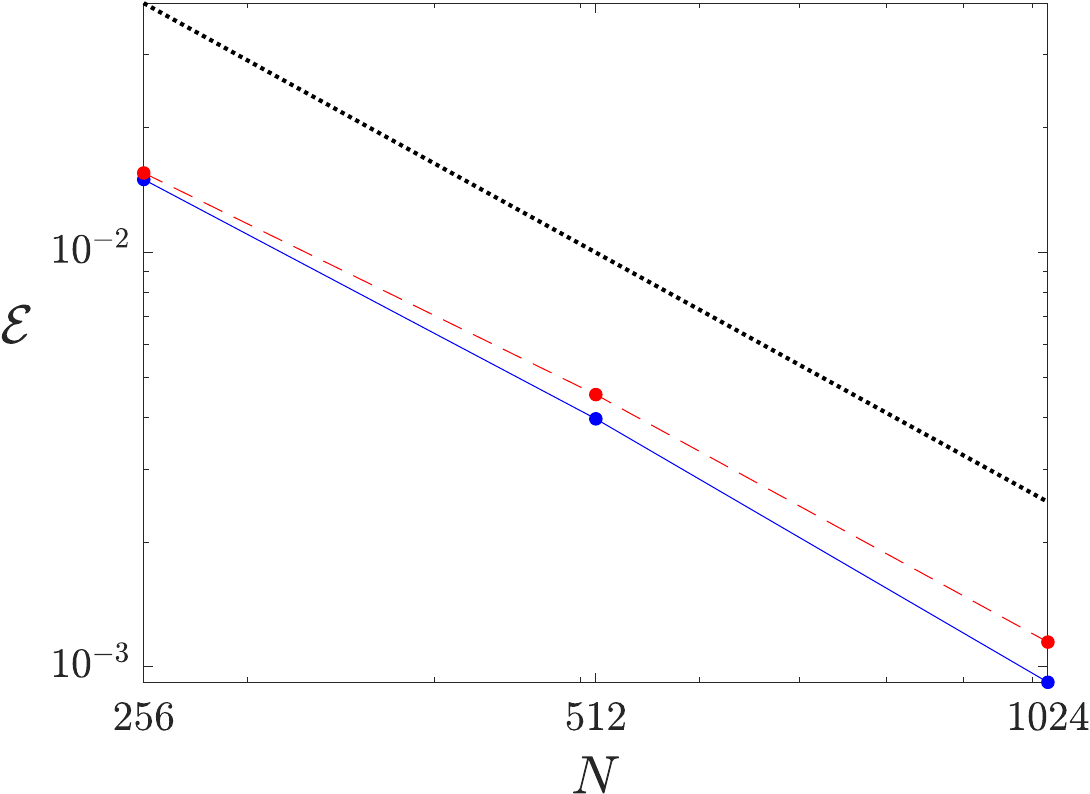}}}
\caption{(a) The initial solitary wave (dashed black) and the numerical integration with $N=2048$
 using the dipole formulation up to $t=124.23$,
i.e. after one revolution of the analytical solution. 
(b) Convergence of the dipole method (blue) and vortex method (red). The
error, measured using the Hausdorff distance, for a given resolution $N$ is
computed relative to $N=2048$ for the 
same method at time $t=124.23$. The black dotted line illustrates 
second order convergence. The full computational domain extends to $L=40\pi \, .$
}
\label{fig solitary}
\end{figure}

\subsection{Case 3: wave breaking}

Finally we turn to a strongly non-linear problem, that of a wave
breaking.
We consider a flat bottom and a mean water depth $h=1$.
 We consider an initial interface of the form
 \begin{equation} 
 \eta_{0}(x)= A \,\cos( k x ) \qquad \text{with}\ A=\frac{1}{2}\
 \text{and} \ k=1 \, .
 \label{CI breakinga} 
 \end{equation} 
Our initial velocity
 stems from \eqref{eq CI simple wave} but is now used with large
 amplitude.
 In doing so it is important to relax
 the approximation $\mb n \simeq \mb e_y \, .$ We start with \eqref{eq z Phi lin} which yields the
initial condition on velocity
\begin{equation*}
u_x=\partial_{x}\Phi= A\, \sqrt{\frac{g k}{\tanh kh}} \, \cos kx\, ,
\qquad
u_y= \partial_{y}\Phi=A\, \sqrt{g k \, \tanh kh} \, \sin kx\, .
\end{equation*}
Using expression of the normal
\begin{equation*}
 \mb n =
 \frac{1}{\sqrt{1+(\partial_x \eta)^2}}
 \begin{pmatrix}
 - \partial_x \eta\\
 1
 \end{pmatrix} 
 =
 \frac{1}{\sqrt{1+k^2 A^2 \sin ^2 kx}}
 \begin{pmatrix}
 k A \sin kx \\
 1 
 \end{pmatrix} \, ,
\end{equation*}
we get
\begin{equation}
 \mb u \cdot \mb n = \frac{A \, \sin kx \, \sqrt{g k \, \tanh kh} }{\sqrt{1+k^2 A^2 \sin ^2 kx}}
 \left(1 
 +k \,
 A\, {\frac{1}{\tanh kh}} \, \cos kx\, 
 \right)
\label{CI breakingb} 
\end{equation}
which resumes to \eqref{eq CI simple wave} in the limit of vanishing amplitudes.
It is interesting to note that this differs from \cite{Baker82}, which used
the $\gamma_{0}$ and $\mu_{0}$ constructed from the small amplitude approximation.

It should also be stressed that the construction of a sensible initial
condition at large amplitude (wave breaking) from the simple wave analytics
is necessarily based on simplifying assumption as there is no analytical
solution in this strongly non-linear limit.

This situation is in a fully non-linear regime and more challenging
numerically than the two previous test cases.

We observed here, as reported by \cite{Baker82}, that for fully-nonlinear
configurations the vortex method is unstable to a high wave number
instability. As the resolution is increased and higher wave numbers are
resolved, the integration time before an instability occurs decreases (see
first column in Table~\ref{table final times}).

Such is not the case (again as stressed by \cite{Baker82}) for the dipole
method, for which the integration time is fairly independent of the
resolution (see second column in Table~\ref{table final times}).

\begin{table}
 \[
 \begin{array}{l|llll}
 & {\rm Vortex} &{\rm Dipole} & {\rm OEC\!-\!dipole} & {\rm F\!-\!vortex}\\
 \hline
256 & 2.40 & 3.02 & 3.06 & 2.81 \\
512 & 2.03 & 3.04 & 3.37 & 2.80 \\
1024 & 1.73 & 3.13 & 3.39 & 2.86 \\
2048 & 1.15 & 3.10 & 3.60 & 2.93 \\
 \end{array}
 \]
 \caption{Final integration time before the occurrence of a numerical
 instability for various numerical schemes (the F-vortex method is
 introduced in \S~\ref{sectionFFT}).
We also integrated the {\rm OEC\!-\!dipole} formulation (introduced later
in the paper) with $4096$ points and
obtained a stable solution (nearly undistinguishable from the $2048$ points
simulation) until $t=3.61$. 
 }
 \label{table final times}
\end{table}

As we have shown in the first two cases, both the vortex method and the
dipole method are stable and can be used in situations involving a small to
moderate curvature. In the case of huge curvature (such as wave-breaking),
the vortex method becomes impractical and does not converge (as the
stability decreases with increasing resolution).

For both methods, we observed that the volume and total energy are
conserved up to less than $1 \% \, .$

The case of the dipole
method is not quite as severe, but we observe a numerical instability at a
time independent of $N$, well before the splash occurs.
The instability of the dipole method appears to develop first in the form of points approaching
each other in the direction tangent to the interface, in a manner similar
to the so called `phantom traffic jam' instability. In order to delay the
formation of this instability, we introduced an `odd-even coupling' (OEC)
at the end of each time-step in the form of an arbitrary regularization on $\mu$ 
by replacing at the end of each time step,
$\partial_t \mu_S(e_k)$ with
$\left( \partial_t \mu_S(e_{k-1})+2 \partial_t \mu_S(e_k)+\partial_t
\mu_S(e_{k+1})\right)/4\, .$
Using this approach we could extend the integration time by nearly
$10\%$
(see Table~\ref{table final times}).

For large grids, this coupling procedure does not alter the
simulations on the first two test-cases such as simple waves or solitary waves,
but does stabilise the numerical integration of wave breaking. 

The OEC approach introduces a stabilisation, which however does not affect the
overall convergence of the scheme. At times $t=1$, $t=2$ and $t=3$,
 the convergence of the OEC approach is
undistinguishible form that of the dipole method, see figure~\ref{fig oec
 nooec}.
Besides the OEC vanishes continuously in the limit of large grids.

It is worth stressing that contrarily to other regularization techniques
previously used on this problem (and discussed in the following section),
the above OEC does not involve any arbitrary small parameter $\varepsilon$
other than the grid space. This approach is thus free of the risk to present results
with vanishing distance between points and yet a finite regularization.

The curves marking the interface between water and air, are generally not
graphs for this test-case. We thus stick to the discrete Hausdorff distance 
between two curves, see
\eqref{hausdorff},
to measure the error (see Figs.~\ref{fig oec nooec} and \ref{fig wave breaking}).

It is worth stressing that the very same test case has been recently
investigated using the Navier-Stokes equations and a Finite Element
discretization, and that convergence has been achieved to the solution
portrayed on Fig.~\ref{fig wave breaking}.a as the Reynolds number is
increased~\cite{Riquier23}.

\begin{figure}
 \centerline{(a){\includegraphics[height=0.35\textwidth]{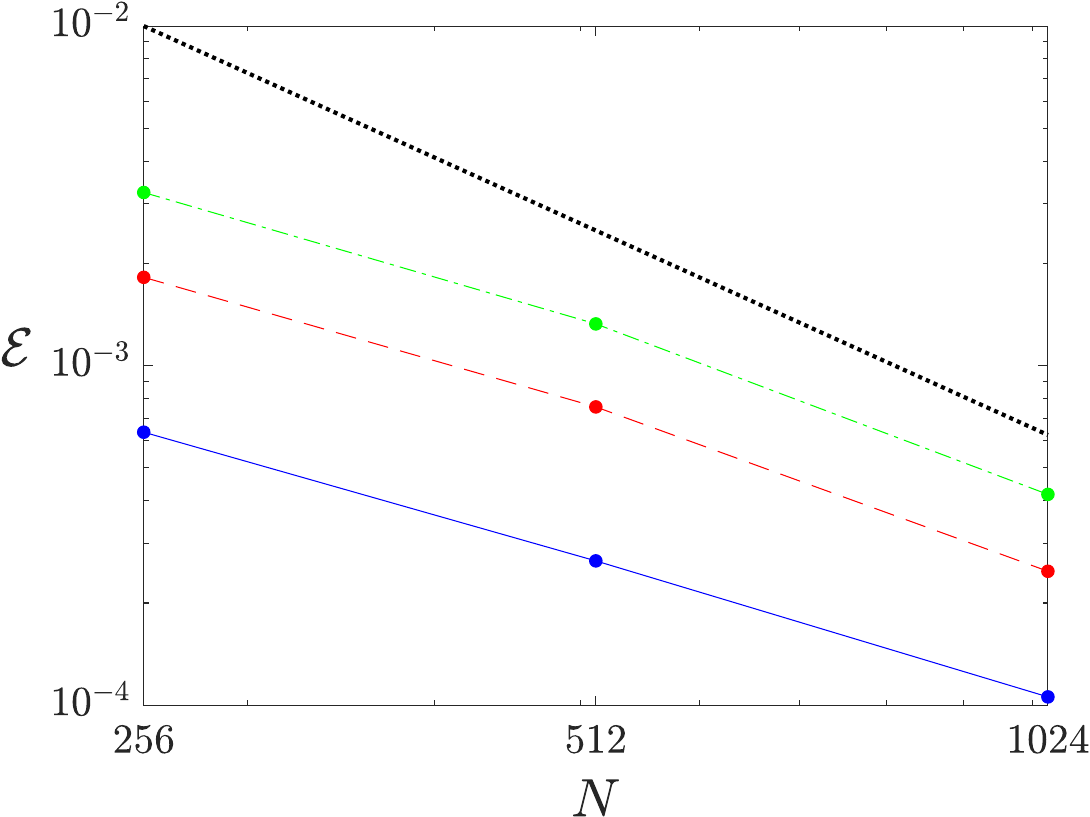}}
 \ \ (b){\includegraphics[height=0.35\textwidth]{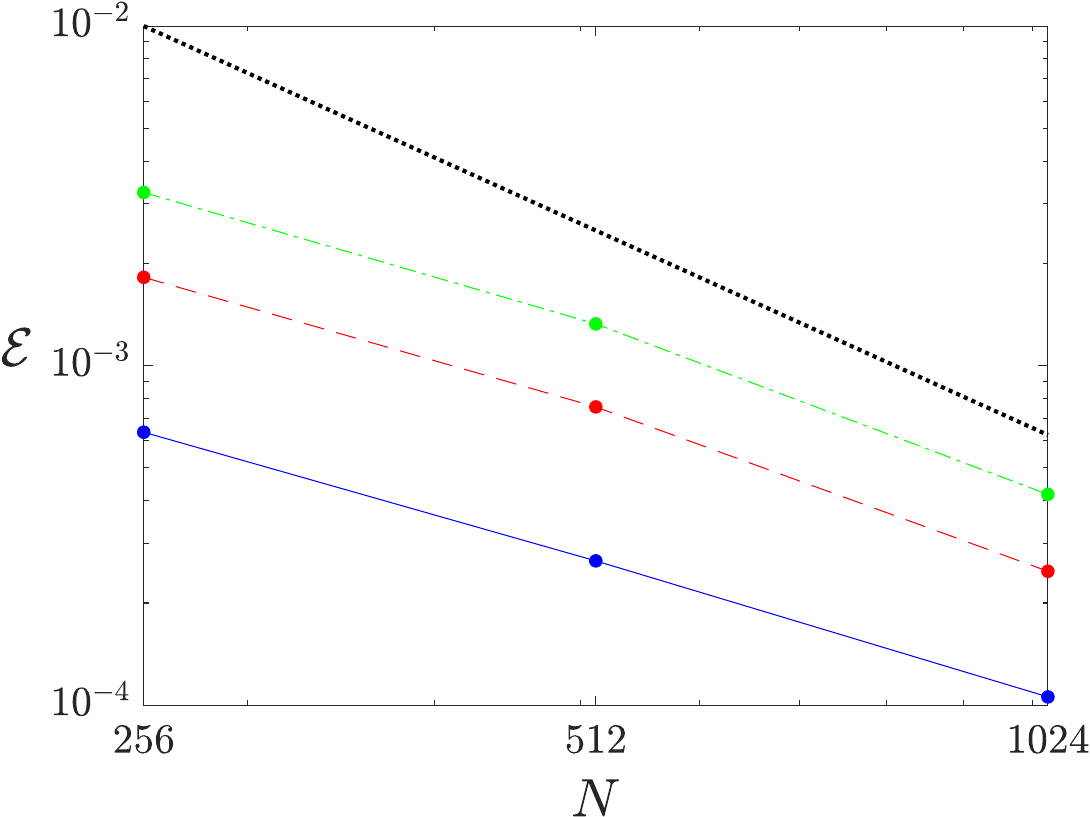}}}
 \caption{Convergence of the dipole method (a) and of the OEC-dipole (b).
 In both cases the reference solution is that obtained with the dipole
 method for $N=2048 \, .$ The error, measured by the Hausdorff distance, is represented at time $t=1.00$
 (solid blue), $t=2.00$ (dashed red) and
 time $t=3.00$ (dot dashed green). The dotted black curve indicates second
 order decrease.}
\label{fig oec nooec}
\end{figure}

\begin{figure}
 \centerline{\includegraphics[width=1.0\textwidth,trim=20 70 10 100, clip]{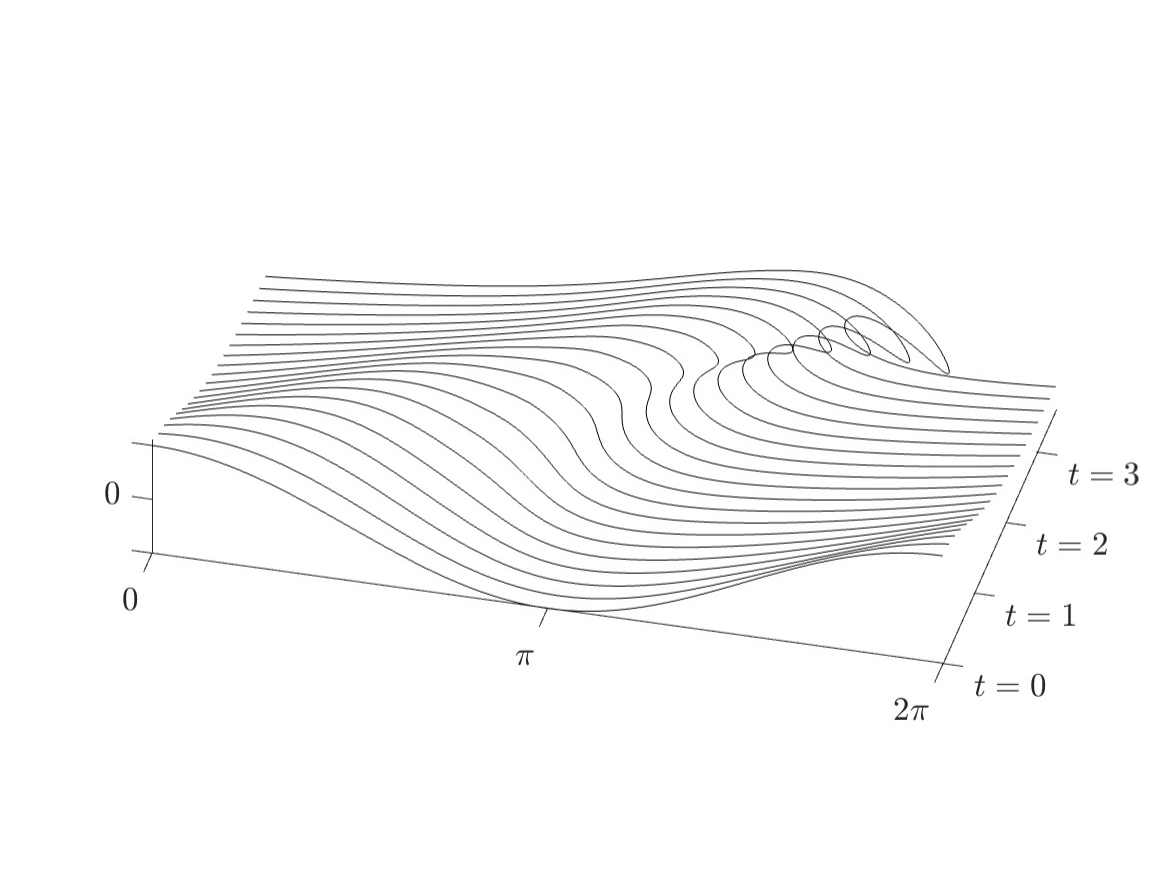}}
 \vskip -6mm
 \leftline{(a)}
 \vskip 1mm
 \centerline{\includegraphics[width=1.0\textwidth,trim=0 0 0 0, clip]{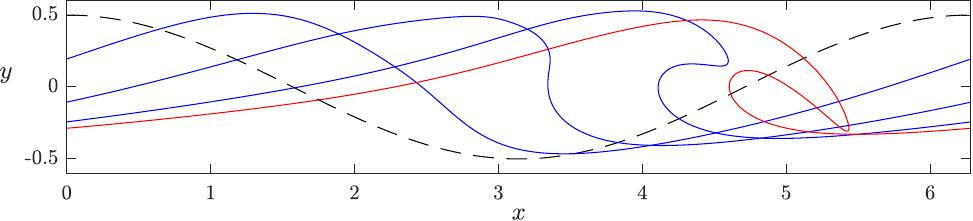}} 
 \vskip -6mm
 \leftline{(b)}
 \centerline{(c)\includegraphics[height=5.7cm]{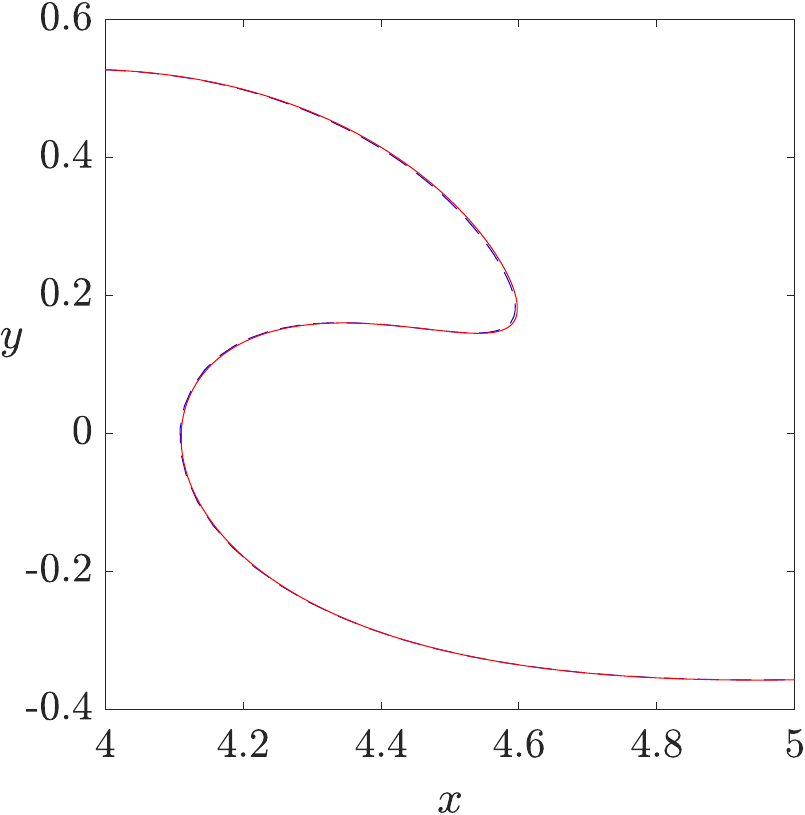} \ \
 (d)\raisebox{-0.2ex}{\includegraphics[height=5.7cm,trim=0 0 0 0, clip]{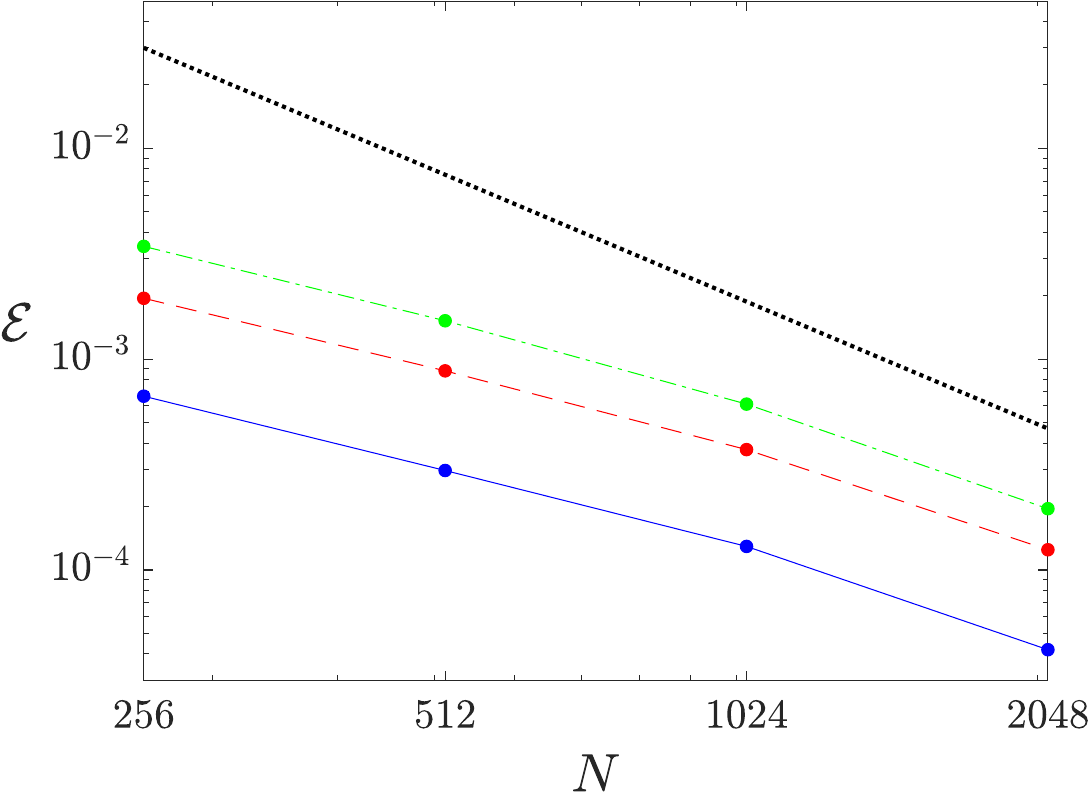}}}
 \caption{(a) Time evolutions of the breaking wave with $N=4096$ and the
 OEC-dipole discretization from the $\tfrac{1}{2}\cos(x)$ profile at
 $t=0$ (see~\eqref{CI breakinga}--\eqref{CI breakingb}) to the splash at
 $t\simeq 3.6$.
 (b) Plots at time $t=1$, $t=2$, $t=3$, $t=3.6$.
 (c) Solutions using the OEC-dipole discretization with 256 points (dashed
 blue) and 4096 points (solid red) at time t=3.0. (d) Convergence
of the OEC-dipole toward the OEC-dipole $N=4096$ curve used as reference at time $t=1.0$
 (solid blue), $t=2.0$ (dashed red) and
 time $t=3.0$ (dot dashed green). The dotted line indicates the exact second order convergence
slope.}
 \label{fig wave breaking}
\end{figure}

\section{Comparison with regularization strategies}\label{sec6}

\subsection{Fourier filtering for the vortex method} \label{sectionFFT}

Since the Vortex methods appears unstable (see \cite{Baker82} and Table~\ref{table final times}),
some Fourier filtering can be introduced.
The strategy of a filtered Vortex method is for example followed by
\cite{Wilkening21}.

The filtering corresponds to a product in Fourier space of both $z_S$ and
$\gamma_S$ with a filter function. We used here the filter function
introduced in \cite{Baker11} (see eq. (2.14) in this reference)
\[
\hat F = 
\frac{1}{2} -\frac{1}{2}\tanh\left(\frac{2 \, |k| \,
 \pi/N_S-\xi_0}{d}\right) \, .
\]
Two parameters are introduced in the above. 
$\xi_0$ locates the center of the transition zone (usually some fraction of $\pi$)
and $d$ controls the width of the transition zone.
Following \cite{Baker11}, we used $\xi_0=\pi/4$ and $d=\pi/40$.

The Fourier filtered method (referred to as `F-vortex' in the text) yields an
increased stability and thus longer integration. Remarkably the final time
of integration for the F-vortex method is independent of $N$ (see Table~\ref{table final times})
and the method does converge see
Fig.~\ref{error_FFT_nooec}(a).
Although the simulation has been extended in time far beyond the unfiltered
vortex method, the solution does converge (to first order) to that of the
dipole method, see
Fig.~\ref{error_FFT_nooec}(b).
The final integration time
however remains shorter than for the dipole method (let alone the
OEC-dipole method).

\begin{figure}
\centerline{(a) \includegraphics[height=0.3\textwidth]{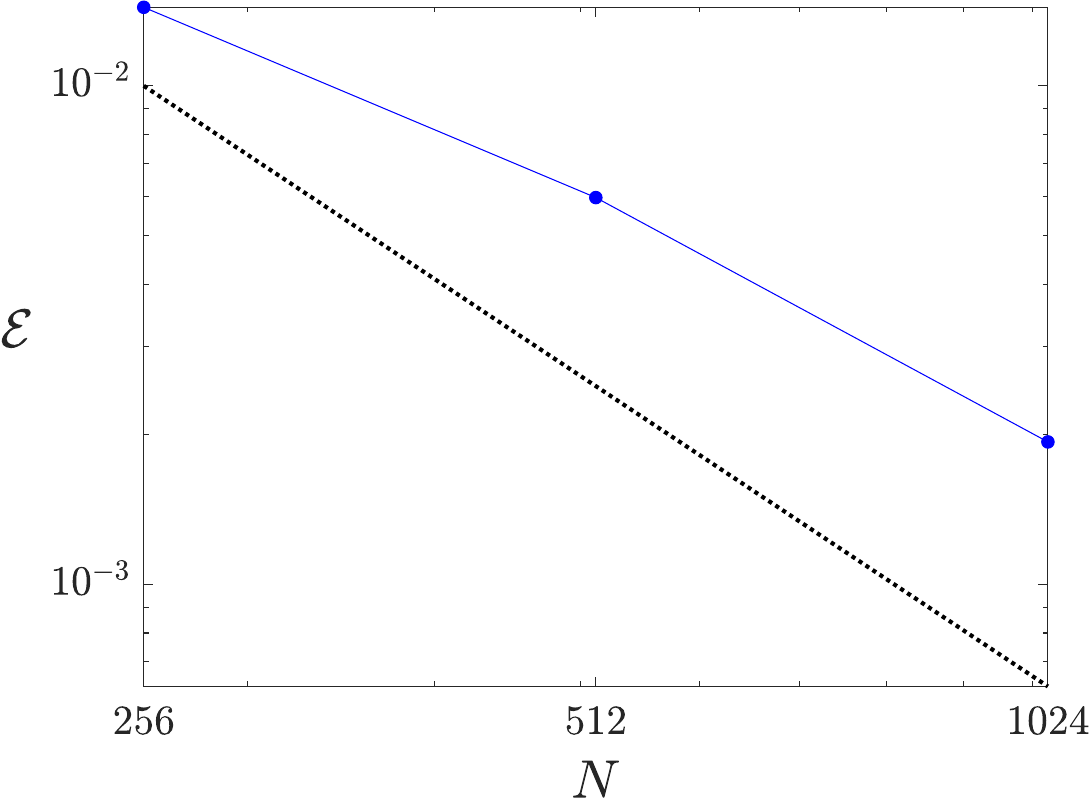}
\ (b) \includegraphics[height=0.3\textwidth]{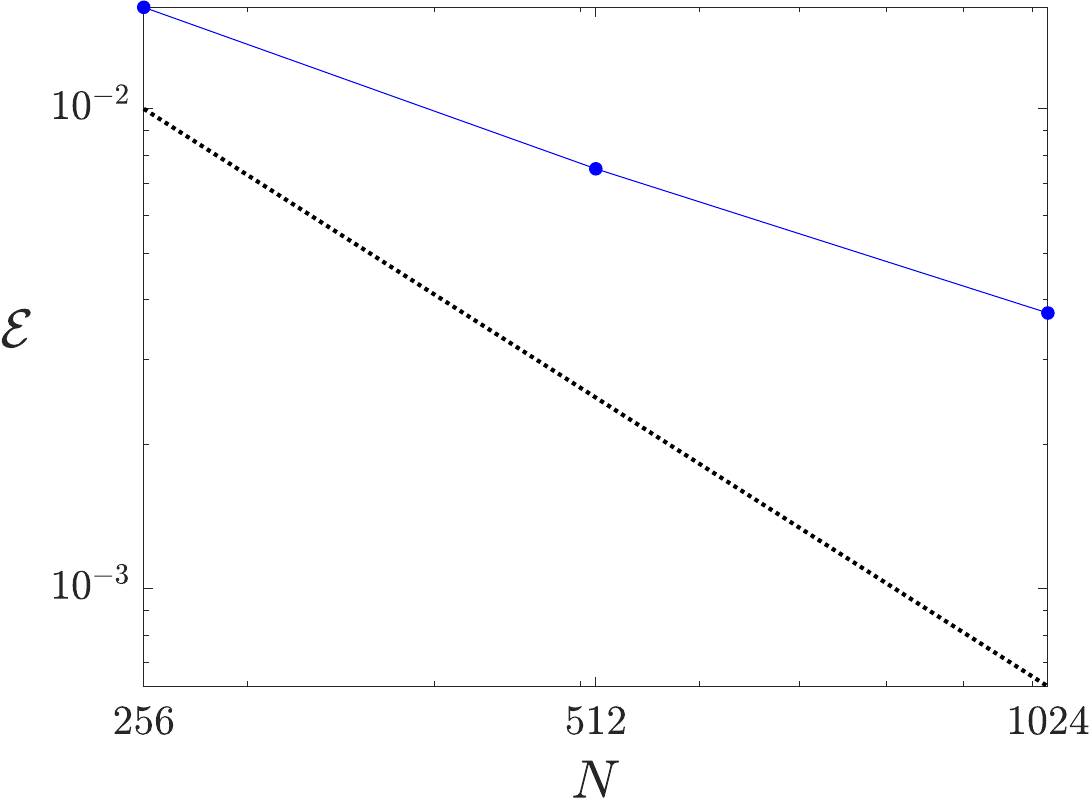}} 
 \caption{Convergence of the F-vortex method at time $t=2.75 \, ;$
 (a) compared to $N=2048$ with the same method (used as a reference). (b) compared to the
 dipole method with $N=2048 \, .$
 The dotted black line indicates second order convergence.}
\label{error_FFT_nooec}
\end{figure}

The use of $\xi_0=\pi/4$ in the above tests (guided by \cite{Baker11})
yields an `effective' resolution of approximately $N/4 \, , $ (though with a
larger stability than the pure vortex method with $N/4$).
Increasing the cut-off frequency, say to $\xi_0=\pi/2$ instead of
$\xi_0=\pi/4$ yields a less stable scheme. The observed time for
instability with $\xi_0=\pi/2$ was $t=2.64$ for $N=512$ and $t=2.24$ for
$N=1024 \, .$

We should also stress that \cite{Baker11} introduced, in the case of a very
stiff initial data, a filtering on the
dipole method. This interesting approach stabilizes the dipole method, thus
allowing for longer time integration.

\subsection{Curve-offset method}
Another approach to regularise the boundary integral consists in considering
that the vortices are located at a finite distance above the free surface
(e.g. \cite{Cao91,Tuck95,Tuck98,Scolan}).
This finite offset prevents any singularity in the integration as the
vortices are fictiously located at
\begin{equation*}
 (X_j,Y_j)=(x_j,y_j)+L_d \, \mb{n} \quad \text{with } L_d=\delta L/N \, .
\end{equation*}
While no singularity can occur with this technique, the kernel in the vortex
formulation takes the form
$\cot (\pi(x_{s,i}-x_{v,j})/L)$ where $(x_{s,i})$ are located on the
free surface whereas $(x_{v,i})$ are located a finite distance above the
free surface. It is finite but of order $\cot (\pi L_{d}/L)\, ,$ which becomes large
when $L_{d}$ is small.

This method is largely discussed in \cite[p. 928]{Scolan} where the author
introduces $L_d \, ,$ but also a second regularization which consists in replacing the
Green kernel $G(x_{s,i},x_{v,j})$, behaving as $\ln |x_{s,i}-x_{v,j}|$ (see
\eqref{Greenkernel}), with
$\ln |x_{s,i}-x_{v,j}|+b$ where $b$ is not necessarily small ($b\in [0; 10\,000]$).
The above $b$ term is difficult to justify from a mathematical point of view,
and not quite a small perturbation.

Instead of the above procedure, we consider a regularization inspired from
the vortex-blob method.
We thus introduce a second regularizing parameter $\varepsilon_{N}>0$ replacing the
previous kernel functions by $\cot (\pi(x_{s,i}-x_{v,j})/L+\varepsilon_{N}\, n)$.

In practice, when testing this approach, we have considered the $L_d$
regularization (finite distance of the vortices from the interface) and
a parameter $\varepsilon_{N}$ on the form of the blob method
(i.e. regularizing the kernel). Following \cite{Scolan},
we explicitely relate $L_d$ to $1/N \, .$ Also $\varepsilon_{N}$ is taken
to vanish as $1/N$ to try to assess the convergence properties of this approach.

We present in Fig.~\ref{fig:Offset} the numerical simulations performed with
the curve-offset method applied to the vortex formulation.
We used 
\eqref{eq CI simple wave}
as initial condition. This configuration corresponds to the wave breaking
and the results should be compared with
those of Fig.~\ref{fig wave breaking}.

We considered four different regularizations weights, namely
$L_d=L/N$ $\varepsilon_{N}=1/2N \, ,$
$L_d=2L/N$ $\varepsilon_{N}=1/2N \, ,$
$L_d=L/2N$ $\varepsilon_{N}=1/2N \, ,$
$L_d=L/N$ $\varepsilon_{N}=1/4N \, .$
We report the various numerical solutions at $t=3.68$ in
Fig.~\ref{fig:Offset}(a,b). We first note that all curves are significantly
different from the results presented in Fig.~\ref{fig wave breaking}.
The wave is much slower with the regularization. The unregularized
method, for which we observed convergence toward the Euler solution on
various test cases, has already reached splashing at that time.
The shape of the
obtained numerical solution also strongly depends on $L_d \, ,$ but only
weakly on $\varepsilon_{N} $ (the green and black curves being extremely
close to another). We performed further tests, which confirmed the weak influence of the
$\varepsilon_{N} $ regularization on the numerical solution.

A puzzling property of this approach is that, for a given choice of regularization,
say $L_d=L/N$ $\varepsilon_{N}=1/2N \, ,$ some convergence is achieved as $N$
is increased (see Fig.~\ref{fig:Offset}c). Yet the numerical curve then
converges toward a curve which could seem plausible, but significantly
differs from the unregularized solution. 
A final concern with this approach is that the total energy of the wave is
not conserved (and varies by $\sim 30\%$ through the simulation).

\begin{figure}
\centerline{(a)\includegraphics[width=0.97\textwidth]{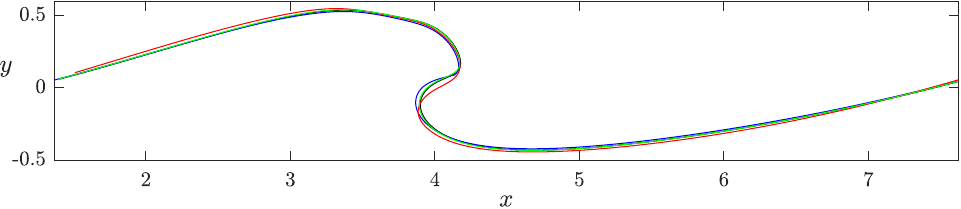}}
\vskip 3mm
\centerline{(b)\includegraphics[height=5.8cm]{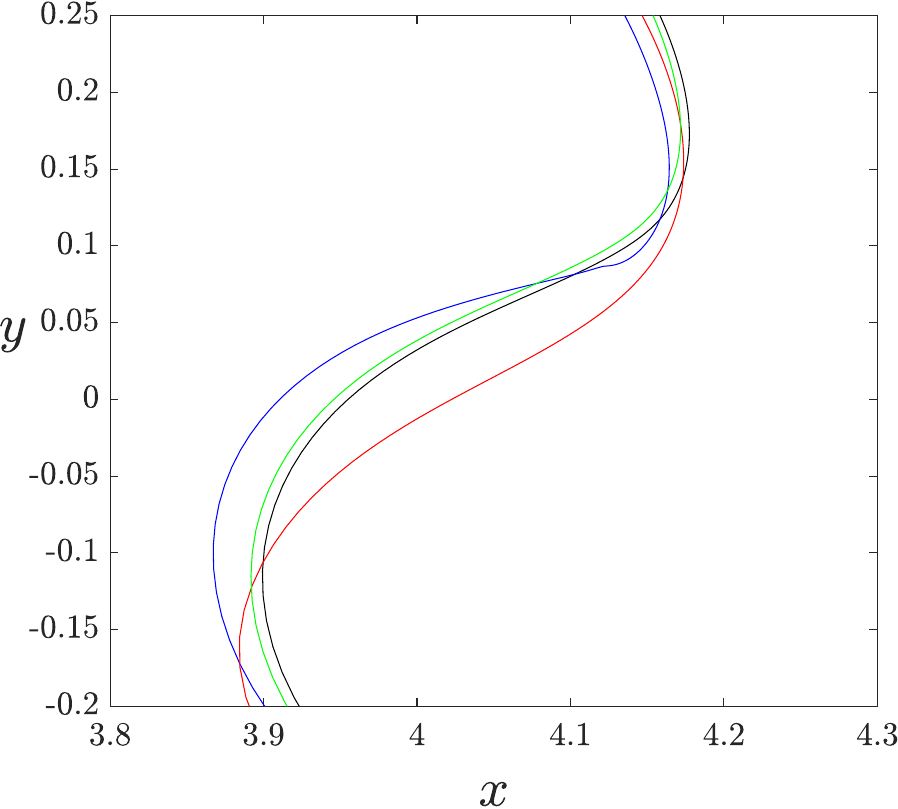} \ \
 (c)\includegraphics[height=5.68cm]{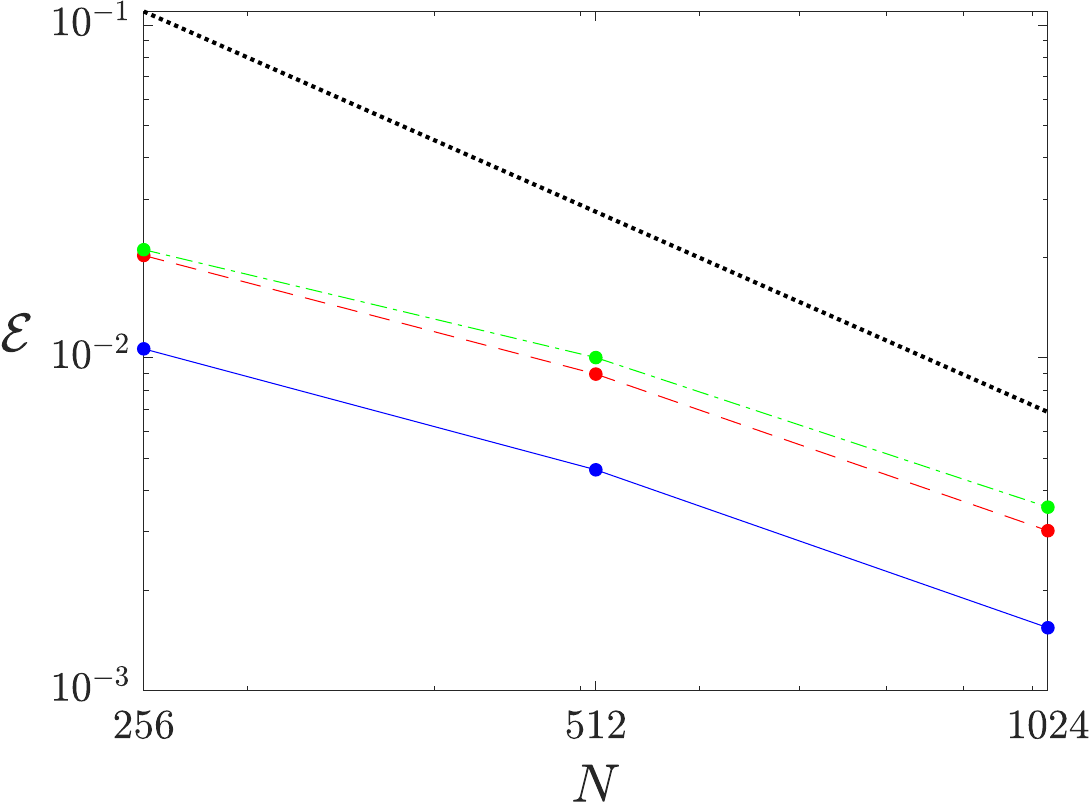}}
\caption{Simulations of wave breaking with initial condition
 \eqref{eq CI simple wave} using the curve-offset method. In graphs (a) and (b),
 profiles of the interface obtained with $N=512$ at time $t=3.68$ with $L_d=L/N\, ,$
 $\varepsilon_{N}=1/2N$
 (black), with $L_d=2L/N$ $\varepsilon_{N}=1/2N$ (blue), with $L_d=L/2N$
 $\varepsilon_{N}=1/2N$ (red), and with $L_d=L/N$ $\varepsilon_{N}=1/4N$
 (green).
(c) Evolution of the Hausdorff error, between
 N=2048 and N=256, 512, 1024, at time $t = 1.0$ (solid blue), $t = 2.0$ (dashed red) and time $t = 3.0$ (dot dashed green). The dotted line indicates an ideal $N^{-2}$ scaling.
}
\label{fig:Offset}
\end{figure} 

\section{Discussion}\label{sec7}

We derived a numerical strategy to discretize inviscid water waves in the
case of overturning interfaces (i.e. when the water-air interface is not a
graph). We showed that this discretisation can be used up to the splash
(i.e. when the interface self intersects). No filtering or regularisation
were introduced other than numerical discretisation. In the most severe
case of a splashing wave, an odd-even coupling was introduced.
It vanishes in the limit of a large number of points.

This formulation opens the way for further studies. In particular, we want
to study the possibility of a finite-time singularity formation at the tip
of a breaking wave. No singularity was observed with the initial condition
considered here. We also want to investigate the effect of an abrupt jump
in water height. Finally, the triple interface of water, air and land
(i.e. the sloping beach problem) still needs to be addressed.

\section*{Acknowledgements}
The authors are grateful to David Lannes and Alan Riquier for discussions.
This work was partially supported by the ANR project `SINGFLOWS'
(ANR-18-CE40-0027-01), the IMPT project `Ocean waves', and the CNES-Tosca project `Maeva'.

\appendix

\section{Cotangent kernel}\label{app-cot}

We begin this appendix by relating the cotangent kernel and the usual
kernel in $\R^2\,$.
Let $f$ a $L$-periodic real function and $z$ a curve verifying
$z(e+L)=z(e)+L\,$, then (we recall the notation $\widehat{\mb u}=u_{1}-\ic u_{2} \, ,$)
\begin{align*}
\widehat{K_{\R^2}}[|z_{e}|^{-1} f](x) 
&= \frac1{2\pi \ic } \int_{-\infty}^\infty \frac{1}{x-z(e')} f(e') \d e'
=\frac1{2\pi \ic } \int_{0}^{L} \sum_{k=-\infty}^{+\infty} \frac{1}{x-z(e') -L k} f(e') \d e'\\
&=
\frac{1}{2\pi L\ic } \int_{0}^{L} \sum_{k=-\infty}^{+\infty} \frac{1}{\frac{1}{L}(x-z(e')) - k} f(e') \d e'\\
&=
\frac1{2L\ic } \int_{0}^{L} \cot\Big(\frac{x-z(e')}{L/\pi}\Big) f(e') \d e'
\end{align*}
because it is well known for $x\in (0,1)$ that
\[
\pi \cot(\pi x)=\lim_{N\to \infty} \sum_{k=-N}^N \frac1{x+n} = \frac1x+ \lim_{N\to \infty} \sum_{k=1}^N \frac{2x}{x^2-n^2}
\]
which can be extended to $\C\setminus \Z$ by unicity of the analytic
extension. This computation can be found in several formal derivations to
get the Biot-Savart formula without using the Green kernel in $\T_{L}\times
\R\,$, see for instance \cite{Wilkening21}. 

The limit of Cauchy integrals from above provides
\begin{multline*}
\lim_{s\to 0^+} \frac1{2L\ic } \int_{0}^{L} \cot\Big(\frac{x-z(e')}{L/\pi}\Big) f(e') \d e'\Big\vert_{x= z(e)+s n } \\
=\frac{1}{2L\ic } {\rm pv}\int_{0}^{L} \cot\Big(\frac{z(e)-z(e')}{L/\pi}\Big) f(e') \d e' - \frac{1}{2}\frac{f(e)}{z_{e}(e)}
\end{multline*}
and from below we have
\begin{multline*}
\lim_{s\to 0^-} \frac1{2L \ic } \int_{0}^{L} \cot\Big(\frac{x-z(e')}{L/\pi}\Big) f(e') \d e'\Big\vert_{x= z(e) +s n } \\
=\frac{1}{2L\ic } {\rm pv}\int_{0}^{L} \cot\Big(\frac{z(e)-z(e')}{L/\pi}\Big) f(e') \d e' + \frac{1}{2}\frac{f(e)}{z_{e}(e)}\,.
\end{multline*}
From this formula, we recover that the tangential component
$\mb u\cdot \mb \tau=\Real (\hat {\mb u} \tau)$
has a jump, whereas the normal component $\mb u\cdot \mb n=-\Imag(\hat {\mb u} \tau)$ is continuous.

We can also note that $\int_{0}^{L} \Real \Big[\frac{z_{e}(e)}{2L\ic }\cot\Big(\frac{z(e)-z(e')}{L/\pi}\Big) \Big]f(e') \d e'$ is actually a classical integral whereas ${\rm pv}\int_{0}^{L} \Imag \Big[\frac{z_{e}(e)}{2L\ic }\cot\Big(\frac{z(e)-z(e')}{L/\pi}\Big) \Big]f(e') \d e'$ only makes sense in terms of principal value.

As usual concerning desingularization of the principal value, we first note that
\begin{align}
{\rm pv} \int_{0}^{L} & z_{e}(e') \cot\Big(\frac{z(e)-z(e')}{L/\pi}\Big) \d e'= \frac{L}{\pi} {\rm pv} \int_{-\infty}^\infty \frac{z_{e}(e')}{z(e)-z(e')} \d e' \label{desing}\\
=&-\frac{L}{\pi} \lim_{\varepsilon \to 0^+} \Bigg(\Bigg[ \ln \Big(z(e)-z(e') \Big) \Bigg]_{-\infty}^{e-\varepsilon} + \Bigg[ \ln \Big(z(e)-z(e') \Big) \Bigg]^{+\infty}_{e+\varepsilon} \Bigg)\nonumber\\
= & -\frac{L}{\pi} \lim_{\varepsilon \to 0^+} \lim_{S \to \infty} \Bigg(\Big( ( \ln (\varepsilon \rho)+i\theta) - \ln S \Big) + \Big( \ln S + i\pi - ( \ln (\varepsilon \rho)+i\theta + i \pi) \Big) \Bigg)=0\nonumber
\end{align}
where $z(e)-z(e \pm \varepsilon) = - \pm \varepsilon z_{e}(e) +
\mathcal{O}(\varepsilon ^2) = -\pm \varepsilon \rho e^{i\theta}+
\mathcal{O}(\varepsilon ^2)\,$, hence we can always write
\begin{equation*}
{\rm pv}\int \cot\Big(\frac{z(e)-z(e')}{L/\pi}\Big) f(e') de' = \int \cot\Big(\frac{z(e)-z(e')}{L/\pi}\Big) \frac{ f(e')z_{e}(e)-f(e)z_{e}(e') }{z_{e}(e)}de'
\end{equation*} 
which is now a classical integral of a continuous function.

For more details on singular integrals, we refer to \cite{musk,kellogg,fabes}, see \cite[Sect. 3]{ADL} for a brief summary.

\section{Discrete operators for the dipole formulation}\label{app-dipole}

For $z_{B}(i):=z_{B}(e_{B,i})\,$, $z_{S}(i):=z_{S}(e_{S,i})$ and
$\mu_{S}(i):=\mu_{S}(e_{S,i})$ given, we construct the matrix $A_{B,N}^*$ 
\begin{align*}
 &A_{B,N}^*(i,j)= \frac{L_{B}}{N_{B}} \Real \Bigg[\frac1{2L\ic} \cot\Big(\frac{z_{B}(i)-z_{B}(j) }{L/\pi}\Big) z_{B,e}(j) \Bigg]
 \ \forall i\neq j \in [1,N_{B}]\times [1,N_{B}]\,,\\
 &A_{B,N}^*(i,i)=\frac12- \frac{L_{B}}{N_{B}} \Real\Bigg[\frac1{4\pi\ic} \frac{z_{B,ee}(i)}{z_{B,e}(i)} \Bigg] \ \forall i\in [1,N_{S}]\,,
\end{align*}
and $F_{D,N}$ 
\[
 F_{D,N}(i)= -\sum_{j=1}^{N_{S}} \frac{L_{S}}{N_{S}} \mu_{S}(j) \Real \Bigg[ \frac{z_{S,e}(j)}{2L\ic} \cot\Big(\frac{z_{B}(i)-z_{S}(j) }{L/\pi}\Big)\Bigg]
 \ \forall i\in [1,N_{B}]\,.
 \]
We set $\mu_{B} = (A_{B,N}^*)^{-1} F_{D,N}\,$. This operation corresponds to \eqref{muB-muS}.

Next, we compute $\gamma_{B}=\partial_{e} \mu_{B}(e)\,$, $\gamma_{S}=\partial_{e} \mu_{S}(e)\,$, next $\partial_{t}z_{S}\,$, $\widehat{\nabla \phi_{F}} (z_{S}(e))$ and $\widehat{\nabla \phi_{A}} (z_{S}(e))$ where we could extend in the integral on $\Gamma_{S}$ for $e'=e$ by $\dfrac{\gamma_{S}(e)z_{S,ee}(e) -\gamma_{S,e}(e)z_{S,e}(e)}{2\pi \ic z_{S,e}^2(e)}\,$.

We compute $A_{S,N}^*$ as
\begin{align*}
 &A_{S,N}^*(i,j)= A_{tw}\frac{L_{S}}{N_{S}} \Real \Bigg[\frac1{2L\ic} \cot\Big(\frac{z_{S}(i)-z_{S}(j) }{L/\pi}\Big) z_{S,e}(j) \Bigg]
 \ \forall i\neq j \in [1,N_{S}]\times [1,N_{S}]\,,\\
 &A_{S,N}^*(i,i)=\frac12-A_{tw}\frac{L_{S}}{N_{S}} \sum_{j\neq i} \Real \Bigg[\frac1{2L\ic} \cot\Big(\frac{z_{S}(i)-z_{S}(j) }{L/\pi}\Big) z_{S,e}(j) \Bigg] \ \forall i\in [1,N_{S}]\,,
\end{align*}
next 
\[
 C_{D,N}(i,j)= A_{tw}\frac{L_{B}}{N_{B}} \Real \Bigg[\frac1{2L\ic} \cot\Big(\frac{z_{S}(i)-z_{B}(j) }{L/\pi}\Big) z_{B,e}(j) \Bigg]
 \ \forall (i, j) \in [1,N_{S}]\times [1,N_{B}]\,.
\]
and finally 
\begin{align*}
 D_{D,N}(i,j)= \frac{L_{S}}{N_{S}} \Real \Bigg[\frac1{2L\ic} \cot\Big(\frac{ z_{B}(i)-z_{S}(j) }{L/\pi}\Big) z_{S,e}(j) \Bigg]
 \ \forall (i,j) \in [1,N_{B}]\times [1,N_{S}]\,.
\end{align*}

Concerning the right hand side term, we compute
\begin{align}
& G_{D,1,N}(i)=\nonumber\\ 
&-A_{tw} \sum_{j\neq i} \frac{L_{S}}{N_{S}} ( \mu_{S}(i) - \mu_{S}(j))\Real \Bigg[\frac{\pi}{2L^2\ic} \sin^{-2}\Big(\frac{z_{S}(i)-z_{S}(j) }{L/\pi}\Big) (\partial_{t}z_{S}(i)-\partial_{t}z_{S}(j) )z_{S,e}(j) \Bigg] \nonumber\\
&-A_{tw} \sum_{j\neq i} \frac{L_{S}}{N_{S}} ( \mu_{S}(j) - \mu_{S}(i))\Real \Bigg[\frac{1}{2L\ic} \cot\Big(\frac{z_{S}(i)-z_{S}(j) }{L/\pi}\Big) \partial_{t}z_{S,e}(j) \Bigg] \nonumber\\
&+A_{tw} \sum_{j=1}^{N_{B}} \frac{L_{B}}{N_{B}} \mu_{B}(j)\Real \Bigg[\frac{\pi}{2L^2\ic} \sin^{-2}\Big(\frac{z_{S}(i)-z_{B}(j) }{L/\pi}\Big)\partial_{t}z_{S}(i) z_{B,e}(j) \Bigg] \label{GD1} \\
& + \frac12\Real \Bigg[ \partial_{t} z_{S}(i) \Big((A_{tw}+1)\widehat{\nabla \phi_{F}} (z_{S}(i))+(A_{tw}-1) \widehat{\nabla \phi_{A}} (z_{S}(i))\Big)\Bigg]\nonumber\\
&-\frac{1}{4} \Big((A_{tw}+1) \vert (\widehat{\nabla \phi_{F}}+ \widehat{\mb u_{\gamma}} )(z_{S}(i)) \vert^2+ (A_{tw}-1) \vert \widehat{\nabla \phi_{A}} (z_{S}(i)) \vert^2 \Big)\nonumber\\
& -\frac{(A_{tw}+1)\sigma }{2\rho_{F}}\kappa(z_{S}(i)) -gA_{tw} \Imag z_{S}(i)\,,\nonumber
\end{align}
for all $i\in [1, N_{S}]\,$, whereas 
\begin{align*}
 G_{D,2,N}(i)=
&- \sum_{j=1}^{N_{S}} \frac{L_{S}}{N_{S}} \mu_{S}(j) \Real \Bigg[\frac{1}{2L\ic} \cot\Big(\frac{ z_{B}(i)-z_{S}(j) }{L/\pi}\Big) \partial_{t}z_{S,e}(j) \Bigg] \\
&- \sum_{j=1}^{N_{S}} \frac{L_{S}}{N_{S}} \mu_{S}(j)\Real \Bigg[\frac{\pi}{2L^2\ic} \sin^{-2}\Big(\frac{ z_{B}(i)-z_{S}(j) }{L/\pi}\Big)\partial_{t}z_{S}(j) z_{S,e}(j) \Bigg]
\end{align*}
for all $i\in [1, N_{B}]\,$.

\begin{remark}
It could seem strange that the diagonal terms in $A_{S,N}^*$ are of a
different nature than those in $A_{B,N}^*$.
It is in fact the same, because we can make use of Appendix~\ref{app-cot}
to rewrite \eqref{muB-muS} in the form
\begin{multline*}
\frac12 \mu_{B}(e) + \int_{0}^{L_{B}} (\mu_{B}(e')-\mu_{B}(e))\Real \Bigg[\frac1{2L\ic} \cot\Big(\frac{z_{B}(e)-z_{B}(e') }{L/\pi}\Big) z_{B,e}(e') \Bigg]\d e'\\
=-\int_{0}^{L_{S}} \mu_{S}(e') \Real \Bigg[\frac1{2L\ic} \cot\Big(\frac{z_{B}(e)-z_{S}(e') }{L/\pi}\Big) z_{S,e}(e') \Bigg]\d e' \,,
\end{multline*}
for which we would defined
\[
A_{B,N}^*(i,i)=\frac12-\frac{L_{B}}{N_{B}} \sum_{j\neq i} \Real \Bigg[\frac1{2L\ic} \cot\Big(\frac{z_{B}(i)-z_{B}(j) }{L/\pi}\Big) z_{B,e}(j) \Bigg] \ \forall i\in [1,N_{S}]
\]
and we recover the same expression by the discretization of desingularization rule \eqref{desing}
\[
\frac{L_{B}}{N_{B}} \sum_{j\neq i} \Real \Bigg[\frac1{2L\ic} \cot\Big(\frac{z_{B}(i)-z_{B}(j) }{L/\pi}\Big) z_{B,e}(j) \Bigg] - \frac{L_{B}}{N_{B}} \Real\Bigg[\frac1{4\pi\ic} \frac{z_{B,ee}(i)}{z_{B,e}(i)} \Bigg] =0 \,.
\]

We can therefore use either of these formulations, but it is more
interesting to avoid the second derivative $z_{S,ee} \, ,$ which tends to destabilize
the numerical code when the curvature of the interface becomes large. As
the bottom boundary does
not depend on time, we can retain the expression in terms of $z_{B,ee}\,$. 

This relation could be also used in the extension for $\partial_{t}z_{S}$ mentioned above replacing
$\dfrac{L_{S}}{N_{S}}\dfrac{\gamma_{S}(i)z_{S,ee}(i) -\gamma_{S,e}(i)z_{S,e}(i)}{2\pi \ic z_{S,e}^2(i)}$
with
\[
\dfrac{L_{S}}{N_{S}}\dfrac{\gamma_{S}(i) }{ z_{S,e}(i)} \sum_{j\neq i} \frac1{2L\ic} \cot\Big(\frac{z_{S}(i)-z_{S}(j) }{L/\pi}\Big) z_{S,e}(j) - \dfrac{L_{S}}{N_{S}}\dfrac{\gamma_{S,e}(i)}{2\pi \ic z_{S,e}(i)}\,.
\]
Unfortunately, this does not improve the stability of the code. The method
explained in \S~\ref{sec-shift} with the shifted grids in space allows us
to avoid $\gamma_{S,e}\, ,$ which would appear by the extension by continuity. 
\end{remark}

\section{Discrete operators for the vortex formulation}\label{app-vortex}

First, we give the precise equation for the vortex formulation, then we give the discret version of the operators.

We compute
\begin{align*}
 &\frac12\partial_{t} \Psi_{S}(e)=
 \int_{0}^{L_{S}} \Real \Bigg[\frac{\partial_{t}\gamma_{S}(e')z_{S,e}(e)-\partial_{t}\gamma_{S}(e)z_{S,e}(e')}{2L\ic} \cot\Big(\frac{z_{S}(e)-z_{S}(e') }{L/\pi}\Big) \Bigg]\d e'\\
&+\int_{0}^{L_{B}} \partial_{t}\gamma_{B}(e')\Real \Bigg[\frac1{2L\ic} \cot\Big(\frac{z_{S}(e)-z_{B}(e') }{L/\pi}\Big) z_{S,e}(e) \Bigg]\d e'\\
&-\int_{0}^{L_{S}} \Real \Bigg[\pi\frac{\gamma_{S}(e')z_{S,e}(e)-\gamma_{S}(e)z_{S,e}(e')}{2L^2\ic} \sin^{-2}\Big(\frac{z_{S}(e)-z_{S}(e') }{L/\pi}\Big) (\partial_{t}z_{S}(e)-\partial_{t}z_{S}(e') ) \Bigg]\d e'\\
&+\int_{0}^{L_{S}} \Real \Bigg[\frac{\gamma_{S}(e')\partial_{t}z_{S,e}(e)-\gamma_{S}(e)\partial_{t}z_{S,e}(e')}{2L\ic} \cot\Big(\frac{z_{S}(e)-z_{S}(e') }{L/\pi}\Big) \Bigg]\d e'\\
&+\int_{0}^{L_{B}} \gamma_{B}(e')\Real \Bigg[\frac1{2L\ic} \cot\Big(\frac{z_{S}(e)-z_{B}(e') }{L/\pi}\Big) \partial_{t}z_{S,e}(e) \Bigg]\d e'\\
&-\int_{0}^{L_{B}} \gamma_{B}(e')\Real \Bigg[\frac{\pi}{2L^2\ic} \sin^{-2}\Big(\frac{z_{S}(e)-z_{B}(e') }{L/\pi}\Big)\partial_{t}z_{S}(e) z_{S,e}(e) \Bigg]\d e'\\
&+ \sum_{j=1}^{N_{v}} \gamma_{v,j} \Real\Bigg[\frac{\partial_{t}z_{S,e}(e)}{2L\ic } \cot\Big(\frac{z_{S}(e)-z_{v,j} }{L/\pi}\Big)-\frac{\pi z_{S,e}(e)\partial_{t}z_{S,e}(e)}{2L^2\ic } \sin^{-2}\Big(\frac{z_{S}(e)-z_{v,j} }{L/\pi}\Big)\Bigg] \\
&+\frac{ \omega_{0}}{4\pi}\int_{0}^{L_{S}} \ln\Big(\cosh \Imag \frac{z_{S}(e)-z_{S}(e')}{L/(2\pi)} -\cos \Real \frac{z_{S}(e)-z_{S}(e')}{L/(2\pi)} \Big) \\
& \hspace{7cm}\times\Real\Big[ \partial_{t} z_{S,e}(e)\overline{z_{S,e}(e')}+ z_{S,e}(e)\overline{\partial_{t} z_{S,e}(e')}\Big] \d e'\\
&- \omega_{0}\int_{0}^{L_{S}} \Imag \Bigg[ \frac{\partial_{t}z_{S}(e)-\partial_{t}z_{S}(e')}{2L\ic} \cot \Big( \frac{z_{S}(e)-z_{S}(e')}{L/\pi}\Big) \Bigg] \Real\Big[ z_{S,e}(e)\overline{z_{S,e}(e')}\Big] \d e'\\
&-\frac{ \omega_{0}}{4\pi}\int_{0}^{L_{B}} \ln\Big(\cosh \Imag \frac{z_{S}(e)-z_{B}(e')}{L/(2\pi)} -\cos \Real \frac{z_{S}(e)-z_{B}(e')}{L/(2\pi)} \Big) \Real\Big[\partial_{t}z_{S,e}(e)\overline{z_{B,e}(e')}\Big] \d e'\\
&+ \omega_{0}\int_{0}^{L_{B}}\Imag \Bigg[ \frac{\partial_{t}z_{S}(e)}{2L\ic} \cot \Big( \frac{z_{S}(e)-z_{B}(e')}{L/\pi}\Big) \Bigg] \Real\Big[z_{S,e}(e)\overline{z_{B,e}(e')}\Big] \d e'
\end{align*}
where we have used the relation on the cotangent \eqref{cot-nabla}.

Every integrals are classically defined, in particular the functions can be extended by continuity for $e=e'$ in the third integral by
\[
 \Real \Bigg[\frac{\gamma_{S}(e)z_{S,ee}(e)-\gamma_{S,e}(e)z_{S,e}(e)}{2\pi\ic z_{S,e}^2(e)} \partial_{t}z_{S,e}(e) \Bigg]
\]
and in the fourth one by
\[
\Real \Bigg[\frac{\gamma_{S}(e)\partial_{t}z_{S,ee}(e)-\gamma_{S,e}(e)\partial_{t}z_{S,e}(e)}{2\pi\ic z_{S,e}(e)} \Bigg]\,,
\]
which can be simplified by a part in the extension of the third integral.
We can replace the term $z_{S,ee}$ in the extension of the second integral by replacing
$
 \frac{L_{S}}{N_{S}}\Real \Bigg[\dfrac{\gamma_{S}(i)z_{S,ee}(i)}{2\pi\ic z_{S,i}^2(e)} \partial_{t}z_{S,e}(i) \Bigg]
$
by
\[
 \frac{L_{S}}{N_{S}} \sum_{j\neq i} \Real \Bigg[\dfrac{\gamma_{S}(i)}{z_{S,i}^2} \partial_{t}z_{S,e}(i)\frac1{2L\ic} \cot\Big(\frac{z_{S}(i)-z_{S}(j) }{L/\pi}\Big) z_{S,e}(j) \Bigg]\,.
\]
Unfortunately, it is more complicated to replace $ \partial_{t}z_{S,e}$ and
$\partial_{t}z_{S,ee}$ because differentiating the previous relation would
introduce an additional $\sin^{-2}$ term.
The first integral has to be replaced by
\[
 \int_{0}^{L_{S}} \Real
 \Bigg[\frac{\partial_{t}\gamma_{S}(e')z_{S,e}(e)}{2L\ic}
 \cot\Big(\frac{z_{S}(e)-z_{S}(e') }{L/\pi}\Big) \Bigg]\d e' \, ,
\]
where the continuous function is extended for $e=e'$ by zero.

These computations provide the explicit expression for $G_{V,1}\,$.
We can note that the many terms disappear when considering the
single-fluids formulation, i.e. the $\alpha=1$ and $A_{tw}=
1\,$ case. 

For the expression of $G_{V,2}\,$, we get
\begin{align*}
 &\hspace{-10pt}G_{V,2}(e) \\=& 
 -\int_{0}^{L_{S}} \gamma_{S}(e') \Imag \Bigg[\frac{\pi z_{B,e}(e) \partial_{t}z_{S}(e')}{2L^2\ic} \sin^{-2}\Big(\frac{z_{B}(e)-z_{S}(e') }{L/\pi}\Big)\Bigg] \d e'\\
&- \sum_{j=1}^{N_{v}}\gamma_{v,j} \Imag \Bigg[\pi\frac{z_{B,e}(e) \partial_{t}z_{v,j}}{2L^2\ic} \sin^{-2}\Big(\frac{z_{B}(e)-z_{v,j} }{L/\pi}\Big) \Bigg] \\
&-\frac{ \omega_{0}}{4\pi}\int_{0}^{L_{S}} \ln\Big(\cosh \Imag \frac{z_{B}(e)-z_{S}(e')}{L/(2\pi)} -\cos \Real \frac{z_{B}(e)-z_{S}(e')}{L/(2\pi)} \Big) \Imag \Bigg[z_{B,e}(e)\overline{\partial_{t }z_{S,e}(e')}\Bigg] \d e'\\
&- \omega_{0}\int_{0}^{L_{S}}\Imag \Bigg[ \frac{\partial_{t}z_{S}(e')}{2L\ic} \cot \Big( \frac{z_{B}(e)-z_{S}(e')}{L/\pi}\Big) \Bigg] \Imag \Bigg[z_{B,e}(e)\overline{\partial_{t }z_{S,e}(e')}\Bigg] \d e' \,.
\end{align*}

Concerning the numerical approximation, for $z_{B}(i):=z_{B}(e_{B,i})\,$, $z_{S}(i):=z_{S}(e_{S,i})$ and $\gamma_{S}(i):=\gamma_{S}(e_{S,i})$ given, we set $\tilde z_{B}(i):=(z_{B}(e_{B,i})+z_{B}(e_{B,i+1}))/2$ for $i=1,\dots, N_{B}-1$ to construct the matrix $B_{B,N}$
\begin{align*}
 &B_{B,N}(i,j)= \frac{L_{B}}{N_{B}} \Imag \Bigg[ \frac{\tilde z_{B,e}(i)}{2L\ic} \cot\Big(\frac{\tilde z_{B}(i)-z_{B}(j) }{L/\pi}\Big)\Bigg]
 \ \forall (i,j)\in [1,N_{B}-1]\times [1,N_{B}]\,,\\
 &B_{B,N}(N_{B},j)=\frac{L_{B}}{N_{B}} \ \forall j\in [1,N_{B}]\,,
\end{align*}
The discretization of ${\rm RHS}_{V0,B,N}$ and ${\rm RHS}_{VB,N}$ are
clear, replacing every $z_{B}(e)$ by $\tilde z_{B}(i)\,$, and where the
last component is $-\gamma\,$. We deduce $\gamma_{B} =
B_{B,N}^{-1} {\rm RHS}_{VB,N}\,$. 

Next, we compute $\partial_{t}z_{S}\,$, $\widehat{u_{F}} (z_{S}(e))$ and $\widehat{u_{A}} (z_{S}(e))$ where we could extend in the integral on $\Gamma_{S}$ for $e'=e$ by $\dfrac{\gamma_{S}(e)z_{S,ee}(e) -\gamma_{S,e}(e)z_{S,e}(e)}{2\pi \ic z_{S,e}^2(e)}\,$, and we compute the derivative with respect to $e\,$.

We compute $A_{S,N}$ as
\begin{align*}
 A_{S,N}(i,j)=& A_{tw}\frac{L_{S}}{N_{S}} \Real \Bigg[\frac1{2L\ic} \cot\Big(\frac{z_{S}(i)-z_{S}(j) }{L/\pi}\Big) z_{S,e}(i) \Bigg]
 \ \forall i\neq j \in [1,N_{S}]\times [1,N_{S}]\,,\\
 A_{S,N}(i,i)=&\frac12 \ \forall i\in [1,N_{S}]\,,
\end{align*}
next 
\[
 C_{V,N}(i,j)= A_{tw}\frac{L_{B}}{N_{B}} \Real \Bigg[\frac1{2L\ic} \cot\Big(\frac{z_{S}(i)-z_{B}(j) }{L/\pi}\Big) z_{S,e}(i) \Bigg]
 \ \forall (i, j) \in [1,N_{S}]\times [1,N_{B}]\,.
\]
and finally 
\begin{align*}
 &D_{V,N}(i,j)= \frac{L_{S}}{N_{S}} \Imag \Bigg[\frac1{2L\ic} \cot\Big(\frac{\tilde z_{B}(i)-z_{S}(j) }{L/\pi}\Big) \tilde z_{B,e}(i) \Bigg]
 \ \forall (i,j) \in [1,N_{B}-1]\times [1,N_{S}]\,,\\
 &D_{V,N}(N_{B},j)=0 \ \forall j\in [1,N_{S}]\,.
\end{align*}

\section{Notations}

We summarize here the notations used in this article.

\begin{itemize}
\item For a given vector $\mb a_R =(a_{R,1},a_{R,2})$, $a_{R,1}$ is the
  first component, whereas $a_{R,2}$ is the second component. We also
  introduce the complex notation: $a_R=a_{R,1}+\ic a_{R,2} \, .$

\item For a function $e\mapsto f_R(e)$, we denote the derivative with respect to $e$ as $f_{R,e}$.

 \item We define three domains: $\mathcal{D}= \T_{L} \times\R$, $\mathcal{D}_{F}$ the fluid domain, $\mathcal{D}_{A}$ the air domain,
$\mathcal{D}_{B}$ the domain below the bottom. $\Gamma_{S}$ is then the
   water-air free surface, and $\Gamma_{B}$ the bottom (see Fig.~\ref{Geometry} in Section~\ref{sect-intro}).

\item Parametrization of the boundaries: the free surface $\Gamma_{S}$
  (initially parameterised by arclength from
left to right) $e\in [0,L_{S}] \mapsto z_{S}(t,e)=z_{S,1}(e,e)+ \ic z_{S,2}(t,e) \,$; the bottom $\Gamma_{B}$ (by arclength from
left to right) $e\in [0,L_{B}] \mapsto z_{B}(e)=z_{B,1}(e)+ \ic
z_{B,2}(e)\, .$ 

\item Tangent vectors $\tau_{S} = \tau_{S,1} +\ic \tau_{S,2}=
  |z_{S,e}(e)|^{-1} z_{S,e}$ and $\tau_{B} = \tau_{B,1} +\ic \tau_{B,2}=
  |z_{B,e}(e)|^{-1} z_{B,e}$
  are pointing to the right. The normal vector $n_{S}=n_{S,1} +\ic n_{S,2}= -\tau_{S,2} +\ic \tau_{S,1}=\ic \tau_{S}$ is pointing out of the fluid domain, whereas $n_{B}=n_{B,1} +\ic n_{B,2}= -\tau_{B,2} +\ic \tau_{B,1}=\ic \tau_{B}$ is pointing in the fluid domain.

\item The jumps of the tangential components $(\gamma_{S},\gamma_{B})$ are
  defined in \eqref{def-gamma}, whereas the jumps of the potentials
  $(\mu_{S},\mu_{B})$ are defined in \eqref{def-mu}. 

\item The mean current (or circulation) is $\gamma\in \R$ and the constant vorticity is $\omega_{0}\in \R$. The non constant part of the vorticity inside the fluid is $\sum_{j=1}^{N_{v}} \gamma_{v,j} \delta_{z_{v,j}(t)} \,$.

\item $\psi_{F}$ is the stream function associated to $\mb u=\nabla^\perp \psi_{F}\,$.

\item $u_{\omega,\gamma}$ is defined as a vector satisfying
  \eqref{u-fgamma-1} and \eqref{u-fgamma-2}, which allows to define $\mb
  u_{R}:=\mb u-\mb u_{\omega,\gamma}$ as the gradient of a potential function
  (see \eqref{uR-pot}). 

\item For vector fields $\mb u = ( u_{1}, u_{2})$, we define
$\widehat{\mb u}=u_{1}-\ic u_{2} \, ,$
$(u_{1},u_{2})^\perp =(-u_{2},u_{1}) \, ,$
as well as the curl operator
$\curl \mb u = \partial_{1} u_{2}-\partial_{2} u_{1}\, .$ 
\end{itemize}

%%%%%%%%%%%%% Bibliography %%%%%%%%%%%%%%%%%%%%

%\bibliographystyle{abbrv}
%\bibliography{Biblio} 

\end{document}